\documentclass[a4paper,11pt]{dstmathbook}
\usepackage[utf8]{inputenc}
\usepackage[T1]{fontenc}
\usepackage[ngerman,french,main=english]{babel}
\usepackage{csquotes}
\usepackage[style=alphabetic, maxcitenames=6, maxalphanames=6,maxbibnames=6, giveninits=true, backref=true,noerroretextools,mincrossrefs=1]{biblatex}
\usepackage{chngcntr}
\usepackage{amsthm}
\usepackage{amssymb}
\usepackage[table]{xcolor}
\usepackage{tikz}
\usepackage{mathtools}
\usepackage{mathshortcuts}
\usepackage{enumitem}
\usepackage{hyperref}
\usepackage{cleveref}
\usepackage{autonum}
\usepackage{mleftright}
\usepackage{todonotes}
\usepackage{braket}
\usepackage{setspace}
\usepackage[labelformat=simple,labelfont=normalfont]{subcaption}
\usepackage{cellspace}
\usepackage{cancel}
\usepackage{lipsum}
\usepackage{tqft}

\hypersetup{hidelinks}

\definecolor{gray}{rgb}{0.95,0.95,0.95}
\definecolor{niceblue}{rgb}{0.30,0.33,1}
\DeclareFieldFormat{extraalpha}{#1}

\DeclareLabelalphaTemplate{
  \labelelement{
    \field[final]{shorthand}
    \field{label}
    \field[strwidth=3,strside=left,ifnames=1]{labelname}
    \field[strwidth=1,strside=left]{labelname}
  }
}

\DeclareBibliographyDriver{incollection}{%
  \usebibmacro{bibindex}%
  \usebibmacro{begentry}%
  \usebibmacro{author/translator+others}%
  \setunit{\printdelim{nametitledelim}}\newblock
  \usebibmacro{title}%
  \newunit
  \printlist{language}%
  \newunit\newblock
  \usebibmacro{byauthor}%
  \newunit\newblock
  \usebibmacro{in:}%
  \iffieldundef{crossref}
    {\usebibmacro{crossref:full}}
    {\usebibmacro{crossref:label}}
  \usebibmacro{chapter+pages}%
  \newunit\newblock
  \usebibmacro{doi+eprint+url}%
  \setunit{\bibpagerefpunct}\newblock
  \usebibmacro{pageref}%
  \newunit\newblock
  \iftoggle{bbx:related}
    {\usebibmacro{related:init}%
     \usebibmacro{related}}
    {}%
  \usebibmacro{finentry}}

\newbibmacro{crossref:full}{%
  \usebibmacro{maintitle+booktitle}%
  \newunit\newblock
  \usebibmacro{byeditor+others}%
  \newunit\newblock
  \iffieldundef{maintitle}
    {\printfield{volume}%
     \printfield{part}}
    {}%
  \newunit
  \printfield{volumes}%
  \newunit\newblock
  \usebibmacro{series+number}%
  \newunit\newblock
  \printfield{note}%
  \newunit\newblock
  \usebibmacro{publisher+location+date}%
  \newunit\newblock
}

\newbibmacro{crossref:label}{%
  \entrydata{\strfield{crossref}}
     {\printtext{\mkbibbrackets
        {\printfield{labelalpha}\printfield{extraalpha}}}}}

\newbibmacro{crossref:extrainfo}{%
  \newunit\newblock
  \iftoggle{bbx:isbn}
    {\printfield{isbn}}
    {}%
  \newunit\newblock
  \usebibmacro{doi+eprint+url}%
  \newunit\newblock
  \usebibmacro{addendum+pubstate}}%

\AtEveryBibitem{
  \clearlist{language}
}

\renewbibmacro{in:}{%
  \ifentrytype{article}{}{\printtext{\bibstring{in}\intitlepunct}}}

\renewcommand*{\bibpagerefpunct}{\addperiod\addspace}

\renewbibmacro*{pageref}{%
  \iflistundef{pageref}
    {}
    {\printtext{$\hookrightarrow$\ppspace%
       \printlist[pageref][-\value{listtotal}]{pageref}}}}

\numberwithin{equation}{chapter}
\counterwithout{figure}{chapter}

\mleftright

\DeclareRobustCommand{\gobblefive}[5]{}

\definecolor{vlightgray}{gray}{0.95}

\usetikzlibrary{babel,cd,decorations,decorations.pathmorphing,decorations.pathreplacing}

\theoremstyle{plain}
\newtheorem{thm}{Theorem}[chapter]
\newtheorem{lem}[thm]{Lemma}
\newtheorem{prop}[thm]{Proposition}
\newtheorem{cor}[thm]{Corollary}
\theoremstyle{definition}
\newtheorem{defn}[thm]{Definition}
\newtheorem{exmp}[thm]{Example}
\newtheorem{rem}[thm]{Remark}

\crefname{defn}{Definition}{Definitions}
\Crefname{defn}{Definition}{Definitions}
\crefname{thm}{Theorem}{Theorems}
\Crefname{thm}{Theorem}{Theorems}
\crefname{lem}{Lemma}{Lemmas}
\Crefname{lem}{Lemma}{Lemma}
\crefname{rem}{Remark}{Remarks}
\Crefname{rem}{Remark}{Remarks}
\crefname{prop}{Proposition}{Propositions}
\Crefname{prop}{Proposition}{Propositions}
\crefname{cor}{Corollary}{Corollaries}
\Crefname{cor}{Corollary}{Corollaries}
\crefname{exmp}{Example}{Examples}
\Crefname{exmp}{Example}{Examples}
\crefname{chapter}{Chapter}{Chapters}
\Crefname{chapter}{Chapter}{Chapters}
\crefname{section}{Section}{Sections}
\Crefname{section}{Section}{Sections}
\crefname{part}{Part}{Parts}
\Crefname{part}{Part}{Parts}
\crefname{figure}{Figure}{Figures}
\Crefname{figure}{Figure}{Figures}
\crefname{equation}{}{}
\Crefname{equation}{}{}

\crefname{genenumi}{}{}
\Crefname{genenumi}{}{}
\creflabelformat{genenumi}{{#2\normalfont#1#3}}

\labelcrefformat{subequation}{#2(#1)#3}
\labelcrefrangeformat{subequation}{#3(#1)#4 to #5(#2)#6}

\makeatletter
\@ifpackageloaded{subcaption}{%
\renewcommand*\subcaption@@label[2]{%
  \@bsphack\begingroup
    \subcaption@ORI@label#1{#2}%
    \let\SK@\@gobbletwo
    \protected@edef\@currentlabel{\csname thesub\@captype\endcsname}%
    \protected@edef\cref@currentlabel{%
      [subs\@captype][\arabic{sub\@captype}][\cref@result]%
      \csname thesub\@captype\endcsname}%
    \subcaption@ORI@label#1{sub@#2}%
  \endgroup\@esphack}%
}
\makeatother

\newlist{genenum}{enumerate}{1}
\setlist[genenum]{label=(\arabic*),leftmargin=*,labelindent=0.5em,font=\normalfont}

\newlist{textenum}{enumerate}{1}
\setlist[textenum]{label=(\alph*),font=\normalfont,wide}

\let\epsilon\varepsilon

\newcommand{\nospacepunct}[1]{\makebox[0pt][l]{#1}}

\newcommand{\Mod}[1]{\:(\mathrm{mod}\:#1)}

\DeclareMathOperator{\obj}{obj}
\DeclareMathOperator{\dom}{dom}
\DeclareMathOperator{\cod}{cod}
\DeclareMathOperator{\chom}{hom}
\DeclareMathOperator{\autom}{Aut}

\DeclareMathOperator{\rot}{rot}
\DeclareMathOperator{\down}{down}
\DeclareMathOperator{\lan}{Lan}

\renewcommand{\id}{\mathrm{id}}
\newcommand{\Id}{\mathrm{Id}}
\newcommand{\idmtx}{\mathbf{1}}
\newcommand{\coev}{\mathrm{coev}}
\newcommand{\ev}{\mathrm{ev}}

\DeclareMathOperator{\diff}{Diff}
\DeclareMathOperator{\Ceil}{ceil}
\newcommand\sircle{\mathbb{S}}
\newcommand{\binfor}{\mathrm{BinFor}}
\newcommand{\annfor}{\mathrm{AnnFor}}

\newcommand{\shfrak}{\mathfrak{h}}

\DeclarePairedDelimiter{\ceil}{\lceil}{\rceil}
\DeclarePairedDelimiter{\floor}{\lfloor}{\rfloor}

\title{Thompson Field Theory}
\author{Deniz Ekrem Stiegemann}

\begin{document}



\begin{titlepage}

\begin{center}
  \vspace*{1cm}
  {\LARGE\bfseries
  Thompson Field Theory
  \par}
  \vspace*{10cm}
  \selectlanguage{ngerman}
  {\setstretch{1.4}
  Von der Fakultät für Mathematik und Physik\\
  der Gottfried Wilhelm Leibniz Universität Hannover\\
  \vspace{4mm}
  zur Erlangung des akademischen Grades\\
  Doktor der Naturwissenschaften\\
  Dr.\ rer.\ nat.\\
  \vspace{4mm}
  genehmigte Dissertation von\par
  }
  \vspace*{1cm}
  {\LARGE
  Deniz Ekrem Stiegemann, M.\,Sc.
  \par}
  \vspace*{2cm}
  {
  2019
  }
\end{center}
\end{titlepage}

\setcounter{page}{0}

\thispagestyle{empty}

\noindent \textbf{Mitglieder der Prüfungskommission:}\vspace*{0.5\baselineskip}\\
Prof.\ Dr.\ Elmar Schrohe (Vorsitzender)\\
Prof.\ Dr.\ Tobias J. Osborne (Betreuer)\\
Prof.\ Dr.\ Reinhard F. Werner

\vspace*{1.5\baselineskip}

\noindent \textbf{Gutachter:}\vspace*{0.5\baselineskip}\\
Prof.\ Dr.\ Robert König\\
Prof.\ Dr.\ Tobias J. Osborne \\
Prof.\ Dr.\ Reinhard F. Werner

\vspace*{1.5\baselineskip}

\noindent Tag der Promotion: 11.\ Juli 2019

\cleardoublepage

\selectlanguage{english}


\thispagestyle{empty}
\begin{center}
  {\Large
  {\bfseries Thompson Field Theory\par}
  \vspace*{0.5\baselineskip}
  Deniz E.\ Stiegemann\par
  \vspace*{\baselineskip}
  {\itshape Dissertation\par}
  }
  \vspace*{1.5\baselineskip}
\end{center}

  \noindent{\bfseries Abstract.}
  We introduce Thompson field theory, a class of toy models of conformal field theory in which Thompson's group $T$ takes the role of a discrete analogue of the chiral conformal group. $T$ and the related group $F$ are discrete transformations of dyadic partitions of the circle and the unit interval, respectively. When vectors or tensors are associated with partitions, one can construct a direct limit Hilbert space, here called the semicontinuous limit, and $F$ and $T$ have unitary representations on this space. We give an abstract description of these representations following the work of Jones. We also show that $T$ can be thought of as acting on the boundary of an equal-time Poincaré disk in $\mathrm{AdS}_3$. This defines a representation of $T$ on the Hilbert space that contains all tree-like holographic states, as introduced by \citeauthor{PastawskiYoshidaHarlowPreskill2015}. It also establishes a bulk-boundary correspondence through Imbert's isomorphism between $T$ and Penner's Ptolemy group. We further propose definitions of field operators and correlation functions for the discrete theory.
  Finally, we sketch new developments like particle creation and annihilation, as well as black holes and possible connections with topological quantum field theory.
  \\[0.8\baselineskip]

\vspace*{1.5\baselineskip}

\selectlanguage{ngerman}
  \noindent{\bfseries Zusammenfassung.}
  Wir stellen Thompson-Feldtheorie vor, eine Klasse von Toy-Modellen für konforme Feldtheorie, in denen Thompsons Gruppe $T$ die Rolle eines diskreten Analogons der chiralen konformen Gruppe übernimmt. $T$ und die verwandte Gruppe $F$ wirken wie diskrete Transformationen auf dyadische Zerlegungen des Kreises und des Einheitsintervalls. Man kann Vektoren oder Tensoren auf diesen Zerlegungen platzieren und daraus im direkten Limes einen Hilbertraum konstruieren, den wir hier semi-kontinuierlichen Limes nennen. $F$ und $T$ wirken über unitäre Darstellungen auf diesem Raum. Wir geben eine abstrakte Beschreibung dieser Darstellungen nach Jones. Es ist möglich, $T$ auf den Rand einer Poincaré-Scheibe wirken zu lassen, der ein Querschnitt von $\mathrm{AdS}_3$ bei konstanter Zeit ist. Das definiert eine Darstellung von $T$ auf dem Hilbertraum, der alle baumartigen holographischen Zustände nach \citeauthor{PastawskiYoshidaHarlowPreskill2015} enthält. Desweiteren erhält man  eine Bulk-Boundary-Korrespondenz durch Imberts Isomorphismus zwischen $T$ und Penners Ptolemäus-Gruppe. Außerdem schlagen wir Definitionen für Feldoperatoren und Korrelationsfunktionen der diskreten Theorie vor. Schließlich skizzieren wir neue Entwicklungen wie Teilchenerzeugung und -vernichtung, schwarze Löcher und mögliche Verbindungen zur topologischen Quantenfeldtheorie.
  \\[0.8\baselineskip]
  \noindent {\itshape Schlagwörter:} Thompson-Gruppen, Feldtheorie, holografisches Prinzip

\selectlanguage{english}


\selectlanguage{ngerman}
\chapter*{Danksagungen}

Ich bedanke mich ganz herzlich bei meinem Betreuer Tobias Osborne. Er hat mich in vielen wertvollen Diskussion nicht nur mit den Themen dieser und anderer Arbeiten vertraut gemacht, sondern auch mit der wissenschaftlichen Arbeitsweise. Außerdem hat er mich sehr dabei unterstützt, meine wissenschaftliche Karriere zu beginnen. Ich habe während meiner Promotion viel fürs Leben gelernt und daran hat Tobias einen wichtigen Anteil.

Darüber hinaus danke ich den Gutachtern und Mitgliedern der Promotionskommission Elmar Schrohe, Reinhard Werner und Robert König für ihr Interesse an meinem Projekt und die Bereitschaft, an meinem Promotionsverfahren teilzunehmen.

Für hilfreiche Gespräche zu den Themen dieser Arbeit bedanke ich mich bei Markus Duwe, Terry Farrelly, Alexander Hahn, Thorsten Holm, Nathan McMahon, \mbox{Yunxiang} Ren, René Schwonnek, Philip Schwartz, Michael Walter, Daniel Westerfeld und Ramona Wolf.

Für das Korrekturlesen oder andere Hinweise zur Arbeit danke ich Terry Farrelly, Louis Fraatz, Tobias Geib, Lina Stahnke und Ramona Wolf. Ihr habt die Arbeit zu dem gemacht, was sie jetzt ist.

Ich hatte eine sehr schöne Zeit in der Arbeitsgruppe Quanteninformation. Das liegt an den vielen lieben Kollegen, die ich währenddessen kennengelernt habe. Stellvertretend für alle danke ich meiner Büronachbarin Ike Dziemba für schöne gemeinsame Arbeitszeit und Freizeit.

Meiner Familie danke ich ganz herzlich für die Unterstützung auf meinem Weg und für die unzähligen Dinge, die sie über die Jahre für mich getan hat.

Und schließlich komme ich zu der, die mich unterstützt hat, an mich geglaubt hat und einfach immer da war -- ganz, ganz lieben Dank, Ramona!

\selectlanguage{english}

\tableofcontents\cleardoublepage


\chapter*{Notation and Conventions}

\begin{center}
  \begin{tabular}{p{1.5cm}l}
    $\disk$ & the open unit disk in $\complexes$\\
    $\halfplane$ & $=\reals+\imag\,\reals_{>0}\subset\complexes$, the upper halfplane\\
    $\field$ & $=\reals\text{ or }\complexes$\\
    $\naturals$ & $=\{ 0, 1, 2, \dotsc \}$, the set of natural numbers\\
    $\naturals^\times$ & $=\naturals\setminus\{0\}$\\
    $\sircle^1$ & $=\partial\disk$, the unit circle\\
    $\integers$ & the set of integers\\
    $M_d(\field)$ & the set of all $d\times d$ matrices over $\field$\\
    &\\
    $\mathfrak{h}, \hfrak$ & Hilbert spaces\\
    $\vfrak$ & a vector space\\
    &\\
    $A, B, \dotsc$ & often denote arbitrary objects\\
    $f, g, \dotsc$ & often denote arbitrary morphisms\\
    $\chom(\ccal)$ & the class of morphisms of the category $\ccal$\\
    $\obj(\ccal)$ & the class of objects of the category $\ccal$\\
    $\ccal(A, B)$ & the set of morphisms from $A$ to $B$\\
    $\dom(f)$ & the domain of the morphism $f$\\
    $\cod(f)$ & the codomain of the morphism $f$\\
    &\\
    $\id_A$ & the identity morphism of the object $A$\\
    $\Id_\ccal$ & the identity functor of the category $\ccal$\\
    $I$ & the monoidal unit object\\
    $\idmtx$ & an identity matrix\\
    &\\
    $f^*$ & the dual of $f$ (rotation by $180^\circ$)\\
    $f^\dagger$ & `$f$ dagger' (horizontal reflection)\\
    &\\
    $\subset$ & inclusion (with the possibility of equality)\\
    $\circ$ & composition of morphisms\\
    $\otimes$ & the monoidal product bifunctor\\
    &\\
    &
  \end{tabular}
\end{center}


\chapter*{Introduction}

\addtocontents{toc}{\protect\setcounter{tocdepth}{0}}

A few years ago Vaughan F.\ R.\ Jones initiated a research program to build a continuum limit theory using two groups known as Thompson's group $F$ and $T$. Thompson's groups are certain groups of piecewise-linear homeomorphisms. The elements of $F$ are functions of the interval $[0, 1]$, and the elements of $T$ are functions of the circle $\sircle^1$.

In the proposed limiting procedure, space is divided into smaller and smaller pieces, which may be seen as an instance of fine-graining. Thompson's groups act naturally on such divisions of space. However, the limit is called the \emph{semicontinuous limit} and not the continuum limit because it only allows for arbitrarily fine piecewise-linear functions but not smooth functions. Thompson's groups therefore serve as \emph{analogues} of the diffeomorphisms from a CFT. The resulting theory, which we call `Thompson field theory', will probably not directly give us new CFTs, but should instead be viewed as a \emph{toy model} in the true sense of the word: It is very rich and interesting, and can be studied as a new theory in its own right.

My supervisor, Tobias Osborne, had been interested in this topic for a while when I joined him in October 2015 as a doctoral candidate.
In the beginning, we tried to find a correct abstract description of unitary representations of Thompson's group $F$ (and later, Thompson's group $T$) such that the groups discretely deform space. Over a long period of time, many different constructions were tried. In 2016, Jones gave a first abstract description \cite{Jones18}. Today we know how these representations are built in full generality. This is detailed in \cref{cha:localization}.

Another facet of the picture are connections to holographic duality and a bulk-boundary correspondence known as the AdS-CFT correspondence (\cref{cha:dynamics}). Thompson's group $T$ acts on the circle---what if the circle is the boundary of a Poincaré disk? It turns out that mathematicians have worked on this in the context of something known as universal Teichmüller theory: It follows from the remarkable works of Penner and Imbert that $T$ is isomorphic to a group that generates certain tessellations of the Poincaré disk, known as the Ptolemy group \cite{Penner1997,Imbert1997}. This is where the connection with holography can be made: \citeauthor{PastawskiYoshidaHarlowPreskill2015} have introduced special tensor networks known as holographic error-correcting codes which  are \emph{kinematical} toy models for the bulk-boundary correspondence. They are defined on the Poincaré disk because the disk is an equal-time slice of $\mathrm{AdS}_3$. When we tile the disk with special tensors (known as `perfect tensors'), then we obtain states on which $T$ can act. This leads to a unitary representation of $T$, and due to its isomorphism with the Ptolemy group, this is also a manifestation of the bulk-boundary correspondence. In this picture, holographic states take the role of vacuum states, and $T$ gives the dynamical evolution and is an analogue of the chiral conformal group $\diff_+(\sircle^1)$.

We come back to Thompson's groups $F$ and $T$ acting on the interval and the circle, respectively. In order to make our model a good toy model of a quantum field theory, we need to propose a definition of field operators. They can readily be defined on coarse-grained systems, but in the semicontinuous limit, we need to resort to some physical argument. Since the theory is, as a \emph{physical theory}, determined by the behaviour of its correlations function, we explain how to compute correlation functions of (putative) field operators in the semicontinuous limit, which is like an implicit partial definition of these operators. We also define an analogue of the scaling dimension, and the operator product expansion, again in terms of correlation functions. We complete our proposals by providing two examples of how to compute fusion rules: a spin system and the Fibonacci category. This is the content of \cref{cha:quantum-fields}.

Throughout this work, we have strived for mathematical rigour as much as possible, except for some parts of \cref{cha:quantum-fields} and some prototypical ideas in \cref{cha:particles}. It is certainly true that the abstractness of category theory can initially hamper the understanding of a work, but it should certainly be used if it is worth it. We argue that this is the case here: Not only are we using category theory to encapsulate long definitions and formalize our theory, but we also \emph{do} actual category theory when we use the Kan extension to construct fraction groups together with their representations. What's more, the abstract formulation will allow us in the future to find dynamics for tensor networks with more complicated underlying graphs, since it makes transparent the necessary and sufficient aspects of the construction.

The thesis is divided into two parts. \cref{part:structure} contains most preliminaries needed for \cref{part:applications}. However, \cref{part:structure} also contains material that is very interesting in its own right, like the Kan extension of a localization functor (\cref{cha:localization}), our calculations for trivalent categories (\cref{cha:trivalent}), or the comprehensive overview of Thompson's groups (\cref{cha:thompson}).

Above we have already outlined \cref{cha:dynamics,cha:quantum-fields}. In \cref{cha:particles}, we elaborate on some new ideas that have not yet been fully developed but already show the surprising richness of Thompson field theory. Specifically, we discuss models of particle creation and annihilation and of black holes. We also briefly make connections with cobordisms.

\subsection*{Contributions}

Here we list what we believe are the novel contributions found in this work, both major and minor, in order of appearance:
\begin{enumerate}[label=\arabic*.]
  \item a very general description of Jones' representations (\cref{sec:fractrep})
  \item a direct and computationally transparent proof of how to approximate diffeomorphisms by elements of Thompson's groups (\cref{sec:approximating-diffeomorphisms})
  \item the construction of dynamics for holographic codes \& the bulk-boundary correspondence (\cref{cha:dynamics})
  \item quantum fields for Thompson's groups (\cref{cha:quantum-fields})
  \item prototype dynamics for particles \& black holes (\cref{sec:particles,sec:black-holes})
  \item ideas for discrete cobordisms (\cref{sec:cobordisms})
\end{enumerate}

\subsection*{A Guide to This Thesis}

Finally, here are some suggestions on how to read the thesis. This is for anyone who prefers to begin with an overview. I suggest the following order:
\begin{enumerate}[label=\arabic*.]
  \item this introduction
  \item the introductions of the chapters
  \item \cref{cha:dynamics,cha:quantum-fields}, and bits from chapters in \cref{part:structure} where needed
\end{enumerate}
After this, the rest of the thesis should be more accessible.

\section*{Publications}
\sectionmark{Introduction}

\begin{center}
  \vspace*{0.1\baselineskip}
  The thesis is based on the following articles:
  \vspace*{0.3\baselineskip}
\end{center}

  \noindent\fullcite{OsborneStiegemann2017}
  \vspace*{0.3\baselineskip}

  \noindent\fullcite{Stiegemann2018b}
  \vspace*{0.3\baselineskip}

  \noindent\fullcite{OsborneStiegemann2019}

\addtocontents{toc}{\protect\setcounter{tocdepth}{1}}

\part{Structure}\label{part:structure}


\chapter{Categories and Monoidal Structure}
\label{cha:cats}

Categories axiomatize the general notion of an associative partially defined composition with neutral elements, together with notions of structure preserving maps and so-called natural transformations. In the present work, the use of categorical language serves three purposes:
\begin{itemize}
	\item \textsl{It makes it easier to connect different ideas in a rigorous way.} For instance, it allows us to treat trivalent categories and spin systems under a unified framework.
	\item \textsl{It is a convenient form of book-keeping of long definitions.} As an example, take the definition of a spherical category.
	\item \textsl{The level of abstraction makes it easier to decide if conditions are necessary or sufficient, or neither.} In this way, it can help us in the very formation of new useful definitions. See e.g.\ \cref{sec:fractrep}.
\end{itemize}

This chapter is a summary of notions from category theory that will be used in later chapters. The aim is to fix notation and make our work largely self-contained. However, the exposition is not an introduction to the topics covered, and we direct the reader to the literature cited at the beginning of each section for more information.

\section{Categories, Functors and Natural Transformations}

Many introductory texts on category theory have been written, satisfying the tastes of different groups of readers. In addition to the standard reference \cite{MacLane98}, we recommend a combination of the books \cite{Leinster14,Awodey16,Spivak2014}.

\begin{defn}
	A \defemph{category} $\ccal$ consists of
	\begin{itemize}
		\item a collection of \defemph{objects} $\obj(\ccal)$;
		\item for each pair $A, B\in\obj(\ccal)$ a set $\ccal(A, B)$ of \defemph{morphisms};
		\item for each triple $A, B, C\in\obj(\ccal)$ a map
			\begin{equation}
				\ccal(A, B)\times\ccal(B, C)\to\ccal(A, C),
					\quad (f, g)\mapsto f\circ g;
			\end{equation}
      called \defemph{composition};
    \item for each $A\in\obj(\ccal)$ an element $\id_A\in\ccal(A, A)$, called the \defemph{identity} on $A$,
	\end{itemize}
	such that
	\begin{genenum}
		\item $\id_B\circ f=f=f\circ\id_A$ for all $A, B\in\obj(\ccal)$ and $f\in\ccal(A, B)$;
		\item $(h\circ g)\circ f=h\circ(f\circ g)$ for all $A, B, C, D\in\obj(\ccal)$, $f\in\ccal(A, B)$, $g\in\ccal(B, C)$, and $h\in\ccal(C, D)$ (associativity).
	\end{genenum}
\end{defn}

Instead of $f\in\ccal(A, B)$, we will often write $f\colon A\to B$. $A$ is called the \defemph{domain} of $f$, denoted $\dom(f)$, and $B$ is the \defemph{codomain} $\cod(f)$. The collection of all morphisms of $\ccal$ is denoted $\hom(\ccal)$, and we will sometimes write $f\in\ccal$ instead of $f\in\hom(\ccal)$. The set of morphisms between two objects is called a hom-set. If two morphisms $f, g$ are such that $\cod(f)=\dom(g)$, they are called \defemph{composable} since $g\circ f$ is then defined. Sometimes we will write $g\circ f$ without further justification, implicitly assuming that $f$ and $g$ are composable.

A category can be viewed as a special directed graph, with the morphisms being edges and the objects vertices.

\begin{exmp}
	The smallest non-empty category $\ical$ has a single object $\star$ and only one morphism ${\id_\star}$.
\end{exmp}

\begin{exmp}
	A \defemph{proset} $J$ is a small category such that for any pair of objects $x, y$ in $J$, there is at most one morphism $x\to y$. Here, the term `proset' is an abbreviation of `preordered set', and existence of a morphism $x\to y$ on the category side corresponds to the relation $x\le y$.
\end{exmp}

\begin{exmp}\label{exmp:groupoid}
	A \defemph{groupoid} is a category in which every morphism is an isomorphism.
	A \defemph{group} is a groupoid with exactly one object. If $G$ is a group and $\star$ its one object, then the morphisms in $G$ are precisely the automorphisms of $\star$, and the set $\autom_G(\star)$ of automorphisms of $\star$ forms a group, under composition of morphisms, in the traditional sense of algebra. $\ical$ represents the trivial group.
\end{exmp}

\begin{defn}
	Let $\acal$ be a category. An object $I$ in $\acal$ is \defemph{initial} if for every object $A$ in $\acal$, there is exactly one morphism $I\to A$. An object $T$ in $\acal$ is \defemph{terminal} if for every object $A$ in $\acal$, there is exactly one morphism $A\to T$.
\end{defn}

\begin{defn} \label{defn:functor}
	Let $\acal$ and $\bcal$ be categories. A \defemph{functor} $F\colon\acal\to\bcal$ consists of
	\begin{itemize}
		\item a function $\obj(\acal)\to\obj(\bcal)$, assigning to each object $A$ of $\acal$ an object $F(A)$ of $\bcal$,
		\item for all objects $A, B$ in $\acal$, a function $\acal(A, B)\to\bcal(F(A), F(B))$, assigning to each morphism $f\colon A\to B$ in $\acal$ a morphism $F(f)\colon F(A)\to F(B)$ in $\bcal$,
	\end{itemize}
	such that
	\begin{genenum}[label=(F\arabic*)]
		\item\label{it:functor1} $F(f\circ g)=F(f)\circ F(g)$ for all composable $f, g\in\acal$,
		\item\label{it:functor2} $F({\id_A})={\id_{F(A)}}$ for all objects $A$ in $\acal$.
	\end{genenum}
	Sometimes a functor $\acal\to\bcal$ is also called a \defemph{diagram in $\bcal$ of type $\acal$}.
\end{defn}

\begin{exmp}
	The identity functor $\Id_\ccal\colon\ccal\to\ccal$ is the identity map on objects and on morphisms.
\end{exmp}

Functors are structure-preserving maps between categories and can therefore be viewed as morphisms between categories. Functors are related to two other mathematical structures. The concepts of \emph{representation} of an algebraic structure and of \emph{module} over an algebraic structure are different perspectives on the concept of functor. If $\acal$ is a category corresponding to an algebraic structure, and $\ccal$ is another category, then a functor $F\colon\acal\to\ccal$ can be viewed as a representation of $\acal$ on $\ccal$ (or, in the traditional language, on one or several objects of $\ccal$). $F$ can also be regarded as a module\footnote{A module is an algebraic structure similar to a vector space but more general, since scalars are only required to be elements of a ring, not a field.} on $\ccal$ over $\acal$; for morphisms $f\in\acal$ and $c\in\ccal$, we consider
\begin{equation}
	f\cdot c=F(f)\circ c
\end{equation}
to be the multiplication of the scalar $f$ with the element $c$. \cref{it:functor1,it:functor2} translate to the conditions
\begin{equation}
	F(f\circ g)\circ c= F(f)\circ F(g)\circ c, \quad F({\id_A})\circ c={\id_{F(A)}}\circ c = c,
\end{equation}
(where $\cod(c)=F(A)$), which in the language of modules is written
\begin{equation}
	(fg)\cdot c= f\cdot (g\cdot c), \quad 1\cdot c = c.
\end{equation}

For the sake of completeness, we give here the full defintion of comma categories.

\begin{defn}
	Given categories and functors $\acal\xrightarrow{F}\ccal\xleftarrow{G}\bcal$, the \defemph{comma category} $(F\downarrow G)$ is defined as follows:
	\begin{itemize}
		\item the \emph{objects} are all triples $(A, B, \alpha)$ with $A\in\obj(\acal)$, $B\in\obj(\bcal)$, and $\alpha\colon F(A)\to G(B)$;
		\item the \emph{morphisms} $(A, B, \alpha)\to (A', B', \alpha')$ are all pairs $(f, g)$ of morphisms $f\colon A\to A'$ and $g\colon B\to B'$ such that
		\begin{equation}
			\begin{tikzcd}
				F(A) \ar[r, "F(f)"] \ar[d, "\alpha", '] & F(A') \ar[d, "\alpha'"]\\
				G(B) \ar[r, "G(g)", ']& G(B')
			\end{tikzcd}
		\end{equation}
		commutes;
		\item \emph{composition} is given by $(f', g')\circ (f, g)=(f'\circ f, g'\circ g)$;
		\item the \emph{identity} for an object $(A, B, \alpha)$ is $({\id_A}, {\id_B})$.
	\end{itemize}
\end{defn}

Next, we introduce natural transformations. Besides categories and functors, they are the third important ingredient of category theory.

\begin{defn}
	Let $\acal$ and $\bcal$ be categories and let $F, G\colon\acal\to\bcal$ be functors. A \defemph{natural transformation} $\alpha\colon F\to G$ consists of a family of morphisms $\alpha_A\colon F(A)\to G(A)$, one for each object $A$ in $\acal$, such that for every $f\colon A\to B$ in $\acal$, the square
	\begin{equation}
		\begin{tikzcd}
			F(A) \ar[r, "F(f)"] \ar[d, "\alpha_A", '] & F(B) \ar[d, "\alpha_B"]\\
			G(A) \ar[r, "G(f)", '] & G(B)
		\end{tikzcd}
	\end{equation}
	commutes. The morphisms $\alpha_A$ are called the components of $\alpha$.
\end{defn}

Natural transformations can be understood as morphisms between functors. Indeed, there is an associative composition of natural transformations; given $\alpha\colon F\to G$ and $\beta\colon G\to H$, where $F, G, H\colon\acal\to\bcal$ are functors, the composite natural transformation $\beta\circ\alpha\colon F\to H$ is defined by $(\beta\circ\alpha)_A=\beta_A\circ\alpha_A$ for all objects $A$. There is also an identity natural transformation $\id_F\colon F\to F$, defined by $(\id_F)_A=\id_{F(A)}$. Therefore for any two categories $\acal$ and $\bcal$, there is a category whose objects are the functors $\acal\to\bcal$ and whose morphisms are natural transformations. It is called the \defemph{functor category} from $\acal$ to $\bcal$ and denoted $[\acal, \bcal]$.

\section{Limits and Colimits}

The following is standard material found in any textbook on category theory.
Let $\jcal, \ccal$ be categories. For every object $C$ in $\ccal$, the \defemph{constant functor} $\Delta(C)\colon\jcal\to\ccal$ sends each object of $\jcal$ to $C$ and each morphism to ${\id_C}$. For every morphism $f\colon C\to C'$ in $\ccal$, the \defemph{constant natural transformation} $\Delta(f)\colon\Delta(C)\to\Delta(C')$ is defined by having the constant component $f$. The \defemph{diagonal functor} $\Delta\colon\ccal\to[\ccal, \jcal]$ maps each object $C$ in $\ccal$ to the constant functor $\Delta(C)$ and each morphism $f\in\ccal$ to the constant natural transformation $\Delta(f)$.

Let $F\colon\jcal\to\ccal$ be a diagram in $\ccal$ of type $\jcal$. The comma category $(\Delta\downarrow F)$ is called the \defemph{category of cones to $F$}; $(F\downarrow \Delta)$ is the \defemph{category of cocones from $F$}. Every terminal object of $(\Delta\downarrow F)$ is called a \defemph{limit of $F$}; every initial object of $(F\downarrow\Delta)$ is called a \defemph{colimit of $F$}.

Since we will need colimits later on, we spell out the definition of cocone and colimit in more readable terms. An object of $(F\downarrow\Delta)$ is (effectively) a pair $(C, \eta)$ with an object $C$ in $\ccal$ and a natural transformation $\eta\colon F\to\Delta(C)$. This means that a cocone from $F$ to $C$ is given by a family of morphisms $\eta_x\colon F(x)\to C$ for each object $x$ in $\jcal$ such that for every morphism $j\colon x\to y$ in $\jcal$,
\begin{equation}
	\begin{tikzcd}
		F(x) \ar[rd, "\eta_x", '] \ar[rr, "F(j)"] & & F(y) \ar[ld, "\eta_y"] \\
		& C &
	\end{tikzcd}
\end{equation}
commutes. We will refer to the morphisms $\eta_x$ as the components of the cocone. A cocone $(C, \eta)$ is a colimit if for every other cocone $(D, \tilde\eta)$, there exists exactly one morphism $u\colon C\to D$ such that
\begin{equation}
	\begin{tikzcd}
		F(x) \ar[ddr, bend right, "\tilde\eta_x", '] \ar[rd, "\eta_x"] \ar[rr, "F(j)"] & & F(y) \ar[ddl, bend left, "\tilde\eta_y"] \ar[ld, "\eta_y", '] \\
		& C \ar[d, "u", dashed] &\\
		& D &
	\end{tikzcd}
\end{equation}
commutes.

\section{Monoidal and Linear Categories}

In this section and the next, we freely follow several works on monoidal and linear categories, depending on our later needs. A standard reference is \cite{EtingofGelakiNikshychOstrik2015}; more simple expositions can be found in \cite{Muger2003,Quinn2017}.

\begin{defn}
	A \defemph{strict monoidal category} $(\ccal, \otimes,I)$ is a category $\ccal$ together with a functor $\otimes\colon\ccal\times\ccal \to \ccal$, called the \defemph{monoidal product}, and an object $I$ in $\ccal$, called the \defemph{unit object}, such that
	\begin{genenum}
		\item $(A\otimes B)\otimes C=A\otimes(B\otimes C)$ and $(f\otimes g)\otimes h=f\otimes (g\otimes h)$,
		\item $I\otimes A=A=A\otimes I$ and $\id_I\otimes f=f\otimes\id_I=f$
	\end{genenum}
	for all $A, B, C\in\obj(\ccal)$ and $f, g, h\in\chom(\ccal)$.
\end{defn}

There is the more general notion of a---not necessarily strict---monoidal category, where the above identities are replaced with isomorphisms satisfying coherence conditions. We do not need this level of generality. Therefore, in this work we will assume all monoidal categories to be strict monoidal.

\begin{defn}
	An \defemph{additive category} is a category $\ccal$ with the following properties:
	\begin{genenum}[label=(A\arabic*)]
		\item Every hom-set is an additive abelian group, and composition of morphisms in $\ccal$ is biadditive.
		\item There is an object $0$ in $\ccal$ such that $\ccal(0, 0)$ is the trivial group.
		\item\label{it:additive3} for all objects $A_1, A_2$ there exists an object $B$ and morphisms $p_j\colon B\to A_j$ and $i_j\colon A_j\to B$ for $j=1, 2$ such that
		\begin{equation}
			p_1\circ i_1=\id_{A_1}, \quad p_2\circ i_2=\id_{A_2}, \quad i_1\circ p_1+i_2\circ p_2=\id_B.
		\end{equation}
	\end{genenum}
\end{defn}

We will denote the neutral elements of hom-sets by $0$ as well. In the situation of \cref{it:additive3}, it is easy to see that $p_2\circ i_1=0$ and $p_1\circ i_2=0$. Furthermore, the object $B$ is unique up to unique isomorphism. We denote it by $A_1\oplus A_2$ and call it the direct sum of $A_1$ and $A_2$; the $p_j$ are to be thought of as projections and the $i_j$ as embeddings.

\begin{defn}
	A category $\ccal$ is \defemph{$\field$-linear} if every hom-set is a finite-dimensional $\field$-linear space and composition of morphisms is bilinear. A monoidal category $\ccal$ is called $\field$-linear if additionally the monoidal product is bilinear.
\end{defn}

Every $\field$-linear category is automatically an additive category, with addition given by addition of vectors.

Two objects $A$ and $B$ in an additive category $\ccal$ are said to be \defemph{disjoint} if $\ccal(A, B)$ is the trivial group.

\begin{defn}
	An object $X$ in a $\field$-linear category $\ccal$ is called \defemph{simple} if $\ccal(X, X)=\field\,\id_X$. $\ccal$ is called \defemph{semisimple} if it is additive and there exists a set $\{X_i\}$ of pairwise disjoint simple objects such that every object in $\ccal$ is isomorphic to a finite direct sum of objects in $\{X_i\}$.
\end{defn}

\section{Duality in Monoidal Categories}

\begin{defn}\label{defn:duals}
	A strict monoidal category $(\ccal, \otimes,I)$ is called \defemph{rigid} if for every object $A$ in $\ccal$ there exists an object $A^*$, called its \defemph{left dual}, and a pair of morphisms
	\begin{itemize}
		\item $\coev_A\colon I\to A\otimes A^*$, called the \defemph{coevaluation}, and
		\item $\ev_A\colon A^*\otimes A\to I$, called the \defemph{evaluation},
	\end{itemize}
	such that
	\begin{equation}\label{eq:zigzag}
		\begin{tikzcd}[column sep=huge]
			A \ar[r, "\coev_A\otimes\id_A"] \ar[rd, "\id_A", '] & A\otimes A^* \otimes A \ar[d, "\id_A\otimes\ev_A"]\\
			& A,
		\end{tikzcd}\hspace{20pt}
		\begin{tikzcd}[column sep=huge]
			A^* \ar[r, "\id_{A^*}\otimes\coev_A"] \ar[rd, "\id_{A^*}", '] & A^*\otimes A \otimes A^* \ar[d, "\ev_A\otimes\id_{A^*}"]\\
			& A^*
		\end{tikzcd}
	\end{equation}
	commute. In that case, $A$ is called a \defemph{right dual} of $A^*$.
\end{defn}

Note that always $I^*\cong I$. We come to the following proposition, which can be found in \cite[Prop.~4.2.3]{Quinn2017}, see also \cite[Prop.~2.10.5]{EtingofGelakiNikshychOstrik2015}.

\begin{prop}\label{prop:duals-unique}
	Let $A$ be an object in $\ccal$. Assume that $A$ has a left dual $A^*$ with coevaluation and evaluation maps $\coev_A$ and $\ev_A$, respectively, and another left dual $\tilde A$ with coevaluation and evaluation maps $\widetilde\coev_A$ and $\widetilde\ev_A$, respectively. Then there is a unique isomorphism $\alpha\colon A^*\to\tilde A$ such that
	\begin{equation}\label{eq:duals-unique}
		\widetilde\ev_A\circ(\alpha\otimes\id_A)=\ev_A \quad\text{and}\quad (\id_A\otimes\alpha^{-1})\circ\widetilde\coev_A=\coev_A.
	\end{equation}
\end{prop}

In short, dual objects are unique up to unique isomorphism, and this isomorphism commutes with (co-)evaluations. We will therefore speak of \emph{the} left or right dual. The isomorphism and its inverse are given by
	\begin{gather}
		\alpha= A^*\xrightarrow{\id_{A^*}\otimes\widetilde\coev_A}A^*\otimes A\otimes\tilde A\xrightarrow{\ev_A\otimes\id_{\tilde A}} \tilde A,\\
		\alpha^{-1}= \tilde A\xrightarrow{\id_{\tilde A}\otimes\coev_A}\tilde A\otimes A\otimes A^*\xrightarrow{\widetilde\ev_A\otimes\id_{A^*}} A^*.
	\end{gather}

\begin{prop}
	Let $A$ and $B$ be objects in a rigid monoidal category. Then there is a unique isomorphism $\alpha_{A, B}\colon (A\otimes B)^*\to B^*\otimes A^*$ such that
	\begin{equation}
		\ev_B\circ(\id_{B^*}\otimes\ev_A\otimes\id_B)=\ev_{A\otimes B}\circ(\alpha_{A, B}^{-1}\otimes\id_{A\otimes B})
	\end{equation}
	and
	\begin{equation}
		(\id_A\otimes\coev_B\otimes\id_{A^*})\circ\coev_A=(\id_{A\otimes B}\otimes\alpha_{A, B})\circ\coev_{A\otimes B}.
	\end{equation}
\end{prop}

Given a rigid structure, it is possible to define a notion of adjunction for morphisms as well. For any morphism $f\colon A\to B$, define $f^*\colon B^*\to A^*$ by setting
\begin{equation}
	f^*=B^* \xrightarrow{\coev_A\otimes\id_{B^*}} A^*\otimes A\otimes B^* \xrightarrow{\id_{A^*}\otimes f\otimes\id_{B^*}} A^*\otimes B\otimes B^* \xrightarrow{\id_A^*\otimes\ev_B} A^*.
\end{equation}
One can show that $(g\circ f)^*=f^*\circ g^*$ and $\id_A^*=\id_{A^*}$.

From the rigid structure we have a canonical isomorphism $(A\otimes B)^{**}\cong A^{**}\otimes B^{**}$ which is given by
\begin{equation}
	(A\otimes B)^{**} \xrightarrow{(\alpha_{A, B}^*)^{-1}} (B^*\otimes A^*)^* \xrightarrow{\alpha_{B^*, A^*}} A^{**}\otimes B^{**}.
\end{equation}

\begin{defn}
	Let $\ccal$ be a rigid monoidal category. A \defemph{pivotal} structure on $\ccal$ is a collection of isomorphisms $\phi_A\colon A\to A^{**}$ for every object $A$ in $\ccal$, such that
	\begin{equation}
		\begin{tikzcd}
			A \ar[r, "f"] \ar[d, "\phi_A", '] & B \ar[d, "\phi_B"]\\
			A^{**} \ar[r, "f^{**}", '] & B^{**}
		\end{tikzcd}
	\end{equation}
	commutes for all $f\colon A\to B$ in $\ccal$, and such that
	\begin{equation}
		\begin{tikzcd}
			A\otimes B \ar[d, "\phi_{A\otimes B}", '] \ar[dr, "\phi_A\otimes\phi_B"] & \\
			(A\otimes B)^{**} \ar[r, "{\cong}"] & A^{**}\otimes B^{**}
		\end{tikzcd}
	\end{equation}
	commutes for all objects $A$ and $B$ in $\ccal$.
\end{defn}

The second condition is to ensure that the pivotal structure is compatible with the rigid and monoidal structure of the category.

\begin{defn}
	A pivotal category is called \defemph{strict pivotal} if the isomorphisms $\phi_A\colon A\to A^{**}$ and $\alpha_{A, B}\colon (A\otimes B)^*\to B^*\otimes A^*$ are identities and $I=I^*$.
\end{defn}

\begin{defn}
	Let $X$ be an object in a strict pivotal category $\ccal$, and let $X^*$ be its dual. $X$ is called \defemph{self-dual} if there is an isomorphism $\beta\colon X\to X^{*}$ and $X$ is called \defemph{symmetrically self-dual} if $\beta^*=\beta$.
\end{defn}

Given a fixed self-dual object $X$, we can define particularly simple versions of the coevaluation and evaluation morphisms for $X$ which do not have an orientation. They are defined by
\begin{align}
	(\widehat\coev\colon I\to X\otimes X) &= (\id_X\otimes\beta^{-1}) \circ \coev_X,\\
	(\widehat\ev\colon X\otimes X\to I)   &= \ev_X\circ(\beta\otimes\id_X),
\end{align}
and are useful in the graphical calculus developed in \cref{sec:graphical-calculus}.

Let $\ccal$ be strict pivotal. For every object $A$ and $f\colon A\to A$ we define morphisms in $\ccal(I, I)$ by
\begin{align}
	\tr_\mathrm{L}(f)&=\ev_{A^*} \circ (f\otimes\id_{A^*}) \circ \coev_{A},\\
	\tr_\mathrm{R}(f)&= \ev_A \circ (\id_{A^*}\otimes f) \circ \coev_{A^*}.
\end{align}

\begin{defn}
	A strict pivotal category $\ccal$ is called \defemph{spherical} if $\tr_\mathrm{L}(f)=\tr_\mathrm{R}(f)$ for all $f$.
\end{defn}

Finally, we briefly mention the concept of dagger category, which we will use a few times in this work. A dagger is a functor $\dagger\colon\ccal^\mathrm{op}\to\ccal$ which is the identity on objects and an involution on morphisms, that is, $\dagger\circ\dagger=\Id_\ccal$. It is usually required or inferred that the dagger operation is compatible with additional structure. Applied to morphisms, the most important of these constraints are
\begin{align}
	(f+g)^\dag&=f^\dag+g^\dag,& (f\otimes g)^\dag&=f^\dag\otimes g^\dag,\\
	(f^*)^\dag&=(f^\dag)^*,& \tr(f^\dag)&=\tr(f)^\dag.
\end{align}
In a $\complexes$-linear category, the dagger can be linear or conjugate-linear.

\section{The Graphical Calculus of Monoidal Categories}
\label{sec:graphical-calculus}

Monoidal categories have a very simple graphical calculus \cite{JoyalStreet1991,Selinger2010}. Given any category, morphisms are represented by planar graphs with labelled vertices, called `diagrams', and objects are represented by labelled edges, also called `legs'. Diagrams are read \emph{from bottom to top}. For instance, if $f\colon A\to B$ is a morphism, it is drawn as
\begin{equation}
	\begin{tikzpicture}[scale=0.1,baseline=20]\scriptsize
		\draw (0, 0) node[above left] {$A$} -- (0, 16) node[below left] {$B$};
		\filldraw[fill=vlightgray] (0, 8) circle [radius=2] node {$f$};
	\end{tikzpicture}\;.
\end{equation}
Every morphism diagram has an ingoing leg, representing the domain, and an outgoing leg, representing the codomain, and diagrams are to be read from bottom to top. Composition is represented by the operation of attaching legs to legs, so that
\begin{equation}
	\begin{tikzpicture}[scale=0.1]\scriptsize
		\draw (0, 0) node[above left] {$A$} -- (0, 13) node[left] {$B$} -- (0, 26) node[below left] {$C$};
		\filldraw[fill=vlightgray] (0, 8) circle [radius=2] node {$f$};
		\filldraw[fill=vlightgray] (0, 18) circle [radius=2] node {$g$};
	\end{tikzpicture}
\end{equation}
is the diagram for $g\circ f=A \xrightarrow{f} B \xrightarrow{g} C$. Identities are more simply represented by straight lines, so
\begin{equation}
	\begin{tikzpicture}[scale=0.1,baseline=20]\scriptsize
		\draw (0, 0) node[above left] {$A$} -- (0, 16) node[below left] {$A$};
		\node[fill=vlightgray,draw,rounded corners] at (0, 8) {$\id_A$};
	\end{tikzpicture}
	\,=
	\begin{tikzpicture}[scale=0.1,baseline=20]\scriptsize
		\draw (0, 0) -- (0, 8) node[left] {$A$} -- (0, 16);
	\end{tikzpicture}
	\;.
\end{equation}
Sometimes even the label of the identity (here:\ $A$) is omitted when it is clear from context.

The monoidal product of morphisms corresponds to horizonal juxtaposition of diagrams, and the monoidal product of objects corresponds to edges drawn side by side. Given morphisms $f\colon A\to A'$ and $g\colon B\to B'$, we have
\begin{equation}
	\begin{tikzpicture}[scale=0.12,baseline=25]\footnotesize
		\draw (0, 0) node[above left] {$A\otimes B$} -- (0, 16) node[below left] {$A'\otimes B'$};
		\node[fill=vlightgray,draw,rounded corners] at (0, 8) {$f\otimes g$};
	\end{tikzpicture}
	\,=\!
	\begin{tikzpicture}[scale=0.12,baseline=25]\footnotesize
		\draw (-3, 0) node[above left] {$A$} -- (-3, 16) node[below left] {$A'$};
		\draw (3, 0) node[above left] {$B$} -- (3, 16) node[below left] {$B'$};
		\node[fill=vlightgray,draw,rounded corners,minimum width=35] at (0, 8) {$f\otimes g$};
	\end{tikzpicture}
	\,=\mspace{-6mu}
	\begin{tikzpicture}[scale=0.12,baseline=25]\footnotesize
		\draw (0, 0) node[above left] {$A$} -- (0, 16) node[below left] {$A'$};
		\filldraw[fill=vlightgray] (0, 8) circle [radius=2] node {$f$};

		\draw (8, 0) node[above left] {$B$} -- (8, 16) node[below left] {$B'$};
		\filldraw[fill=vlightgray] (8, 8) circle [radius=2] node {$g$};
	\end{tikzpicture}\,.
\end{equation}
The unit object $I$ is indicated by the absence of any edge, and the identity of the unit object, $\id_I$, is represented by the empty diagram containing no vertices and edges. For example, a morphism $h\colon I\to A\otimes B$ is depicted
\begin{equation}
	\begin{tikzpicture}[scale=0.12,baseline=-3]\footnotesize
		\draw (-3, 0) -- (-3, 8) node[below left] {$A$};
		\draw (3, 0) -- (3, 8) node[below left] {$B$};
		\node[fill=vlightgray,draw,rounded corners,minimum width=35] at (0, 0) {$h$};
	\end{tikzpicture}\,.
\end{equation}

In a rigid category, the coevalution and evaluation morphisms have special representations in the graphical calculus. They are depicted (in order) as
\begin{equation}
 \begin{tikzpicture}[scale=0.12,baseline=-3]\footnotesize
	 \draw (0, 4) node[below left,xshift=4] {$A\phantom{{}^*}$} -- (0, 0) arc [start angle=180, end angle=360, radius=3] -- (6, 4) node[below left,xshift=2] {$A^*$};
 \end{tikzpicture}
 \qquad\text{and}\mspace{24mu}
 \begin{tikzpicture}[scale=0.12,baseline=3]\footnotesize
	 \draw (0, 0) node[above left,xshift=2] {$A^*$} -- (0, 4) arc [start angle=180, end angle=0, radius=3] -- (6, 0) node[above left] {$A$};
 \end{tikzpicture}\,.
\end{equation}
These diagrams are called cup and cap.
The two defining conditions of \cref{eq:zigzag} in \cref{defn:duals} then take the form
\begin{equation}
 \begin{tikzpicture}[scale=0.12,baseline=19]\footnotesize
	 \draw (0, 0) node[above left] {$A$} -- (0, 8) arc [start angle=0, end angle=180, radius=3] -- node[pos=0.5, left] {$A^*$} (-6, 5) arc [start angle=360, end angle=180, radius=3] -- (-12, 13) node[below left] {$A$};
 \end{tikzpicture}
 \,=
 \begin{tikzpicture}[scale=0.12,baseline=19]\footnotesize
	 \draw (0, 0) -- node[pos=0.5, left] {$A$} (0, 13);
 \end{tikzpicture}
 \qquad\text{and}\quad
 \begin{tikzpicture}[scale=0.12,baseline=19]\footnotesize
	 \draw (0, 0) node[above left] {$A^*$} -- (0, 8) arc [start angle=180, end angle=0, radius=3] -- node[pos=0.5, left] {$A$} (6, 5) arc [start angle=180, end angle=360, radius=3] -- (12, 13) node[below left] {$A^*$};
 \end{tikzpicture}
 \,=
 \begin{tikzpicture}[scale=0.12,baseline=19]\footnotesize
	 \draw (0, 0) -- node[pos=0.5, left] {$A^*$} (0, 13);
 \end{tikzpicture}\,.
\end{equation}
As another example, the uniqueness statement of \cref{prop:duals-unique} becomes clearer when we write the two equations of \cref{eq:duals-unique} as equations of diagrams, so that
\begin{equation}
 \begin{tikzpicture}[scale=0.12,baseline=25]\footnotesize
	 \draw (0, 0) node[above left] {$A$} -- (0, 14) arc[start angle=0,end angle=180,radius=3] -- (-6, 11) node[pos=0.5,left] {$\tilde A$} -- (-6, 0) node[above left,xshift=2] {$A^*$};
	 \node[below] at (-3, 17) {${\sim}$};
	 \filldraw[fill=vlightgray] (-6, 8) circle [radius=2] node {$\alpha$};
 \end{tikzpicture}
 \,=
 \begin{tikzpicture}[scale=0.12,baseline=3]\footnotesize
	 \draw (0, 0) node[above left,xshift=2] {$A^*$} -- (0, 4) arc [start angle=180, end angle=0, radius=3] -- (6, 0) node[above left] {$A$};
 \end{tikzpicture}
 \qquad\text{and}\quad
 \begin{tikzpicture}[scale=0.12,baseline=-30]\footnotesize
	 \draw (0, 0) node[below left] {$A$} -- (0, -14) arc[start angle=180,end angle=360,radius=3] -- (6, -11) node[pos=0.5,left] {$\tilde A$} -- (6, 0) node[below left,xshift=2] {$A^*$};
	 \node[above] at (3, -17) {${\sim}$};
	 \filldraw[fill=vlightgray] (6, -8) ellipse [x radius=3, y radius=2] node[yshift=0.6] {$\alpha^{-1}$};
 \end{tikzpicture}
 \,=
 \begin{tikzpicture}[scale=0.12,baseline=-3]\footnotesize
	 \draw (0, 4) node[below left,xshift=4] {$A\phantom{{}^*}$} -- (0, 0) arc [start angle=180, end angle=360, radius=3] -- (6, 4) node[below left,xshift=2] {$A^*$};
 \end{tikzpicture}\,,
\end{equation}
where we have indicated $\widetilde\coev_A$ and $\widetilde\ev_A$ with a tilde.

If $X$ is a symmetrically self-dual object in a strict pivotal category, then duals can be written particularly simple. Whenever the coevaluation or evaluation and $X^*$ appear together, they can be replaced by $\widehat\coev$ or $\widehat\ev$, which we depict as
\begin{equation}
	\begin{tikzpicture}[scale=0.12,baseline=-3]\footnotesize
 	 \draw (0, 4) node[below left] {$X$} -- (0, 0) arc [start angle=180, end angle=360, radius=3] -- (6, 4);
  \end{tikzpicture}
	\qquad\text{and}\quad
	\begin{tikzpicture}[scale=0.12,baseline=3]\footnotesize
		\draw (0, 0) node[above left] {$X$} -- (0, 4) arc [start angle=180, end angle=0, radius=3] -- (6, 0);
	\end{tikzpicture}\,,
\end{equation}
respectively. When $X$ is symmetrically self-dual, it admits a graphical calculus of unoriented strands where all isotopies are allowed, possibly by inserting or removing the modified (co-)evaluation morphisms.


\chapter{Trivalent Categories}
\label{cha:trivalent}

Trivalent categories are used in this work as another test case for the structures developed in later chapters, adding to the well-known case of spin systems. There is a reasonable chance that if something works with trivalent categories, it is transferrable to general fusion categories. Due to the graphical calculus, calculations in trivalent categories are particularly easy to do, both by hand and with the help of the computer \cite{Stiegemann2018}. In this chapter, we also record some calculations which are not used in later chapters but good to have in the interest of completeness and for future use.

Originally, trivalent categories were introduced by \citeauthor{MorrisonPetersSnyder2016} as the beginning of \enquote{\textelp{} a general program to automate skein theoretic arguments in quantum algebra and quantum topology} \cite{MorrisonPetersSnyder2016}. 

We recommend \cref{sub:small-diagrams} to the reader who primarily wants to learn how to do computations with trivalent diagrams, and it is also helpful to skim through \cref{sub:diagrams-in-c4}.
For more details on the graphical calculus used in this chapter we refer the reader to \cref{sec:graphical-calculus}.

\section{Basic Facts}

A spherical category $\ccal$ is called \defemph{non-degenerate} if for every morphism $f\in\chom(\ccal)$ there is another morphism $g$ such that $\tr(f\circ g)\neq 0$ in $\ccal(I, I)$.

\begin{defn}
  A triple $(\ccal, X, \tau)$ is called a \defemph{trivalent category} if the following are fulfilled:
  \begin{genenum}[label=(T\arabic*)]
    \item $\ccal$ is a strict pivotal non-degenerate $\complexes$-linear category;
    \item $X$ is a self-dual simple object in $\ccal$;
    \item the linear spaces $\ccal_n=\ccal(I, X^{\otimes n})$, for $n=0, 1, 2, 3$, have dimensions
    \begin{equation}
      \dim\ccal_0=1,\quad \dim\ccal_1=0,\quad \dim\ccal_2=1,\quad \dim\ccal_3=1;
    \end{equation}
    \item\label{item:generated} the element $\tau\in\ccal_3$, called the \defemph{trivalent vertex}, generates $\ccal$ as a pivotal category;
    \item $\tau$ is invariant under rotation.
  \end{genenum}
\end{defn}

Since $\dim\ccal_0=1$, we can identify $\ccal_0$ with the ground field $\complexes$. Such a category is also called evaluable. In this context, it is helpful to note that $I$ can be thought of as $X^{\otimes 0}$.

Property \cref{item:generated} means that every morphism in $\ccal$ is obtained from a combination of composing and tensoring the identity, (modified) evaluation and coevaluation, self-duality, and the trivalent vertex.

In this chapter, $(\ccal, X, \tau)$ will always denote a trivalent category. We recall a few important facts from \cite{MorrisonPetersSnyder2016}.

\begin{lem}
  $X$ is symmetrically self-dual.
\end{lem}

This is \cite[Lemma~2.2]{MorrisonPetersSnyder2016}. While self-duality allows us to define a modified coevaluation and evaluation without $X^*$, \emph{symmetric} self-duality allows us to drop the orientation and draw the modified coevaluation and evaluation as an undecorated cup and cap, respectively.

\begin{lem}
  Every trivalent category is spherical.
\end{lem}

(See \cite[Remark~2.5]{MorrisonPetersSnyder2016}.) In summary, every unoriented planar trivalent graph with $n$ boundary points can be interpreted as an element of $\ccal_n$, so that we can drop all labels and orientations.

Trivalents categories are also dagger categories; the $\dagger$ operation is given by horizontal reflection. It allows us to define an inner product. We set
\begin{equation}
  \langle f, g\rangle = \tr(f^\dag\circ g)
\end{equation}
for all morphisms $f, g$ with $\dom(f)=\cod(g)$ and $\cod(f)=\dom(g)$. This inner product is non-degenerate since the category is non-degenerate. In accordance with \cite{MorrisonPetersSnyder2016}, we assume that $\dagger$ is linear, and that $\langle\blank,\blank\rangle$ is bilinear.

\section{Computations with Trivalent Diagrams}

We will now explain how to simplify and evaluate trivalent diagrams. For diagrams with up to three legs---here called `small diagrams'---the respective simplification rules follow directly from the definition of, and basic facts about, trivalent categories explained in the previous section.

For practical and meaningful calculations, we also need to know the relations between diagrams in $\ccal_4$. In order to accomplish this, we explain how to compute an orthonormal basis for $\ccal_4$ and use it to enable further calculations. In that context, we will also recapitulate how trivalent categories are distinguished by different dimensions of $\ccal_4$.

\subsection{Small Diagrams}
\label{sub:small-diagrams}

First, since every diagram with no open legs is a scalar multiple of the empty diagram, the loop is a complex number which we call $d$,
\begin{equation}\label{eq:loop}
  \begin{tikzpicture}[scale=0.25,baseline=-1mm]
    \draw (0, 0) circle [radius=1];
  \end{tikzpicture}
  =d.
\end{equation}
$d$ is non-zero since $\ccal$ is non-degenerate and $\dim\ccal_2=1$.
Secondly, $\dim\ccal_1=0$ means that every diagram containing the lollipop diagram
\begin{equation}
  \begin{tikzpicture}[scale=0.25,yscale=-1]
    \draw (0, 0) circle [radius=1];
    \draw (0, -3) -- (0, -1);
  \end{tikzpicture}
\end{equation}
is zero. Thirdly, $\dim\mathcal{C}_2=1$ means that the bigon can be contracted as
\begin{equation}
  \begin{tikzpicture}[scale=0.25,baseline=-1mm]
    \draw (0, 0) circle [radius=1];
    \draw (0, 1) -- (0, 3);
    \draw (0, -3) -- (0, -1);
  \end{tikzpicture}=b\ %
  \begin{tikzpicture}[scale=0.25,baseline=-1mm]
    \draw (0, -3) -- (0, 3);
  \end{tikzpicture}\,,
\end{equation}
where $b$ is another complex constant which, similar to above, is non-zero because $\ccal$ is non-degenerate and $\dim\ccal_3=1$. In \cite{MorrisonPetersSnyder2016}, the normalization $b=1$ is used throughout; we will not make this choice until much later. Finally, $\dim\mathcal{C}_3=1$ means that
\begin{equation}\label{eq:triangle}
  \begin{tikzpicture}[scale=0.5,baseline=-4mm]
    \draw (0, 0) arc [start angle=180, end angle=360, radius=1];
    \draw ({1-sqrt(0.5)},{-sqrt(0.5)}) to[bend left=10] (1, -0.4);
    \draw ({1+sqrt(0.5)},{-sqrt(0.5)}) to[bend right=10] (1, -0.4);
    \draw (1, -0.4) -- (1, 0);
  \end{tikzpicture}
  = t\: \begin{tikzpicture}[scale=0.5,baseline=-4mm]
    \draw (0, 0) arc [start angle=180, end angle=360, radius=1];
    \draw (1, -1) -- (1, 0);
  \end{tikzpicture},
\end{equation}
where $t$ is a parameter which can be zero.

\subsection{Orthonormalization}

We are confronted with the following problem. We work in a linear space $\vfrak$ of diagrams and want to perform computations in $\vfrak$, in particular, investigate operators on $\vfrak$ and compute their matrix elements. We are provided with a basis $\{\beta_1,\dotsc, \beta_n\}$ for $\vfrak$ that is relatively easy to obtain and diagrammatically simple but not orthonormal. Therefore the matrix elements of an operator $A$ on $\vfrak$ with respect to this basis must be computed by hand. This task is easy if $n=\dim\vfrak$ is small, however, for larger $n$ it is more reasonable to use the computer.

With a computer we can quickly compute scalar products, which is only possible if the corresponding basis is orthonormal. This leads to the problem of computing an orthonormal basis $\{\theta_1,\dotsc,\theta_n\}$ for $\vfrak$ and converting back and forth between $\beta$ and $\theta$.

The Gram matrix for the basis $\beta$ is defined by
\begin{equation}
  G_{ij}=\braket{\beta_i, \beta_j}.
\end{equation}
It is positive semidefinite because the scalar product $\braket{\blank, \blank}$ is positive semidefinite. Therefore there exists a unitary matrix $U$ such that
\begin{equation}
  D=U^* G U
\end{equation}
is diagonal. Now set $\Theta_{ik}=U_{ki}/\sqrt{D_i}$ and
\begin{equation}
  \theta_i=\frac{1}{\sqrt{D_i}}\sum_{k} U_{ki}\beta_k=\sum_k\Theta_{ik}\beta_k.
\end{equation}
These vectors certainly form a basis of $\vfrak$. Moreover,
\begin{equation}
  \braket{\theta_i, \theta_j}=\frac{1}{\sqrt{D_i D_j}} \sum_{kl} \overline{U_{ki}} U_{lj} \braket{\beta_k, \beta_l}=\frac{1}{\sqrt{D_i D_j}} (U^* G U)_{ij}=\delta_{ij},
\end{equation}
so $\theta$ is an orthonormal basis. Gram--Schmidt orthonormalization is one of several algorithms that take as input the basis $\beta$ and output the orthonormal basis $\theta$. In other words, it gives the matrix $\Theta$ of coefficients.

\subsection{Conversion of Vectors and Matrix Elements}

Assume a vector in $\vfrak$, say $x$, is given, with decomposition
\begin{equation}
  x = \sum_k c_k \beta_k
\end{equation}
in the non-orthonormal basis $\beta$. Knowing only the scalar products $\braket{\beta_l, x}$ for $l=1,\dotsc, n$ and the coefficients $\Theta$, we can compute the components $c$ of $x$. They are given by
\begin{equation}\label{eq:vectorcoeff}
  c_k = \sum_{jl} \Theta_{jl} \Theta_{jk} \braket{\beta_l, x}
\end{equation}
since
\begin{equation}
\begin{split}
  \sum_{jl} \Theta_{jl} \Theta_{jk} \braket{\beta_l, x}
  &= \sum_{ijl} c_i \frac{U_{lj}}{\sqrt{D_j}} \frac{U_{kj}}{\sqrt{D_j}} \braket{\beta_l, \beta_i}
  = \sum_{il} c_i \sum_j U_{kj} D_j^{-1} (U^*)_{jl} G_{li}\\
  &= \sum_{il} c_i (G^{-1})_{kl} G_{li} = \sum_{il} c_i \delta_{ik} = c_k.
\end{split}
\end{equation}
Similarly, if $\braket{\beta_l, A \beta_i}$ for $i, l = 1,\dotsc, n$ are known, then the decomposition
\begin{equation}
  A \beta_i = \sum_k A_{ik} \beta_k
\end{equation}
can be computed using the formula
\begin{equation}
  A_{ik}=\sum_{jl} \Theta_{jl}\Theta_{jk} \braket{\beta_l, A \beta_i}
\end{equation}
obtained from \cref{eq:vectorcoeff} by setting $x=A\beta_i$.

\subsection{Diagrams in $\ccal_4$}
\label{sub:diagrams-in-c4}

Using the Gram--Schmidt process an orthonormal basis for $\ccal_4$ was computed. It can be used to express the square as a linear combination of faceless diagrams, which is needed in order to reduce diagrams containing squares. We will now present and explain these results. For the source code, see \cite{Stiegemann2019b,Stiegemann2018}.

In the following, $D(n, k)$ denotes the set of trivalent graphs with $n$ boundary points and at most $k$ internal faces having four or more edges. Correspondingly, $M(n, k)$ is the matrix of inner products of $D(n, k)$.

The elements of $D(4, 1)$ are named in the following order:
\begin{align}
  \beta^4_1&=\vcenter{\hbox{\begin{tikzpicture}[scale=0.2]
    \draw (0, 0) arc [start angle=180, end angle=360, radius=1.809];
    \draw (1, 0) arc [start angle=180, end angle=360, radius=0.809];
  \end{tikzpicture}}},
  &\beta^4_2&=\vcenter{\hbox{\begin{tikzpicture}[scale=0.2]
    \draw (0, 0) arc [start angle=180, end angle=360, radius=0.809];
    \draw (2.618, 0) arc [start angle=180, end angle=360, radius=0.809];
  \end{tikzpicture}}},\\
  \beta^4_3&=\vcenter{\hbox{\begin{tikzpicture}[scale=0.2]
    \draw (0, 0) arc [start angle=180, end angle=360, radius=0.809];
    \draw (2.618, 0) arc [start angle=180, end angle=360, radius=0.809];
    \draw (0.809, -0.809) arc [start angle=180, end angle=360, radius=1.309];
  \end{tikzpicture}}},
  &\beta^4_4&=\vcenter{\hbox{\begin{tikzpicture}[scale=0.2]
    \draw (0, 0) arc [start angle=180, end angle=360, radius=1.809];
    \draw (1, 0) arc [start angle=180, end angle=360, radius=0.809];
    \draw (1.809, -0.809) -- (1.809, -1.809);
  \end{tikzpicture}}},\\
  f_4&=\vcenter{\hbox{\begin{tikzpicture}[scale=0.2]
    \draw (0, 0) arc [start angle=180, end angle=360, radius=1.809];
    \draw (1, 0) arc [start angle=180, end angle=360, radius=0.809];
    \draw[rotate around={42.5:(1.809, 0)}] (1.809, -0.809) -- (1.809, -1.809);
    \draw[rotate around={-42.5:(1.809, 0)}] (1.809, -0.809) -- (1.809, -1.809);
  \end{tikzpicture}}}.
\end{align}
The matrix of scalar products is
\begin{equation}
  M(4, 1) = \begin{pmatrix}
    d^2 & d & bd & 0 & b^2d\\
    d & d^2 & 0 & bd & b^2d\\
    bd & 0 & b^2d & bdt & bdt^2\\
    0 & bd & bdt & b^2d & bdt^2\\
    b^2d & b^2d & bdt^2 & bdt^2 & w_4
  \end{pmatrix},
\end{equation}
where the value of the window diagram
\begin{equation}
  w_4 = \braket{f_4, f_4} = \vcenter{\hbox{\begin{tikzpicture}[scale=0.2]
    \draw (0, 0) arc [start angle=180, end angle=-180, radius=1.809];
    \draw (1, 0) arc [start angle=180, end angle=-180, radius=0.809];
    \draw[rotate around={42.5:(1.809, 0)}] (1.809, -0.809) -- (1.809, -1.809);
    \draw[rotate around={-42.5:(1.809, 0)}] (1.809, -0.809) -- (1.809, -1.809);
    \draw[rotate around={137.5:(1.809, 0)}] (1.809, -0.809) -- (1.809, -1.809);
    \draw[rotate around={222.5:(1.809, 0)}] (1.809, -0.809) -- (1.809, -1.809);
  \end{tikzpicture}}}
\end{equation}
has yet to be determined. $M(4, 1)$ has the Gram matrix $G^4=M(4, 0)$ of $\beta^4$ as a submatrix.

Using $G^4$, $\beta^4$ can be orthonormalized to give
\begin{equation}
  \Theta^4 = \begin{pmatrix}
    \frac{1}{\abs{d}} & 0 & 0 & 0\\
    -\frac{1}{d\sqrt{d^2-1}} & \frac{1}{\sqrt{d^2-1}} & 0 & 0\\
    -\frac{\sgn(b)\sqrt{d}}{\sqrt{(d^2-1)(d^2-d-1)}} & \frac{\sgn(b)}{\sqrt{d(d^2-1)(d^2-d-1)}} & \sqrt\frac{d^2-1}{b^2d(d^2-d-1)} & 0\\
    \frac{b+dt}{d^2-d-1}C_{\Theta^4} & \frac{b-bd-t}{d^2-d-1}C_{\Theta^4} & \frac{-d^2t+t-b}{b(d^2-d-1)}C_{\Theta^4} & C_{\Theta^4}
  \end{pmatrix},
\end{equation}
with
\begin{equation}
  C_{\Theta^4}=\sqrt{\tfrac{d^2-d-1}{d(b^2d(d-2)-2bt-(d^2-1)t^2)}}.
\end{equation}
We will denote the corresponding orthonormal basis vectors by $\theta^4_i$.
If we have $\dim\ccal_4=k$ with $k\in\{ 2, 3, 4 \}$, the diagrams $\{\beta^4_i\}_{i=1}^{k}$ form a basis of $\ccal_4$ and $(\Theta^4_{ij})_{i,j=1}^k$ provides the coefficients for an orthonormal basis of $\ccal_4$.

By using \cref{eq:vectorcoeff}, we can compute the components of the square $f_4$ with respect to $\beta^4$ for different dimensions of $\ccal_4$. They are
\begin{align}
  \beta^4(f_4) &= \left( \tfrac{b^2}{d+1}, \tfrac{b^2}{d+1} \right) &&\text{for }\dim\ccal_4=2,\\
  \beta^4(f_4) &= \left( \tfrac{b^2d-b^2-dt^2}{d^2-d-1}, \tfrac{b^2d-2b^2+t^2}{d^2-d-1}, \tfrac{(d+1)(dt^2+t^2-b^2)}{b(d^2-d-1)} \right) &&\text{for }\dim\ccal_4=3,\\
  \beta^4(f_4) &= \left( \tfrac{b(b^2+bt-t^2)}{bd+t+dt}, \tfrac{b(b^2+bt-t^2)}{bd+t+dt}, \tfrac{t^2(d+1)-b^2}{bd+t+dt}, \tfrac{t^2(d+1)-b^2}{bd+t+dt} \right) &&\text{for }\dim\ccal_4=4.\\
  \intertext{Consequently, the window diagram $w_4$ has the value}
  \beta^4(w_4) &= \frac{2b^4d}{d+1} &&\text{for }\dim\ccal_4=2,\\
  \beta^4(w_4) &= \frac{d(-3b^4+2b^4d+2b^2t^2-2b^2dt^2-t^4+d^2t^4)}{d^2-d-1}  &&\text{for }\dim\ccal_4=3,\\
  \beta^4(w_4) &= \frac{2bd(b^4+b^3t-2b^2t^2+t^4+dt^4)}{bd+t+dt} &&\text{for }\dim\ccal_4=4.
\end{align}

The matrices of scalar products have determinants
\begin{align}
  \det M(4, 0)&=b^2d^4(bd+t-dt-2b)(bd+t+dt),\\
\begin{split}
  \det M(4, 1)&=-b^2d^4(bd+t-dt-2b)\\
  &\mspace{36mu}\cdot\bigl(2b^5d+2b^4dt-4b^3dt^2+2bdt^4+2bd^2t^4\\
  &\mspace{60mu}-(bd+t+dt)w_4\bigr).
\end{split}
\end{align}
Recall from linear algebra that a set of vectors $\{v_1, \dotsc, v_n\}\subset\vfrak$ from some linear space $\vfrak$ of arbitrary dimension and with scalar product $\braket{\blank, \blank}$ is linearly independent if and only if the $n\times n$ matrix $(\braket{v_i, v_j})_{i,j=1}^k$ of scalar products has full rank. Following this, we know that if $\dim\ccal_4\le 3$, we must have $\det M(4, 0)=\det M(4, 1)=0$. Since $d\neq 0$, this means that
\begin{equation}\label{eq:pso3}
  P_{\mathit{SO}(3)}=bd+t-dt-2b=0
\end{equation}
or
\begin{equation}\label{eq:fibcase}
  bd+t+dt=0 \quad\text{and}\quad 2b^5d+2b^4dt-4b^3dt^2+2bdt^4+2bd^2t^4=0.
\end{equation}
In the case of \cref{eq:fibcase}, solving gives
\begin{equation}\label{eq:fibparams}
  d=\Phi^{\pm}, \qquad t=b\Phi^{\mp},
\end{equation}
where we have set
\begin{equation}
  \Phi^{\pm}=\frac{1}{2}\left(1\pm\sqrt5\right)
\end{equation}
for the golden ratio $\Phi^+$ and its (negative) conjugate $\Phi^-$. Incidentally, \cref{eq:pso3} then holds as well.

In conclusion, a trivalent category with $\dim\ccal_4\le 3$ must have $P_{\mathit{SO}(3)}=0$. \citeauthor{MorrisonPetersSnyder2016} have shown that in this case $\ccal$ is
\begin{itemize}
  \item an $\mathit{SO}(3)_q$ category with $d=q^2+1+q^{-1}$ if $(d, t)\neq(-1, 3/2)$, or
  \item the $\mathit{OSp}(1\mspace{1mu}|\mspace{1mu}2)$ category if $(d, t)=(-1, 3/2)$.
\end{itemize}
On the one hand, when $d\neq\Phi^{\pm}$, the dimensions of $\mathit{SO}(3)_q$ are the Riordan numbers $\dim\ccal_n=a_n$ with
\begin{equation}
  a_0=1,\quad a_1=0, \quad a_n=\frac{n-1}{n+1}(2a_{n-1}+3a_{n-2}), \quad n\ge 2,
\end{equation}
whose first values are $1, 0, 1, 1, 3, 6, 15, 35, \dotsc$.
In particular, we then have $\dim\ccal_4=3$. On the other hand, if \cref{eq:fibparams} holds then $\dim\ccal_n$ are the Fibonacci numbers $1, 0, 1, 1, 2, 3, 5, 8, \dotsc$, so that $\dim\ccal_4=2$, in which case $\ccal$ is a Fibonacci category. $\mathit{OSp}(1\mspace{1mu}|\mspace{1mu}2)$ has $\dim\ccal_4=3$.

Finally, we note that there are also categories with $\dim\ccal_4=4$. They are called cubic and described in detail in \cite{MorrisonPetersSnyder2016}. We will not need them in this work.



\chapter{Localization Theory and Kan Extensions}
\label{cha:localization}

Localization is a canonical way of making things invertible. We will need it in \cref{cha:dynamics} to construct dynamics for holographic codes. Localization can be described in full generality in the language of category theory. However, before moving to categories, we discuss a few simpler cases from abstract algebra.

\subsection*{Monoids}

Let's say we are given a set $R$ on which an associative multiplication is defined, that is, $(fg)h=f(gh)$ for all $f,g,h\in R$, and which has a neutral element $e$, such that $ef=fe=f$ for all $f\in R$. In algebra, such an $R$ is called a monoid. It is similar to a group, except in a group every element $g$ possesses an inverse $g^{-1}$ such that $g^{-1}g=gg^{-1}=e$. Elements of a monoid need not be invertible (though they can be). An example of a (commutative) monoid is given by the integers $\integers$ with multiplication and with $1$ as the neutral element.

\subsection*{From Monoids to Groups}

Out of $R$ we want to form a group $G$ which contains all elements of $R$ plus inverses for all non-invertible elements of $R$. This is possible if the elements of $R$ satisfy the following two conditions:
\begin{itemize}
  \item $\forall f, g\;\exists p, q: pf=qg$ (left Ore condition),
  \item ($\forall f, g, v)\;(fv = gv) \implies (\exists w : wf=wg)$ (cancellation).
\end{itemize}
If these are satisfied, such a $G$ can be constructed. $G$ will be \emph{universal} in the sense that in order to obtain $G$ we will only add exactly as many new elements as needed, not less and not more.

Now we sketch the construction of the group $G$. The idea is to formally invert all morphisms by introducing fractions. First we form a set $\hat G$ consisting of all pairs $(f, g), (h, i), \dotsc$ of elements of $R$, which may be written like fractions know from elementary algebra, as in
\begin{equation}
  \frac{f}{g},\;\frac{h}{i},\;\dotsc.
\end{equation}
We define two fractions $f/g$ and $f'/g'$ to be equivalent, written $f/g\sim f'/g'$, if there exists a $p\in R$ such that $pf=f'$ and $pg=g'$. Next, we form the set of equivalence classes $G=\hat G/{\sim}$. Therefore in $G$ we have
\begin{equation}
  \left[\frac{f}{g}\right]_\sim=\left[\frac{pf}{pg}\right]_\sim
\end{equation}
for all $p\in R$. (Here $[\blank]_\sim$ denotes $\sim$ equivalence classes; we will omit the brackets in the remainder of this paragraph and say `fraction' instead of `equivalence class of fractions'.) This equation is similar to fractions of numbers, where for example
\begin{equation}
  \frac{5}{3}=\frac{2\cdot 5}{2\cdot 3}.
\end{equation}

Given two fractions $f/g$ and $h/i$, their product is defined with the help of the left Ore condition. By this condition, there exist $p, q\in R$ such that $pg=qh$. Then we define
\begin{equation}
  \frac{f}{g}\frac{h}{i}=\frac{pf}{pg}\frac{qh}{qi}=\frac{pf}{\cancel{pg}}\frac{\cancel{qh}}{qi}=\frac{pf}{qi},
\end{equation}
where in the third term we have illustrated the idea behind the definition. (This is more complicated than the multiplication of rational numbers, which is a particularly simple case.) The neutral element for the multiplication is $e/e$. Using the left Ore condition and cancellation one can check that everything is well-defined and obtains the group $G$ of fractions of $R$. For the example $R=\integers$ (with multiplication), the group $G$ is the group of rational numbers $\rationals$ with multiplication.

\subsection*{Rings and Modules}

It is also possible and common to localize the multiplicative monoid of a ring $R$. Furthermore, if the ring is part of a left $R$-module and the multiplicative monoid of $R$ satisfies the left Ore and cancellation conditions, one can define a sensible notion of localization of a module, which really means that the ring of scalars of the module is being localized. If $m$ is an element of a left $R$-module $M$ and $r, s\in R$, then by a procedure similar to the above for monoids, a new module $R^{-1}M$ may be defined whose elements are of the form
\begin{equation}
  r \cdot \frac{m}{s}
\end{equation}
It can be proved that $R^{-1}M$ isomorphic to a module in which the localization is only applied to the ring, so the above expression would rather be
\begin{equation}
  \frac{r}{s} \cdot m.
\end{equation}
Formally, this module is $R^{-1}R\otimes_RM$.

In any case, as outlined in the explanation of \cref{defn:functor} in \cref{cha:cats}, the module $M$ can be viewed as a representation of the ring $R$. This means that the localization of the module can be viewed as the canonical way to form the `localization of the representation' of $R$ on $M$.

Finally, as a simplified version of Jones' representations, consider a term of the form
\begin{equation}
  \frac{t}{m}
\end{equation}
where the element of the module is in the denominator and $t\in R$. As an analogue to the above definition of multiplication in the localization, one can image that
\begin{equation}
  \frac{r}{s}\cdot\frac{t}{m}=\frac{pr}{ps}\frac{qt}{qm}=\frac{pr}{qm},
\end{equation}
with $ps=qt$, is the correct scalar multiplication.

\vspace{\baselineskip}
\noindent For general categories the above situations are often similar, except composition of morphisms is only defined for composable morphisms. Furthermore, in our simplified version we have inverted \emph{all} elements of the monoid, while in general it is only possible to invert a specific subset of $R$, and this is similary the case for categories. Jones' representations are a more general version of the localization of modules, which---as explained above---can be viewed as representations. The purpose of this chapter is to recapitulate localization of categories in detail (\cref{sec:localizations,sec:fractions}), and to then give an abstract description of Jones' representations (\cref{sec:fractrep}).

The idea of localizing monoids and rings is well established in abstract algebra, and Ore localization specifically dates back to \citeyear{Ore1931}, see \cite{Ore1931} and also \cite{Skoda2004}. The first formal exposition on categories of fractions was done by \citeauthor{GabrielZisman1967} in \cite{GabrielZisman1967}. We largely follow a combination of \cite{GabrielZisman1967,PopescuPopescu1979,Fritz2011} and recommend the latter as an introduction.

\section{Localization}
\label{sec:localizations}

A functor $F\colon\ccal\to\dcal$ is said to make a morphism $f\in\chom(\ccal)$ invertible if $F(f)$ is an isomorphism.

\begin{defn}\label{defn:localization}
  Let $\ccal$ be a category and $\Sigma$ a set of morphisms in $\ccal$. A \defemph{localization of $\ccal$ by $\Sigma$} consists of a category $\ccal[\Sigma^{-1}]$ and a functor $Q\colon\ccal\to\ccal[\Sigma^{-1}]$ such that
  \begin{genenum}[label=(L\arabic*)]
    \item $Q$ makes the elements of $\Sigma$ invertible;
    \item\label{enum:localization2} for every other functor $P\colon\ccal\to\dcal$ that makes the elements of $\Sigma$ invertible, there exists a unique functor $H\colon\ccal[\Sigma^{-1}]\to\dcal$ such that $P=H\circ Q$.
  \end{genenum}
\end{defn}

Given a localization of $\ccal$ by $\Sigma$, the universal property \cref{enum:localization2} ensures that it is unique up to unique isomorphism of categories:\footnote{If one relaxes \cref{defn:localization} and merely requires a natural isomorphism $P\cong H\circ Q$ (plus a certain third condition), then uniqueness of a localization is up to \emph{equivalence} of categories, which is more favourable from an abstract point of view. Since we do not need to exploit the notion of localization in full generality, the above definition is sufficient.} Given two localizations $(Q, \Sigma^{-1}\ccal)$ and $(P, \dcal)$ of $\ccal$, we get functors $H\colon\ccal[\Sigma^{-1}]\to\dcal$ and $H'\colon\dcal\to\ccal[\Sigma^{-1}]$ from \cref{enum:localization2}. Then again by \cref{enum:localization2}, $H'\circ H=\Id_{\ccal[\Sigma^{-1}]}$ and $H\circ H'=\Id_\dcal$, so $H$ is the desired isomorphism of categories. This justifies the unambiguous notation ``$\,\ccal[\Sigma^{-1}]\,$''.

The following theorem shows that a localization can always be constructed by formally adding inverses for all morphisms in $\Sigma$, and can be found in the literature \cite[Section~1.1]{GabrielZisman1967}, \cite[Theorem~2.1]{Fritz2011}.

\begin{thm}
  $\ccal[\Sigma^{-1}]$ and $Q\colon\ccal\to\ccal[\Sigma^{-1}]$ always exist.
\end{thm}

\begin{proof}
  We construct a diagram scheme $\tcal$ as follows. Let $\obj(\tcal)=\obj(\ccal)$ and $\chom(\tcal)=\chom(\ccal)\sqcup\Sigma$. Denote by $i_1$ and $i_2$ the canonical injections of $\hom(\ccal)$ and $\Sigma$ into $\chom(\ccal)\sqcup\Sigma$ and set
  \begin{align}
    {\dom_\tcal}\circ i_1&={\dom_\ccal}, & {\dom_\tcal}\circ i_2&={\cod_\ccal}|_\Sigma,\\
    {\cod_\tcal}\circ i_1&={\cod_\ccal}, & {\cod_\tcal}\circ i_2&={\dom_\ccal}|_\Sigma.
  \end{align}
  $\tcal$ now consists of the objects and morphisms of $\ccal$ plus an additional copy of $\Sigma$ such that the direction of all arrows coming from this copy is reversed in $\tcal$. (Effectively, $\chom(\tcal)=\chom(\ccal)\sqcup\Sigma^\mathrm{op}$.)

  Next, define the following equivalence relation on the free category $\fcal(\tcal)$ generated by $\tcal$:
  \begin{genenum}
    \item\label{item:locequiv1} $i_1(f)\circ i_1(g)\sim i_1(f\circ g)$ if $f$ and $g$ are composable in $\ccal$,
    \item\label{item:locequiv2} $i_1(\id^\ccal_a)\sim{\id^{\fcal(\tcal)}_a}$ for all $a\in\obj(\ccal)$,
    \item\label{item:locequiv3} $i_2(s)\circ i_1(s)\sim\id^{\fcal(\tcal)}_{\dom_\ccal(s)}$ and $i_1(s)\circ i_2(s)\sim\id^{\fcal(\tcal)}_{\cod_\ccal(s)}$ for all $s\in\Sigma$.
  \end{genenum}
  (Here $\id^\ccal$ and $\id^{\fcal(\tcal)}$ are the identity maps on $\ccal$ and $\fcal(\tcal)$.) The first two conditions reproduce the composition on $\ccal$; the third turns the morphisms in the range of $i_2$ into formal inverses.

  Now let $I\colon\ccal\to\fcal(\tcal)$ be the inclusion functor that is the identity on objects and maps a morphism $f$ of $\ccal$ to $\fcal(i_1(f))$ in $\fcal(\tcal)$, and let $S\colon\fcal(\tcal)\to\fcal(\tcal)/{\sim}$ be the canonical quotient functor. Define the functor $Q\colon\ccal\to\fcal(\tcal)/{\sim}$ by $Q=S\circ I$. Then by (3), $Q$ clearly makes every morphism of $\Sigma$ invertible.

  Furthermore, let $P\colon\ccal\to\dcal$ be another functor that makes the elements of $\Sigma$ invertible, and set $H\colon\fcal(\tcal)/{\sim}\to\dcal$ by letting $H(Q(f))=P(f)$. $H$ is well-defined since $P$ maps the above relations to equalities; for \cref{item:locequiv1} and \cref{item:locequiv2} this holds because $P$ is a functor and for \cref{item:locequiv3} because $P$ makes the elements of $\Sigma$ invertible. The functor $H$ thus defined is obviously unique. Therefore the universal property is fulfilled and $\fcal(\tcal)/{\sim}=\ccal[\Sigma^{-1}]$.
\end{proof}

\section{Calculus of Fractions}
\label{sec:fractions}

If we know more about the set $\Sigma$, then it can be shown that the localization takes a specific form; it can be described as the \emph{category of fractions} with denominators in $\Sigma$. We will construct this category $\Sigma^{-1}\ccal$ and show that it is isomorphic to the localization defined in the previous section. We largely follow \textcite{Krause10} and \textcite{Fritz2011}.

Let $\ccal$ be a category and $\Sigma$ a set of morphisms in $\ccal$.

\begin{defn}\label{defn:calculus-fractions}
  $\Sigma$ \defemph{admits a calculus of left fractions} if the following properties hold:
  \begin{genenum}[label=(LF\arabic*)]
    \item\label{it:lf1} If $\sigma, \tau$ are composable morphisms in $\Sigma$, then $\tau\circ\sigma\in\Sigma$. Furthermore, the identity morphism $\id_X$ is in $\Sigma$ for every object $X$ in $\ccal$.
    \item\label{it:lf2} Every pair of morphisms $X' \xleftarrow{\tau} X \xrightarrow{g} Y$ with $g\in\ccal$ and $\tau\in\Sigma$ can be completed to a commutative square
    \begin{equation}
      \begin{tikzcd}
        X \ar[r, "g"] \ar[d, "\tau", '] & Y \ar[d, "\sigma"]\\
        X' \ar[r, "f", '] & Y'
      \end{tikzcd}
    \end{equation}
    such that $f\in\ccal$ and $\sigma\in\Sigma$.
    \item\label{it:lf3} Let $f, g\colon X\to Y$ be parallel morphisms in $\ccal$. If there is a morphism $\sigma\colon X'\to X$ in $\Sigma$ with $f\circ\sigma=g\circ\sigma$, then there exists a morphism $\tau\colon Y\to Y'$ in $\Sigma$ with $\tau\circ f=\tau\circ g$.
  \end{genenum}
\end{defn}

Now assume that $\Sigma$ admits a calculus of left fractions. A pair $(f, \sigma)$ of morphisms
\begin{equation}
  \begin{tikzcd}
    X \ar[r, "f"] & Y' & \ar[l, "\sigma", '] Y
  \end{tikzcd}
\end{equation}
with $\sigma\in\Sigma$ is called a left fraction. Two left fractions $(f_1, \sigma_1)$ and $(f_2, \sigma_2)$ are defined to be equivalent (written $(f_1, \sigma_1)\sim(f_2, \sigma_2)$) if there exist morphisms $p$ and $q$ such that the diagram
\begin{equation}\label{eq:defn_fraction_equivalence}
  \begin{tikzcd}[column sep=large]
    & Y_1 \ar[d, "p"] &\\
    X \ar[ur, "f_1"] \ar[r, "f'"] \ar[dr, "f_2", '] & Y' & Y \ar[ul, "\sigma_1", '] \ar[l, "\sigma'" '] \ar[ld, "\sigma_2"]\\
    & Y_2 \ar[u, "q", '] &
  \end{tikzcd}
\end{equation}
with $f'=p\circ f_1=q\circ f_2$ and $\sigma'=p\circ\sigma_1=q\circ\sigma_2$
commutes and $\sigma'\in\Sigma$. $(f', \sigma')$ is called the \defemph{expansion} of $(f_1, \sigma_1)$ and $(f_2, \sigma_2)$.

\begin{rem}\label{rem:expansions}
  In fact, two equivalent fractions are always each equivalent to their common expansion. In the situation of \cref{eq:defn_fraction_equivalence}, this is shown by the commutative diagram
  \begin{equation}\label{eq:expansion}
    \begin{tikzcd}[column sep=large]
      & Y_1 \ar[d, "p"] &\\
      X \ar[ur, "f_1"] \ar[r, "f'"] \ar[dr, "f'", '] & Y' & Y \ar[ul, "\sigma_1", '] \ar[l, "\sigma'" '] \ar[ld, "\sigma'"]\\
      & Y' \ar[u, equal] &
    \end{tikzcd}
  \end{equation}
  and similarly for $(f_2, \sigma_2)$. Conversely, if two fractions $(f_1, \sigma_1)$ and $(f', \sigma')$ form a commutative square as in \cref{eq:expansion}, they are equivalent, and $(f', \sigma')$ is an \defemph{expansion} of $(f_1, \sigma_1)$.
\end{rem}

\begin{lem}
  The relation $\sim$ thus defined is an equivalence relation.
\end{lem}

\begin{proof}
  If $(f_1, \sigma_1)=(f_2, \sigma_2)$, then \cref{eq:defn_fraction_equivalence} commutes with the choices $(f', \sigma')=(f_1, \sigma_1)$ and $p=q=\id_{Y_1}$; therefore $\sim$ is reflexive, and it is obvious that it is also symmetric.

  For transitivity, consider the relations $(f_1, \sigma_1)\sim(f_2, \sigma_2)$ and $(f_2, \sigma_2)\sim(f_3, \sigma_3)$. Then the diagram
  \begin{equation}\label{eq:equivalence_commuting}
    \begin{tikzcd}[column sep=huge]
      & Y_1 \ar[d, "p"] &\\
      & Y' &\\
      X \ar[uur, "f_1"] \ar[ur, near end, outer sep=-1pt, "f'"] \ar[r, "f_2", '] \ar[dr, near end, outer sep=-2pt, "f''", '] \ar[ddr, "f_3", '] & Y_2 \ar[u, "q", '] \ar[d, "s", '] & Y \ar[uul, "\sigma_1", '] \ar[ul, near end, outer sep=-1pt, "\sigma'" '] \ar[l, "\sigma_2"] \ar[dl, near end, outer sep=-2pt, "\sigma''"] \ar[ddl, "\sigma_3"]\\
      & Y'' &\\
      & Y_3 \ar[u, "r"] &
    \end{tikzcd}
  \end{equation}
  commutes. The morphisms in the diagram
  \begin{equation}\label{eq:equivalence_subset}
    \begin{tikzcd}
      & & Y' \ar[ld, dashed, "t", '] &\\
      Z'& Z \ar[l, dashed, "\tau", '] & Y_2 \ar[u, "q", '] \ar[d, "s", '] & Y \ar[ul, "\sigma'" '] \ar[l, "\sigma_2"] \ar[dl, "\sigma''"]\\
      & & Y'' \ar[lu, dashed, "\rho"] &
    \end{tikzcd}
  \end{equation}
  drawn with continuous lines form a commuting subset of \cref{eq:equivalence_commuting}. By \cref{it:lf2}, we can find morphisms $t\in\ccal$ and $\rho\in\Sigma$ such that the square formed by $\sigma'$, $t$, $\rho$, and $\sigma''$ commutes. Since
  \begin{equation}
    t\circ q\circ\sigma_2=t\circ\sigma'=\rho\circ\sigma''=\rho\circ s\circ\sigma_2,
  \end{equation}
  it follows from \cref{it:lf3} that there exists $\tau\in\Sigma$ such that all of \cref{eq:equivalence_subset} commutes. Inserting these facts into \cref{eq:equivalence_commuting}, we see that
  \begin{equation}
    \begin{tikzcd}[column sep=huge,row sep=large]
      & Y_1 \ar[d, "\tau\circ t\circ p"] &\\
      X \ar[ur, "f_1"] \ar[r, "\tau\circ t\circ\rho\circ f_2"] \ar[dr, "f_3", '] & Z' & Y \ar[ul, "\sigma_1", '] \ar[l, "\tau\circ\rho\circ\sigma''", '] \ar[dl, "\sigma_3"]\\
      & Y_3 \ar[u, "\tau\circ\rho\circ r"] &
    \end{tikzcd}
  \end{equation}
  commutes and that $\tau\circ\rho\circ\sigma''\in\Sigma$, which proves that $(f_1, \sigma_1)\sim(f_3, \sigma_3)$.
\end{proof}

We write $[f, \sigma]$ for the equivalence class of a fraction $(f, \sigma)$.

\begin{lem}\label{lem:lf2givesequivalent}
  Let morphisms $X' \xleftarrow{\tau} X \xrightarrow{g} Y$ with $\tau\in\Sigma$ be given. If $(f_1, \sigma_1)$ and $(f_2, \sigma_2)$ are two fractions completing the commutative square in \cref{it:lf2}, then $(f_1, \sigma_1)\sim (f_2, \sigma_2)$.
\end{lem}

\begin{proof}
  We are given the commutative diagram
  \begin{equation}\label{eq:two_choices_lf2}
    \begin{tikzcd}
      & Y' &\\
      X' \ar[ur, "f_1"] \ar[dr, "f_2", '] & X \ar[l, "\tau", '] \ar[r, "g"] & Y \ar[ul, "\sigma_1", '] \ar[ld, "\sigma_2"]\\
      & Y'' &
    \end{tikzcd}
  \end{equation}
  and consider the diagram
  \begin{equation}
    \begin{tikzcd}
      & & Y' \ar[ld, dashed, "t", '] &\\
      Z'& Z \ar[l, dashed, "\nu", '] & & Y \ar[ul, "\sigma_1" '] \ar[dl, "\sigma_2"]\\
      & & Y''\nospacepunct{.} \ar[lu, dashed, "\sigma_1'"] &
    \end{tikzcd}
  \end{equation}
  By \cref{it:lf2}, there are $t\in\ccal$ and $\sigma_1'\in\Sigma$ that make the square commute. Since
  \begin{equation}
    t\circ f_1\circ\tau = t\circ\sigma_1\circ g=\sigma_1'\circ\sigma_2\circ g=\sigma_1'\circ f_2\circ\tau,
  \end{equation}
  it follows from \cref{it:lf3} that there exists $\nu\in\Sigma$ such that
  \begin{equation}
    \begin{tikzcd}[column sep=large,row sep=large]
      & Y' \ar[d, "\nu\circ t"] &\\
      X \ar[ur, "f_1"] \ar[r] \ar[dr, "f_2", '] & Z' & Y \ar[ul, "\sigma_1", '] \ar[l] \ar[ld, "\sigma_2"]\\
      & Y'' \ar[u, "\nu\circ\sigma_1'"] &
    \end{tikzcd}
  \end{equation}
  commutes and $\nu\circ t\circ\sigma_1=\nu\circ\sigma_1'\circ\sigma_2\in\Sigma$.
\end{proof}

Given two equivalence classes of fractions $[f, \sigma]$ and $[g, \omega]$ with $\dom(\sigma)=\dom(g)$, their composition $[g, \omega]\circ [f, \sigma]$ is defined to be $[g'\circ f, \sigma'\circ\omega]$, where $g'$ and $\sigma'$ are obtained from \cref{it:lf2} as in the commutative diagram
\begin{equation}
  \begin{tikzcd}
    && Z'' &&\\
    & Y' \ar[ur, "g'"] && Z' \ar[ul, "\sigma'", '] &\\
    X \ar[ur, "f"] && Y \ar[ul, "\sigma", '] \ar[ur, "g"] && Z\nospacepunct{.} \ar[ul, "\omega", ']
  \end{tikzcd}
\end{equation}

\begin{lem}
  The equivalence class of the composition $[g, \omega]\circ [f, \sigma]=[g'\circ f, \sigma'\circ\omega]$ is independent of the choice of $g'$ and $\sigma'$ obtained from \cref{it:lf2} and independent of the chosen representatives $(f, \sigma)$ and $(g, \omega)$.
\end{lem}

\begin{proof}
  The first statement follows directly from \cref{lem:lf2givesequivalent}. For the second statement, assume first that $(f_1, \sigma_1)$ has the expansion $(f', \sigma')$, which means that $f'=p\circ f_1$ and $\sigma'=p\circ\sigma_1$ for some $p\in\ccal$. The composition of both $(f, \sigma)$ and $(f', \sigma')$ with some $[g, \omega]$ is depicted with continuous lines in the diagram
  \begin{equation}\label{eq:composition_welldef_left}
    \begin{tikzcd}
      &&V'&&\\
      & V \ar[ru, "\tau", dashed] &&&\\
      Y' \ar[ru, "\hat g", dashed] && Z'' \ar[lu, "\pi", dashed, '] \ar[uu, dashed, "\tau\circ\pi", '] &&\\
      & Y_1' \ar[lu, "p", '] \ar[ur, "g'"] && Z' \ar[ul, "\sigma_1'", '] &\\
      X \ar[uu, "f'"] \ar[ur, "f_1"] && Y \ar[ul, "\sigma_1", '] \ar[ur, "g"] && Z\nospacepunct{.} \ar[ul, "\omega", ']
    \end{tikzcd}
  \end{equation}
  Since $p\circ\sigma_1\in\Sigma$, \cref{it:lf2} implies that there exist $\hat g\in\ccal$ and $\pi\in\Sigma$ such that $\hat g\circ (p \circ\sigma_1)=\pi\circ (g'\circ\sigma_1)$. This can be regrouped as $(\hat g\circ p) \circ\sigma_1=(\pi\circ g')\circ\sigma_1$, so by \cref{it:lf3} there exists $\tau\in\Sigma$ such that all of \cref{eq:composition_welldef_left} commutes. But from \cref{rem:expansions} it thus follows that $[g, \omega]\circ[f', \sigma']\sim [g, \omega]\circ[f_1, \sigma_1]$.

  Next, we assume that $(g_1, \omega_1)$ has the expansion $(g', \omega')$. As above, the composition of their equivalence classes with some $[f, \sigma]$ is shown in the diagram
  \begin{equation}
    \begin{tikzcd}
      &&&W&\\
      && Z'' \ar[ru, "\hat r", dashed] && Z' \ar[lu, "\rho", dashed, ']\\
      & Y' \ar[ru, "g_1'"] && Z_1' \ar[ul, "\sigma'", '] \ar[ru, "r"] &\\
      X \ar[ur, "f"] && Y \ar[ul, "\sigma", '] \ar[ur, "g_1"] && Z\nospacepunct{.} \ar[ul, "\omega_1", '] \ar[uu, "g'", ']
    \end{tikzcd}
  \end{equation}
  $\hat r\in\ccal$ and $\rho\in\Sigma$ exist by \cref{it:lf2}, from which it follows that $[g', \omega']\circ [f, \sigma]\sim [g_1, \omega_1]\circ [f, \sigma]$.

  Assuming that $(f', \sigma')$ and $(g', \omega')$ are expansions of $(f_1, \sigma_1)$ and $(g_1, \omega_1)$, respectively, we have
  \begin{equation}
    [g_1, \omega_1]\circ [f_1, \sigma_1]\sim [g_1, \omega_1]\circ [f', \sigma']\sim [g', \omega']\circ [f', \sigma'].
  \end{equation}
  Since $\sim$ is transitive, the first and third composition are also equivalent. This means that if $(f_1, \sigma_1)\sim (f_2, \sigma_2)$ have common expansion $(f', \sigma')$ and $(g_1, \omega_1)\sim (g_2, \omega_2)$ have common expansion $(g', \omega')$, we have
  \begin{equation}
    [g_1, \omega_1]\circ [f_1, \sigma_1]\sim [g', \omega']\circ [f', \sigma']\sim [g_2, \omega_2]\circ [f_2, \sigma_2],
  \end{equation}
  which proves the claim after using transitivity again.
\end{proof}

\begin{lem}
  Composition of equivalence classes of fractions is associative.
\end{lem}

\begin{proof}
  This follows directly after three applications of \cref{it:lf2}, which give the dashed arrows in the diagram
  \begin{equation}
    \begin{tikzcd}
      &&& W''' &&&\\
      && Z'' \ar[ru, dashed, "h'''"] && W'' \ar[lu, dashed, "\sigma'''", '] &&\\
      & Y' \ar[ru, dashed, "g'"] && Z' \ar[lu, dashed, "\sigma'", '] \ar[ru, dashed, "h'"] && W' \ar[lu, dashed, "\omega'", '] &\\
      X \ar[ru, "f"] && Y \ar[lu, "\sigma", '] \ar[ru, "g"] && Z \ar[lu, "\omega", '] \ar[ru, "h"] && W\nospacepunct{. \qedhere} \ar[lu, "\kappa", ']
    \end{tikzcd}
  \end{equation}
\end{proof}

Finally, note that for all morphisms $X\xrightarrow{f}Y'\xleftarrow{\sigma}Y$ with $\sigma\in\Sigma$,
\begin{equation}
  [f, \sigma]\circ [{\id_{X}}, {\id_{X}}]=[{\id_{Y}}, {\id_{Y}}]\circ[f, \sigma]=[f, \sigma].
\end{equation}
This allows us to form the promised category $\Sigma^{-1}\ccal$ with the following data:
\begin{itemize}
  \item the \emph{objects} are those of $\ccal$;
  \item the \emph{morphisms} are equivalence classes $[f, \sigma]$ of fractions with $f\in\ccal$ and $\sigma\in\Sigma$;
  \item \emph{composition} of morphisms is given by composition of equivalence classes of fractions;
  \item for every object $X$, the \emph{identity} morphism is $[{\id_X}, {\id_X}]$.
\end{itemize}
We proceed by comparing $\Sigma^{-1}\ccal$ with the localization $\ccal[\Sigma^{-1}]$. To this end, we define the canonical functor $R\colon\ccal\to\Sigma^{-1}\ccal$ which is the identity on objects and maps $f\colon X\to Y$ to $[f, {\id_Y}]$. (It is easy to prove that $R$ is indeed a functor.) We also recall the quotient functor $Q\colon\ccal\to\ccal[\Sigma^{-1}]$ of the localization.

\begin{thm} \label{thm:lociso}
  The functor $F\colon \Sigma^{-1}\ccal\to \ccal[\Sigma^{-1}]$ which is the identity on objects and the map $[f, \sigma]\mapsto Q(\sigma)^{-1}\circ Q(f)$ on morphisms is an isomorphism.
\end{thm}

\begin{proof}
  First we show that $F$ is well-defined. Let $[p\circ f, p\circ \sigma]=[f, \sigma]$. Then
  \begin{equation}
    \begin{split}
      F([p\circ f, p\circ \sigma])&=Q(p\circ \sigma)^{-1}\circ Q(p\circ f)\\
      &=Q(\sigma^{-1})\circ Q(p)^{-1} \circ Q(p) \circ Q(f)\\
      &=Q(\sigma)^{-1}\circ Q(f)\\
      &=F([f, \sigma]).
    \end{split}
  \end{equation}
  Next, we show that $F$ is a functor. Let $[f, \sigma]$ and $[g, \omega]$ be composable morphisms and let $p, q$ be such that $p\circ\sigma=q\circ g$. Then
  \begin{equation}
    \begin{split}
      F([g, \omega]\circ [f, \sigma])&=F([p\circ f, q\circ\omega])\\
      &=Q(q\circ\omega)^{-1}\circ Q(p\circ f)\\
      &=Q(\omega)^{-1}\circ Q(q)^{-1}\circ Q(p) \circ Q(f)\\
      &=Q(\omega)^{-1}\circ Q(g)^{-1}\circ Q(g)\circ Q(q)^{-1}\\
      &\qquad\circ Q(p) \circ Q(\sigma)\circ Q(\sigma)^{-1} \circ Q(f)\\
      &=Q(\omega)^{-1}\circ Q(g)\circ Q(\sigma)^{-1}\circ Q(f)\\
      &=F([g, \omega])\circ F([f, \sigma]).
    \end{split}
  \end{equation}

  Since $R$ makes all elements of $\Sigma$ invertible, there is a functor $G\colon\ccal[\Sigma^{-1}]\to\Sigma^{-1}\ccal$ such that $R=G\circ Q$. Let $[f, \sigma]$ be an equivalence class of fractions with $\cod(f)=\dom(\sigma)=Y$. Then
  \begin{equation}
    G\circ F([f, \sigma])=G(Q(\sigma)^{-1}\circ Q(f))=R(\sigma)^{-1}\circ R(f)=[\id_Y, \sigma]\circ[f, \id_Y]=[f, \sigma],
  \end{equation}
  so $G\circ F=\Id_{\Sigma^{-1}\ccal}$.
  Furthermore, for all $f\colon X\to Y$ in $\ccal[\Sigma^{-1}]$, we have
  \begin{equation}
    F\circ G\circ Q(f)=F\circ R(f)=F([f, \id_Y])=Q(f).
  \end{equation}
  If we apply the universal property \cref{it:lf2} to $Q$, it follows that there exists a unique functor $J\colon\ccal[\Sigma^{-1}]\to\ccal[\Sigma^{-1}]$ such that $Q=J\circ Q$. Since also $Q=\Id_{\ccal[\Sigma^{-1}]}\circ Q=F\circ G\circ Q$, we must have $J=F\circ G=\Id_{\ccal[\Sigma^{-1}]}$ by uniqueness.
\end{proof}

In the sequel, we will always identify $\Sigma^{-1}\ccal$ with $\ccal[\Sigma^{-1}]$ whenever $\Sigma$ admits a calculus of left fractions.

\begin{rem}
  If for any object $X$ we have $\ccal(X, X)\subset\Sigma$, then $\ccal(X, X)$ is a group under composition. In this case, inverses are given by $[f, g]^{-1}=[g, f]$. If $\Sigma=\chom(\ccal)$, then $\ccal[\Sigma^{-1}]$ is a groupoid.
\end{rem}

\section{The Kan Extension of a Localization Functor}
\label{sec:fractrep}

Let $\acal, \bcal, \ssans$ be categories and let $F\colon\acal\to\bcal$ be a functor. If the functor
\begin{equation}
  {\blank\circ F}\colon [\bcal, \ssans] \to [\acal, \ssans]
\end{equation}
has a left adjoint, this left adjoint is called a \defemph{left Kan extension} along $F$.

For the following, note that for each object $B\in\bcal$, there are a comma category $(F\downarrow B)$ and a corresponding projection functor $P_B\colon (F\downarrow B)\to\acal$. The following proposition is taken from \cite[Ex.~6.2.25]{Leinster14}.

\begin{thm}\label{prop:kan}
  Let $F\colon\acal\to\bcal$ be a functor between small categories, and
  let $X\colon\acal\to\ssans$ be a functor. For each object $B\in\bcal$, define a functor $\chi_B$ by
  \begin{equation}
    \chi_B\colon (F\downarrow B)\xrightarrow{P_B}\acal\xrightarrow{X}\ssans.
  \end{equation}
  Assume that $\ssans$ has all colimits of shape $\chi_B$ for every $B$, and let $(\operatorname{Lan}_F X)(B)$ be the colimit of $\chi_B$. Then this defines a functor $\operatorname{Lan}_F X$, and for every other functor $Y\colon\bcal\to\ssans$, there is a canonical bijection between natural transformations $\operatorname{Lan}_F X\to Y$ and natural transformations $X\to Y\circ F$. In this case, $\operatorname{Lan}_F$ is a left Kan extension along $F$.
\end{thm}

The proof can be found in \cite[Thm.~3.7.2]{Borceux1994}. We do not explain it here; instead, it is most helpful to carefully write out all parts of the definitions and statements in \cref{prop:kan}.

The objects of $(F\downarrow B)$ are pairs $(A, h\colon F(A)\to B)$, and a morphism from $(A, h\colon F(A)\to B)$ to $(A', h'\colon F(A')\to B)$ is a morphism $f\colon A\to A'$ in $\acal$ such that
\begin{equation}
  \begin{tikzcd}[row sep=small]
    F(A) \ar[rd, "h"] \ar[dd, "F(f)", '] & \\
    & B\\
    F(A') \ar[ru, "h'", '] &
  \end{tikzcd}
\end{equation}
commutes. Then $\chi_B$ acts as
\begin{align}
  (A, h\colon F(A)\to B) &\mapsto X(A) && \text{on objects,}\\
  (f\colon A\to A') &\mapsto X(f) && \text{on morphisms.}
\end{align}
A cocone from $\chi_B$ to $S_B=(\operatorname{Lan}_F X)(B)$ is a family of morphisms
\begin{equation}
  \eta^B_{(A, h)}\colon X(A)\to S_B, \quad (A, h)\in\obj\bigl((F\downarrow B)\bigr),
\end{equation}
such that
\begin{equation}
  \begin{tikzcd}[row sep=small]
    X(A) \ar[rd, "\eta^B_{(A, h)}"] \ar[dd, "X(f)", '] & \\
    & S_B\\
    X(A') \ar[ru, "\eta^B_{(A', h')}", '] &
  \end{tikzcd}
\end{equation}
commutes for every morphism $(f\colon A\to A')\in\acal$ from $(A, h)$ to $(A', h')$ in $(F\downarrow B)$.

From the proof one can see how $\lan_F X$ acts on morphisms. If $g\colon B\to B'$ is in $\bcal$, then
\begin{equation}\label{eq:morphism-lan}
  \begin{tikzcd}
    X(A) \ar[rd,bend right,outer sep=-1.5mm, "\eta^B_{(A, g\circ h)}", '] \ar[r, "\eta^B_{(A, h)}"] & S_B \ar[d, "(\lan_F X)(g)"]\\
    & S_{B'}
  \end{tikzcd}
\end{equation}
commutes for every morphism of the form $h\colon F(A)\to B$.

Let $\ccal$ be small, and assume that $\Sigma\subset\chom(\ccal)$ admits a calculus of left fractions, so that $Q\colon\ccal\to\ccal[\Sigma^{-1}]$ is the localization functor which is the map $f\mapsto [f, \id_{\cod(f)}]$ on morphisms and the identity on objects---though for the sake of clarity we will not evaluate it on objects.

We will now investigate left Kan extensions ${\lan_Q}$ along $Q$. For the moment, we fix an object $B$.

As a first step, we look at $(Q\downarrow B)$. We consider two objects in $(Q\downarrow B)$, $(A, h)$ and $(A', h')$, where $h$ and $h'$ can in general be written as
\begin{align}
  h&=[h_1\colon F(A)\to W, h_2\colon B\to W],\\
  h'&=[h'_1\colon F(A')\to W', h'_2\colon B\to W']
\end{align}
for some $W$ and $W'$. A morphism $(A, h)\to (A', h')$ consists of a morphism $f\colon A\to A'$ in $\ccal$ such that
\begin{equation}\label{eq:special-comma-morphism}
  h'\circ[f, \id_{A'}]=h.
\end{equation}

We assume that we again have a functor $X\colon\ccal\to\ssans$ which we want to extend. The colimit for $B$ consists of a family of morphisms $\eta^B_{(A, h)}\colon X(A)\to S_B$ indexed by objects $(A, h)=(A, h\colon F(A)\to B)$. We  consider the subset of objects of the form
\begin{equation}\label{eq:special-object}
  \bigl(A, [\id_{Q(A)}, \epsilon\colon B\to Q(A)]\colon Q(A)\to B  \bigr),
\end{equation}
or, in short, $(A, [\id_{Q(A)}, \epsilon])$. We now look at \cref{eq:morphism-lan} for the special case of the cocone indexed by objects as in \cref{eq:special-object}. In this case, $h=[\id_{Q(A)}, \epsilon]$, and we can compute
\begin{equation}
  g\circ h = [g_1, g_2]\circ [\id_{Q(A)}, \epsilon]=[g_1', \epsilon'\circ g_2],
\end{equation}
where we have written $g=[g_1, g_2]$, and where $g_1'\circ \epsilon=\epsilon'\circ g_1$. We can further trivially write
\begin{equation}
  [g_1', \epsilon'\circ g_2]= [\id_Z, \epsilon'\circ g_2]\circ [g_1', \id_Z],
\end{equation}
where $Z$ is the codomain of $g_1'$ and $\epsilon'$. In this form, we see the resemblance with the left-hand side of \cref{eq:special-comma-morphism}. Therefore we arrive at the final equation
\begin{equation}
  \eta^B_{(A, g\circ h)}=\eta^B_{(Z, [\id_Z, \epsilon'\circ g_2])} \circ X(g_1').
\end{equation}
In summary, we have
\begin{equation}
  (\lan_F X)([g_1, g_2])\circ\eta^B_{(A, [\id_{Q(A)}, \epsilon])} = \eta^B_{(Z, [\id_Z, \epsilon'\circ g_2])} \circ X(g_1').
\end{equation}
This allows us to easily compute $(\lan_F X)(g)$ in applications.

\begin{rem}
  Jones' representations are a special case of this, where $\ccal$ needs to be replaced with $\ccal^\mathrm{op}$, the category $\ccal$ with all morphism reversed.
\end{rem}


\chapter{Thompson's Groups $F$ and $T$}
\label{cha:thompson}

In this chapter, we introduce two of Thompson's groups, $F$ and $T$. They can be described as groups of certain homeomorphisms of the unit interval $[0, 1]$ and the circle $\sircle^1$, respectively. Another viewpoint is that they arise from the localizations of certain forest categories, which we carefully introduce first.

In \cref{part:applications} of this thesis, $F$ and $T$ will take the roles of \emph{discrete analogues} of the groups of orientation-preserving diffeomorphisms of the interval and the circle, respectively. To partly justify our choice, we show how elements of $F$ and $T$ can be used to approximate diffeomorphisms in the final section of this chapter.

We will often follow a combination of the standard reference \cite{CannonFloydParry96} by \citeauthor{CannonFloydParry96}, and Belk's thesis \cite{Belk2004}. A very helpful characterization of Thompson's group $F$ can be found in \cite{FioreLeinster2010}. The final section is a version of \cite{Stiegemann2018b}.

\section{$F$ and $T$ as Groups of Homeomorphisms}

Dyadic rationals are all numbers of the form $m/2^k$ with $m\in\integers$ and $k\in\naturals$. By a breakpoint of a piecewise linear function we mean the points at which it is not differentiable.

\begin{defn}
  Thompson's group $F$ is the group of piecewise linear homeomorphisms $g$ of the closed unit interval $[0, 1]$ such that
  \begin{genenum}[label=(Th$_\arabic*$)]
    \item\label{item:breakpointsF} the breakpoints of $g$ are dyadic rationals;
    \item\label{item:slopesF} on intervals of differentiability, the derivatives of $g$ are integer powers of $2$.
  \end{genenum}
\end{defn}

For the definition of $T$, we consider the circle $\sircle^1$ as the unit interval $[0, 1]$ with the endpoints $0$ and $1$ identified.

\begin{defn}
  Thompson's group $T$ is the group of piecewise linear homeomorphisms $g$ of $\sircle^1$ such that
  \begin{genenum}[label=(Th$_\arabic*$)]
    \item[(Th$_1'$)]\label{item:breakpointsT} the breakpoints of $g$ and their images are dyadic rationals;
    \item[(Th$_2$)]\label{item:slopesT} on intervals of differentiability, the derivatives of $g$ are integer powers of $2$.
  \end{genenum}
\end{defn}

A \defemph{standard dyadic interval} is an interval of the form
\begin{equation}
  \biggl[\frac{k}{2^n},\frac{k+1}{2^n}\biggr], \quad k, n\in\naturals.
\end{equation}
A partition of $[0, 1]$ is called a \defemph{standard dyadic partition} if all intervals in the partitions are standard dyadic intervals. Every standard dyadic partition can be obtained by first taking the undivided interval $[0, 1]$ and then successively cutting intervals in two equally sized halves; for instance, the standard dyadic partition
\begin{equation}
	\begin{tikzpicture}[scale=5]
		\def\h{0.025}
		\draw (0, 0) -- (1, 0);
		\draw (0, -1.5*\h) -- ++(0, 3*\h)  node[at start,below=2pt] {$0$};
		\draw (0.25, -\h) -- ++(0, 2*\h) node[at start,below=3pt] {$\frac{1}{4}$};
		\draw (0.375, -\h) -- ++(0, 2*\h) node[at start,below=3pt] {$\frac{3}{8}$};
		\draw (0.5, -\h) -- ++(0, 2*\h) node[at start,below=3pt] {$\frac{1}{2}$};
		\draw (1, -1.5*\h) -- ++(0, 3*\h) node[at start,below=2pt] {$1$};
	\end{tikzpicture}
\end{equation}
is obtained by successively performing the following three cuts:
\begin{equation}
	\begin{tikzpicture}[scale=5]
		\def\h{0.025}
    \begin{scope}[yshift=0]
      \draw (0, 0) -- (1, 0);
      \draw (0, -1.5*\h) -- ++(0, 3*\h);
      \draw (1, -1.5*\h) -- ++(0, 3*\h);
    \end{scope}
    \begin{scope}[yshift=-3]
      \draw (0, 0) -- (1, 0);
  		\draw (0, -1.5*\h) -- ++(0, 3*\h);
  		\draw[very thick] (0.5, -\h) -- ++(0, 2*\h);
  		\draw (1, -1.5*\h) -- ++(0, 3*\h);
    \end{scope}
    \begin{scope}[yshift=-6]
      \draw (0, 0) -- (1, 0);
  		\draw (0, -1.5*\h) -- ++(0, 3*\h);
  		\draw[very thick] (0.25, -\h) -- ++(0, 2*\h);
  		\draw (0.5, -\h) -- ++(0, 2*\h);
  		\draw (1, -1.5*\h) -- ++(0, 3*\h);
    \end{scope}
    \begin{scope}[yshift=-9]
      \draw (0, 0) -- (1, 0);
  		\draw (0, -1.5*\h) -- ++(0, 3*\h);
  		\draw (0.25, -\h) -- ++(0, 2*\h);
  		\draw[very thick] (0.375, -\h) -- ++(0, 2*\h);
  		\draw (0.5, -\h) -- ++(0, 2*\h);
  		\draw (1, -1.5*\h) -- ++(0, 3*\h);
    \end{scope}
	\end{tikzpicture}
\end{equation}
Every element $g$ of $F$ can be described by a pair $P, Q$ of standard dyadic partitions with the same number of cuts, where $g$ sends each interval of $P$ linearly onto the corresponding interval in $Q$. These are called ``dyadic rearrangements'' in \cite{Belk2004}. Elements of $T$ have a similar description when we replace partitions of the interval with partitions of the circle.

As is proved in \cite{CannonFloydParry96}, $F$ and $T$ are finitely presented infinite groups. The generators are shown in \cref{fig:generators}.

\begin{figure}
  \begin{subfigure}[b]{0.3\linewidth}
    \begin{tikzpicture}[scale=1]
      \footnotesize

      \def\base{0.7}

      \draw[gray] (2*\base, 0) to ++(0, 4*\base);
      \draw[gray] (3*\base, 0) to ++(0, 4*\base);
      \draw[gray] (4*\base, 0) node[black,below] {$1$} to ++(0, 4*\base);

      \draw[gray] (0, 1*\base) to ++(4*\base, 0);
      \draw[gray] (0, 2*\base) to ++(4*\base, 0);
      \draw[gray] (0, 4*\base) node[black,left] {$1$} to ++(4*\base, 0);

      \draw[->] (0, 0) -- (4.5*\base, 0);
      \draw[->] (0, 0) -- (0, 4.5*\base);

      \node[below left] at (0, 0) {$0$};

      \draw (0, 0*\base) -- (2*\base, 1*\base) -- (3*\base, 2*\base) -- (4*\base, 4*\base);

      \node at (2*\base, -1*\base) {$A$};
    \end{tikzpicture}
  \end{subfigure}
  \begin{subfigure}[b]{0.3\linewidth}
    \begin{tikzpicture}[scale=1]
      \footnotesize

      \def\base{0.7}

      \draw[gray] (2*\base, 0) to ++(0, 4*\base);
      \draw[gray] (3*\base, 0) to ++(0, 4*\base);
      \draw[gray] (3.5*\base, 0) to ++(0, 4*\base);
      \draw[gray] (4*\base, 0) node[black,below] {$1$} to ++(0, 4*\base);

      \draw[gray] (0, 2*\base) to ++(4*\base, 0);
      \draw[gray] (0, 2.5*\base) to ++(4*\base, 0);
      \draw[gray] (0, 3*\base) to ++(4*\base, 0);
      \draw[gray] (0, 4*\base) node[black,left] {$1$} to ++(4*\base, 0);

      \draw[->] (0, 0) -- (4.5*\base, 0);
      \draw[->] (0, 0) -- (0, 4.5*\base);

      \node[below left] at (0, 0) {$0$};

      \draw (0, 0*\base) -- (2*\base, 2*\base) -- (3*\base, 2.5*\base) -- (3.5*\base, 3*\base) -- (4*\base, 4*\base);

      \node at (2*\base, -1*\base) {$B$};
    \end{tikzpicture}
  \end{subfigure}
  \begin{subfigure}[b]{0.3\linewidth}
    \begin{tikzpicture}[scale=1]
      \footnotesize

      \def\base{0.7}

      \draw[gray] (2*\base, 0) to ++(0, 4*\base);
      \draw[gray] (3*\base, 0) to ++(0, 4*\base);
      \draw[gray] (4*\base, 0) node[black,below] {$1$} to ++(0, 4*\base);

      \draw[gray] (0, 2*\base) to ++(4*\base, 0);
      \draw[gray] (0, 3*\base) to ++(4*\base, 0);
      \draw[gray] (0, 4*\base) node[black,left] {$1$} to ++(4*\base, 0);

      \draw[->] (0, 0) -- (4.5*\base, 0);
      \draw[->] (0, 0) -- (0, 4.5*\base);

      \node[below left] at (0, 0) {$0$};

      \draw (0, 3*\base) -- (2*\base, 4*\base);
      \draw (2*\base, 0) -- (3*\base, 2*\base) -- (4*\base, 3*\base);

      \node at (2*\base, -1*\base) {$C$};
    \end{tikzpicture}
  \end{subfigure}
  \caption{The three generators $A$, $B$, and $C$ of $T$. $A$ and $B$ generate $F$.}
  \label{fig:generators}
\end{figure}
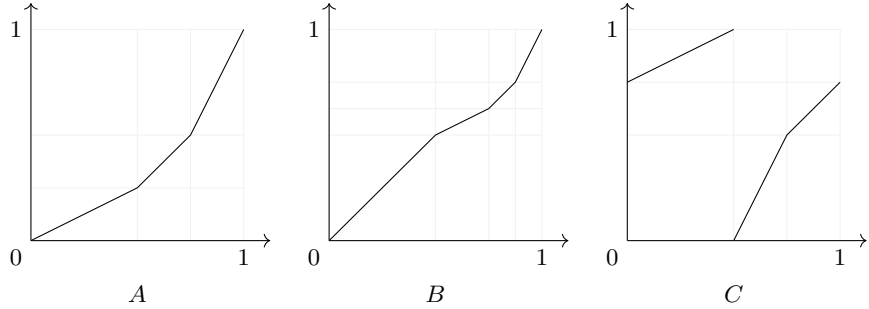

\section{The Category of Binary Forests}
\label{sec:category-binary-forests}

In order to give a rigorous definition of the category of binary forests, we first review the necessary facts from graph theory. Here we follow a standard textbook by \citeauthor{Diestel17} \cite{Diestel17}. Afterwards we can define the category of binary forests and record some important properties. Finally, we use these properties to form a localization of this category, which will give us Thompson's groupoid and Thompson's group.

\subsection{Graphs, Trees, and Forests}
\label{sub:graphs-trees-forests}

A \defemph{graph} is a pair $G=(V, E)$ of sets such that $E$ consists of $2$-element subsets of $V$. The elements of $V$ are the \defemph{vertices} of the graph $G$, and the elements of $E$ are its \defemph{edges}. We will write $xy$ for an edge $\{x, y\}$. In general, for any graph $G$, we write $V(G)$ for the set of vertices of $G$ and $E(G)$ for the set of edges. For any vertex $x$, the \defemph{degree} $\deg(x)$ of $x$ denotes the number of distinct edges of which $x$ is a member; a vertex with degree $d$ is called a $d$-valent vertex. If $G=(V, E)$ and $G'=(V', E')$ are two graphs with $V'\subset V$ and $E'\subset E$, then $G'$ is a \defemph{subgraph} of $G$, written $G'\subset G$. A \defemph{path} is a non-empty graph $P=(V, E)$ of the form
\begin{equation}
  V=\{x_0, x_1, \dotsc, x_k\}, \quad E=\bigl\{\{x_0, x_1\}, \{x_1, x_2\}, \dotsc, \{x_{k-1}, x_k\}\bigr\},
\end{equation}
which we write as $P=x_0x_1\dotsm x_k$ (in accordance with the notation for edges). The number of edges of a path is its \defemph{length}. A graph $C$ is called a \defemph{cycle} if there is a path $P=x_0\dotsm x_{k}$, $k\ge 2$, such that $V(C)=V(P)$ and $E(C)=E(P)\cup\{x_kx_0\}$.

A graph $G$ is called \defemph{connected} if it is non-empty and for any two vertices $x, y$ there exists a path $x\dotsm y$ in $G$. For a general graph $G$, any maximal connected subgraph of $G$ (with respect to inclusion of graphs) is called a \defemph{connected component} of $G$.

An acyclic graph (not containing any cycles) is called a \defemph{forest}. A connected forest is called a \defemph{tree}. Therefore, a forest is a (possibly empty) disjoint union of trees, which are its connected components.

A \defemph{rooted tree} is a tree with at least $2$ vertices and with a designated $1$-valent vertex called the \defemph{root}; all other $1$-valent vertices of the tree are its \defemph{leaves}. The number of leaves of a rooted tree $t$ is denoted $\abs{t}$.
If a vertex $x$ immediately precedes a vertex $y$ on the path from the root to $y$, then $x$ is the \defemph{parent} of $y$ and $y$ is a \defemph{child} of $x$.

\begin{defn}
  A \defemph{binary tree} is a rooted tree without $2$-valent vertices such that every vertex has at most $2$ children, and the two children of every trivalent\footnote{Trivalent is another word for $3$-valent.} vertex are assigned a fixed order.
\end{defn}

If the vertex $x$ of a tree has children $y_l$ and $y_r$ with the given order $y_l\le y_r$, then $y_l$ is called the \defemph{left child} of $x$ and $y_r$ is called the \defemph{right child}.

\begin{rem}
  To be consistent with the graphical calculus introduced in \cref{sec:graphical-calculus}, we draw binary trees from bottom to top, that is, the roots are at the bottom and the leaves at the top (just like in nature), for example
  \begin{equation}
    \begin{tikzpicture}[scale=0.25, yscale=-1]
      \draw (3, 3) -- (3, 4);
  		\draw (2.0, 0.0) -- (3.0, 1.0) -- (4.0, 0.0);
  		\draw (0.0, 0.0) -- (2.0, 2.0) -- (3.0, 1.0);
  		\draw (2.0, 2.0) -- (3.0, 3.0) -- (6.0, 0.0);
  	\end{tikzpicture}\,.
  \end{equation}
\end{rem}

\begin{exmp}
  A very simple binary tree that we will need later is the \defemph{caret}
  \begin{equation}
    \begin{tikzpicture}[scale=0.3, yscale=-1]
      \draw (0, 2) -- (0, 1);
      \draw (-1, 0) -- (0, 1) -- (1, 0);
    \end{tikzpicture}\, .
  \end{equation}
\end{exmp}

On the set of all leaves of a binary tree, we define a linear order as follows. For the binary tree $t$ with exactly one leaf $l$, we set $l\preceq l$. Moreover, let $l_1, l_2$ be two leaves in a binary tree $t$ with $\abs{t}\ge 1$. If $l_1=l_2$, we set $l_1\preceq l_2$ and $l_2\preceq l_1$. Otherwise, let $x_0x_1\dotsm x_m$ and $y_0y_1\dotsm y_n$ be the shortest paths connecting the root $r=x_0=y_0$ of $t$ with $l_1=x_m$ and $l_2=y_n$, respectively. Let $i$ be the smallest number such that $x_i\neq y_i$, and note that we have $2\le i\le \min\{m, n\}$. If $x_i$ is the left child of $x_{i-1}=y_{i-1}$ and $y_i$ the right child, we set $l_1\preceq l_2$; otherwise, we set $l_2\preceq l_1$. It is clear that $\preceq$ is antisymmetric and that all leaves of a tree are comparable under $\preceq$. It remains to be shown that $\preceq$ is transitive. To this end, let $l_1\preceq l_2$ and $l_2\preceq l_3$, and assume without loss of generality that $l_1\neq l_2\neq l_3$. Let $x_0\dotsm x_k$, $y_0\dotsm y_m$, and $z_0\dotsm z_n$ be the root--leaf paths as before, and let $i$ be the number defined above for $l_1$ and $l_2$, and $j$ the corresponding number for $l_2$ and $l_3$. Since $l_1\neq l_2\neq l_3$, we must have $i\neq j$. It is also clear that $i'=\min\{i, j\}$ is the smallest number such that $x_{i'}\neq z_{i'}$. If $i'=i$, then $l_1\preceq l_3$ now follows from $l_1\preceq l_2$, and if $i'=j$, it follows from $l_2\preceq l_3$.

Let $t_1$ and $t_2$ be two binary trees, let $l$ be a leaf of $t_1$ and $r$ the root of $t_2$. We can \defemph{connect $t_2$ to $t_1$ at $l$} to form a new binary tree $t_3$ in the following way. Let $rx_0$ be the edge from $r$ to its only child $x_0$, and let $pl$ be the edge to $l$ from its parent. Then $t_3$ has the underlying graph given by
\begin{align}
  V(t_3)&=\bigl(V(t_1)\cup V(t_2)\bigr)\setminus\{l, r\},\\
  E(t_3)&=\bigl(E(t_1)\cup E(t_2)\cup\{px_0\}\bigr)\setminus\{pl, rx_0\}.
\end{align}
The root of $t_3$ is the root of $t_1$, and the ordering among the children of $p$ in $t_3$ is such that $x_0$ is the left child of $p$ in $t_3$ if $l$ is the left child of $p$ in $t_1$, and $x_0$ is the right child of $p$ in $t_3$ if $l$ is the right child of $p$ in $t_1$.

\begin{defn}
  A \defemph{binary forest} is a forest whose connected components are binary trees, together with a linear order of these trees.
\end{defn}

The leaves of all binary trees in a binary forest, called the leaves of the forest, can be linearly ordered according to the order of trees in the forest and the order of leaves in each tree.

Now let $w=(t_1, \dotsc, t_k)$ be a binary forest consisting of $k$ binary trees $t_1, \dotsc, t_k$. The \defemph{domain} $\dom(w)$ of $w$ is the number $k$ of roots, that is, connected components; the \defemph{codomain} of $w$ is
\begin{equation}
  \cod(w)=\sum_i\, \abs{t_i},
\end{equation}
the number of leaves in $w$. The \defemph{composition} $w_1\circ w_2$ of two binary forests $w_1, w_2$ with $\cod(w_1)=\dom(w_2)$ is the binary forest obtained by connecting the $i$-th tree of $w_2$ with the $i$-th leaf of $w_1$ for all possible $i$.

Composition of forests is associative. Furthermore, for every $n\ge 1$ the unique forest ${\id_n}$ with $\dom({\id_n})=\cod({\id_n})=n$, consisting of $n$ copies of the tree with one leaf, is a left and right identity of composition. We also define the \defemph{empty (binary) forest} $\emptyset$ to be the graph with no vertices, no edges, and $\dom(\emptyset)=\cod(\emptyset)=0$, such that $\emptyset\circ\emptyset=\emptyset$. It cannot be composed with any other forests, and it serves as the identity $\id_0$.

We can define a second operation on pairs of binary forests, the monoidal product $\otimes$. Given two arbitrary binary forests $w_1$ and $w_2$, $w_1\otimes w_2$ has as underlying graph the (disjoint) union of $w_1$ and $w_2$ (that is, take the respective disjoint unions of vertices and edges); the roots and leaves are ordered such that those of $w_1$ come, in order, before those of $w_2$. Graphically, $w_1\otimes w_2$ simply corresponds to $w_1$ and $w_2$ drawn side by side. Note that
\begin{align}
  \dom(w_1\otimes w_2)&=\dom(w_1) + \dom(w_2),\\
  \cod(w_1\otimes w_2)&=\cod(w_1) + \cod(w_2),
\end{align}
and that $\emptyset$ is the neutral element for $\otimes$.

An isomorphism of binary forests $w_1$ and $w_2$ is a bijection $V(w_1)\to V(w_2)$ that preserves adjacency, maps roots to roots, and that respects the order of the trees and the order of the children of each vertex. When in the following we speak of binary trees or forests, we always mean isomorphism classes of binary trees or forests. This means that only the shape of the underlying graphs and the ordering of roots and leaves are important, and not the particular labelling of vertices.

\begin{defn}
  The \defemph{category of binary forests}, $\binfor$, has as objects the naturals numbers $\naturals$, and the morphisms $m\to n$ are all binary forests with $\dom(w)=m$ and $\cod(w)=n$. Composition is given by composition of forests, and for every object $n$ the forest ${\id_n}$ is the identity. $\binfor$ admits a strict monoidal structure; the product is given on objects by addition and on morphisms by $\otimes$; the monoidal unit is $0$.
\end{defn}

\subsection{Properties of Composition}
\label{sub:properties-composition}

We first focus on properties of $\binfor$ as a category, ignoring the monoidal structure. The aim is to establish a useful calculus of fractions on $\binfor$ from which Thompson's group $F$ will emerge.

We define a wide subcategory $\binfor_1$ of $\binfor$ with objects and morphisms
\begin{equation}
  \obj(\binfor_1)=\naturals^\times, \quad \chom(\binfor_1)=\chom(\binfor)\setminus\{\emptyset\}.
\end{equation}

\begin{rem}
  $\binfor_1$ is actually not needed for the construction in this section; since $\emptyset=\id_0$ is already invertible in $\binfor$, there is no difference between including or not including it in the localization procedure. However, we will later define categories similar to $\binfor$, and when localizing these categories, it will be important to exclude certain morphisms from localization, among them the empty forest. Therefore, the above convention helps by highlighting similarities with later cases.
\end{rem}

\begin{lem}
  $\binfor$ has all pushouts.
\end{lem}

\begin{proof}
  It suffices to prove that every diagram of the form $m\xleftarrow{s} 1 \xrightarrow{t} n$ has a pushout, where the forests $s$ and $t$ are trees. Indeed, let $s\vee t$ be the smallest tree containing both $s$ and $t$ as rooted subtrees (that is, roots are mapped to roots). Then
  \begin{equation}
    s\circ (\tau_1, \dotsc, \tau_m) = s\vee t = t\circ (\sigma_1, \dotsc, \sigma_n)
  \end{equation}
  for unique trees $\tau_i$ and $\sigma_j$. It is also clear that every tree containing $s\vee t$ factors uniquely through the $\tau_i$ and $\sigma_j$.
\end{proof}

\begin{cor}\label{cor:all-pushouts}
  For every pair of morphisms $m' \xleftarrow{a} m \xrightarrow{w} n$ with $a\in\binfor_1$ and $w\in\binfor$, there are morphisms $b\colon n\to n'$ in $\binfor_1$ and $v\colon n'\to m'$ in $\binfor$ such that $b\circ w=v\circ a$.
\end{cor}

It is easy to see that the following cancellation property holds.

\begin{lem}\label{lem:cancellation}
  Let $v, w\colon m\to n$ be parallel morphisms in $\binfor$. If there is a morphism $b\colon k\to m$ with $v\circ b=w\circ b$, then $v=w$.
\end{lem}

\subsection{Localization}

As \cref{cor:all-pushouts,lem:cancellation} show, the set $\Sigma=\chom(\binfor_1)$ satisfies \cref{it:lf1,it:lf2,it:lf3} from \cref{defn:calculus-fractions}. Therefore $\Sigma$ admits a calculus of left fractions and the localization $\fcal$ of $\binfor$ by $\Sigma$ exists. $\fcal$ is a groupoid called \defemph{Thompson's groupoid}, a choice of terminology that will become clear very soon.

Consider the monoidal category $\fcal'$ with $\obj(\fcal')=\naturals$ and whose morphisms $m\to n$ are bijections $g\colon [0, m]\to [0, n]$ such that $g$ is piecewise linear, has breakpoints only at dyadic rationals, and on intervals of differentiability, the derivatives of $g$ are integer powers of $2$. (Note that we have one morphism $\emptyset\colon 0\to 0$.) Composition is given by composition of functions. The monoidal product is given on objects by addition and on morphism by juxtaposition of functions. Evidently, $\fcal'$ is a groupoid, and from the description it is clear that $\fcal$ and $\fcal'$ are isomorphic, and that $\fcal'(1, 1)=F$, see also \cite{FioreLeinster2010}.

As an illustration of these facts, note that there is an obvious one-to-one correspondence between standard dyadic partitions of $[0, 1]$ and binary trees. For example, the partition
\begin{equation}
	\begin{tikzpicture}[scale=5]
		\def\h{0.025}
		\draw (0, 0) -- (1, 0);
		\draw (0, -2*\h) -- ++(0, 4*\h)  node[at start,below=2pt] {$0$};
		\draw (0.25, -\h) -- ++(0, 2*\h) node[at start,below=3pt] {$\frac{1}{4}$};
		\draw (0.375, -\h) -- ++(0, 2*\h) node[at start,below=3pt] {$\frac{3}{8}$};
		\draw (0.5, -\h) -- ++(0, 2*\h) node[at start,below=3pt] {$\frac{1}{2}$};
		\draw (1, -2*\h) -- ++(0, 4*\h) node[at start,below=2pt] {$1$};
	\end{tikzpicture}
\end{equation}
is represented by the tree
\begin{equation}
  \begin{tikzpicture}[scale=0.25, yscale=-1]
    \draw (3, 3) -- (3, 4);
    \draw (2.0, 0.0) -- (3.0, 1.0) -- (4.0, 0.0);
    \draw (0.0, 0.0) -- (2.0, 2.0) -- (3.0, 1.0);
    \draw (2.0, 2.0) -- (3.0, 3.0) -- (6.0, 0.0);
  \end{tikzpicture}\,.
\end{equation}
All standard dyadic partitions can therefore be described as finite rooted subtrees of the infinite binary tree of standard dyadic intervals \cite{Belk2004}.
As elements of $F$ can be written as fractions of binary trees, they can as well be written as fractions of standard dyadic partitions. Extending a fraction is done by adding or removing opposing carets from both trees in a fraction, and every element has a so-called \emph{reduced} fraction representative from which no opposing pair of carets can be removed without changing the equivalence class.

\begin{exmp}
  The two generators of $F$, $A$ and $B$, have the following representations as fractions:
  \begin{equation}
    A=
    \begin{tikzpicture}[scale=0.25,baseline=-1mm]
      \begin{scope}[yscale=-1, yshift=-4cm]
        \draw (2.0, 0.0) -- (3.0, 1.0) -- (4.0, 0.0);
        \draw (0.0, 0.0) -- (2.0, 2.0) -- (3.0, 1.0);
        \draw (2, 2) -- (2, 3);
  		\end{scope}

  		\draw (-0.5, 0) -- (4.5, 0);

  		\begin{scope}[yscale=-1, yshift=1cm]
        \draw (0.0, 0.0) -- (1.0, 1.0) -- (2.0, 0.0);
        \draw (1.0, 1.0) -- (2.0, 2.0) -- (4.0, 0.0);
        \draw (2, 2) -- (2, 3);
  		\end{scope}
  	\end{tikzpicture}
    \,,\qquad B=
    \begin{tikzpicture}[scale=0.25,baseline=-1mm]
      \begin{scope}[yscale=-1, yshift=-5cm]
        \draw (4.0, 0.0) -- (5.0, 1.0) -- (6.0, 0.0);
        \draw (2.0, 0.0) -- (4.0, 2.0) -- (5.0, 1.0);
        \draw (0.0, 0.0) -- (3.0, 3.0) -- (4.0, 2.0);
        \draw (3, 3) -- (3, 4);
  		\end{scope}

  		\draw (-0.5, 0) -- (6.5, 0);

  		\begin{scope}[yscale=-1, yshift=1cm]
  			\draw (2.0, 0.0) -- (3.0, 1.0) -- (4.0, 0.0);
  			\draw (3.0, 1.0) -- (4.0, 2.0) -- (6.0, 0.0);
  			\draw (0.0, 0.0) -- (3.0, 3.0) -- (4.0, 2.0);
        \draw (3, 3) -- (3, 4);
  		\end{scope}
  	\end{tikzpicture}
    \,.
  \end{equation}
  Note that equality holds only when we assume the two fractions to represent equivalence classes of fractions as defined for the localization. We can extend a fraction with carets, so that
  \begin{equation}
    A=\begin{tikzpicture}[scale=0.25,baseline=-1mm]
      \begin{scope}[yscale=-1, yshift=-5cm]
        \draw[very thick,rounded corners=0.1mm] (2.0, 0.0) -- (3.0, 1.0) -- (4.0, 0.0);
        \draw (3.0, 1.0) -- (4.0, 2.0) -- (6.0, 0.0);
        \draw (0.0, 0.0) -- (3.0, 3.0) -- (4.0, 2.0);
        \draw (3, 3) -- (3, 4);
  		\end{scope}

  		\draw (-0.5, 0) -- (6.5, 0);

  		\begin{scope}[yscale=-1, yshift=1cm]
        \draw[very thick,rounded corners=0.1mm] (2.0, 0.0) -- (3.0, 1.0) -- (4.0, 0.0);
        \draw (0.0, 0.0) -- (2.0, 2.0) -- (3.0, 1.0);
        \draw (2.0, 2.0) -- (3.0, 3.0) -- (6.0, 0.0);
        \draw (3, 3) -- (3, 4);
  		\end{scope}
  	\end{tikzpicture}\,,
  \end{equation}
  where the added carets are marked in bold.
\end{exmp}

\begin{rem}
  Note that in the literature, trees are usually draw the other way round, and in fractions of trees, the denominator tree is additionally inverted. For consistency within this work, we do not follow this convention.
\end{rem}

\section{The Category of Annular Forests}

This description is adopted from \cite{GrahamLehrer1998,BrothierStottmeister2019}. If $w=(w_0, w_1, \dotsc, w_{m-1})=(w_i)_{i\in\integers_m}$ is a forest whose trees are indexed by the elements of the cyclic group $\integers_m$, we can define, for every $k\in\integers_m$, a new forest
\begin{equation}
  [k](w)=(w_{i+k \Mod m}) = (w_k, w_{k+1\Mod m}, \dotsc, w_{k-1\Mod m}).
\end{equation}

\begin{defn}
  The \defemph{category of annular forests},\footnote{It should be called `category of annular binary forests'. We drop the `binary'.} $\annfor$, has objects and morphisms
  \begin{equation}
    \obj(\annfor)=\naturals, \quad
    \annfor(m, n) = \binfor(m, n)\times\integers_n,
  \end{equation}
  and given $(v, k)\in\annfor(m, n)$ and $(w, l)\in\annfor(n, p)$, composition is defined by
  \begin{equation}
    (w, l)\circ (v, k)=([k](w)\circ v, s(w, l, k)),
  \end{equation}
  where $s(w, l, k)$ is obtained by shifting the trees of $w$ by $k$ steps and counting how many steps the $l$-th leaf of $w$ was moved.
\end{defn}

Similarly to the previous section, we can localize $\annfor$ by all of its morphisms and obtain \defemph{Thompson's groupoid} $\tcal$. The fundamental group $\tcal(1, 1)$ is (isomorphic to) Thompson's group $T$. Tree diagrams are almost the same as for $F$, but pairs of carets are inserted according to the cyclic permutations of the leaves.

\begin{exmp}
  The third generator of $T$ has the tree diagram
  \begin{equation}
    C=
    \begin{tikzpicture}[scale=0.25,baseline=-1mm]
      \begin{scope}[yscale=-1, yshift=-4cm]
        \draw (2.0, 0.0) -- (3.0, 1.0) -- (4.0, 0.0);
        \draw (0.0, 0.0) -- (2.0, 2.0) -- (3.0, 1.0);
        \draw (2, 2) -- (2, 3);
        \draw[fill=white] (2, 0) circle[radius=4mm];
  		\end{scope}

  		\draw (-0.5, 0) -- (4.5, 0);

  		\begin{scope}[yscale=-1, yshift=1cm]
        \draw (2.0, 0.0) -- (3.0, 1.0) -- (4.0, 0.0);
        \draw (0.0, 0.0) -- (2.0, 2.0) -- (3.0, 1.0);
        \draw (2, 2) -- (2, 3);
        \draw[fill=white] (0, 0) circle[radius=4mm];
  		\end{scope}
  	\end{tikzpicture}\,.
  \end{equation}
  The two white circles indicate which leaves belong together according to the cyclic permutations.
\end{exmp}

\section{Approximating Diffeomorphisms}
\label{sec:approximating-diffeomorphisms}

Finally, we show how diffeomorphisms of the interval $[0, 1]$ and the circle $\sircle^1$ can be approximated by elements of Thompson's groups $F$ and $T$, respectively.

Let $\diff^1_+([0, 1])$ denote the group of orientation-preserving $C^1$-diffeomor\-phisms of the interval, and similarly for $\sircle^1$. Our result is stated in terms of the $C^0$-norm
\begin{equation}
  \norm{f}=\sup_x\: \abs{f(x)}.
\end{equation}

\begin{thm}\label{thm:diffapprox}
  For every $f\in\diff^1_{+}([0, 1])$ and $\epsilon>0$, there exists $g\in F$ such that $\norm{f-g}<\epsilon$. Similarly, if $f\in\diff^1_{+}(\sircle^1)$, then there exists $g\in T$ with this property.
\end{thm}

This statement is known and follows from \cite[Thm.~A4.1]{BieriStrebel16} and \cite[Prop.~4.3]{Zhuang07}. It is actually true for all orientation-preserving homeomorphisms, not only diffeomorphisms.

\begin{figure}
\begin{subfigure}[b]{0.3\linewidth}
  \begin{tikzpicture}
    \footnotesize

    \def\ri{0.3}
    \def\ro{0.8}

    \draw[thick] (0, 0) circle [radius=\ri];
    \draw[thick] (0, 0) circle [radius=\ro];

    \draw[in=40, in looseness=3, out looseness=0.5] (\ri, 0) to (-\ro, 0) node[left] {$1/2$};
    \draw[bend right=45] (-\ri, 0) to (0, -\ro) node[below] {$3/4$};
    \draw[bend right=45] (0, -\ri) to (\ro, 0) node[right] {$0\cong 1$};
    \node[above] at (0, \ro) {$1/4$};
  \end{tikzpicture}
  \caption{}
  \label{subfig:Texample1}
\end{subfigure}
\begin{subfigure}[b]{0.3\linewidth}
  \begin{tikzpicture}
    \footnotesize

    \def\base{0.7}

    \draw[thin, color=gray,step=\base] (0, 0) grid (4*\base, 4*\base);

    \draw[->] (0, 0) -- (4.2*\base, 0);
    \draw[->] (0, 0) -- (0, 4.2*\base);

    \node[below left] at (0, 0) {$0$};
    \node[below] at (2*\base, 0) {$1/2$};
    \node[below] at (4*\base, 0) {$1$};

    \node[left] at (0, 2*\base) {$1/2$};
    \node[left] at (0, 4*\base) {$1$};

    \draw (0, 2*\base) -- (2*\base, 3*\base) -- (3*\base, 4*\base);
    \draw (3*\base, 0) -- (4*\base, 2*\base);
  \end{tikzpicture}
  \caption{}
  \label{subfig:Texample2}
\end{subfigure}
\begin{subfigure}[b]{0.3\linewidth}
  \begin{tikzpicture}
    \footnotesize

    \def\base{0.7}

    \draw[thin, color=gray,step=\base] (0, 0) grid (4*\base, 6*\base);

    \draw[->] (0, 0) -- (4.2*\base, 0);
    \draw[->] (0, 0) -- (0, 6.2*\base);

    \node[below left] at (0, 0) {$0$};
    \node[below] at (2*\base, 0) {$1/2$};
    \node[below] at (4*\base, 0) {$1$};

    \node[left] at (0, 2*\base) {$1/2$};
    \node[left] at (0, 4*\base) {$1$};
    \node[left] at (0, 6*\base) {$3/2$};

    \draw (0, 2*\base) -- (2*\base, 3*\base) -- (3*\base, 4*\base) -- (4*\base, 6*\base);
  \end{tikzpicture}
  \caption{}
  \label{subfig:Texample3}
\end{subfigure}
\caption{Three representations of the same element of Thompson's group $T$: \labelcref{sub@subfig:Texample1} as a map $\sircle^1\to \sircle^1$, here drawn by indicating how breakpoints (on the inner circle) are mapped to their images (on the outer circle); \labelcref{sub@subfig:Texample2} the usual representation as a function $[0, 1]\to [0, 1]$; \labelcref{sub@subfig:Texample3} the representation as a function $[0, 1]\to\reals$, which we will use---note that it is a homeomorphism onto its image $[1/2, 3/2]$.}
\label{fig:Texample}
\end{figure}

Given any homeomorphism $f\colon \sircle^1\to \sircle^1$, we can identify it with a homeomorphism $\tilde f\colon\reals\to\reals$ that satisfies
\begin{equation}
  \tilde f(x+1)=\tilde f(x)+1.
\end{equation}
In particular, $\tilde f|_{[0, 1]}$ is continuous, which will be needed later. An example is shown in \cref{fig:Texample}.

Next, we note that the dyadic rationals are dense in $\reals$. To make our construction as explicit as possible, we give an example.
Let $0<p<q$. To find a dyadic rational number in the open interval $(p, q)$, let
\begin{equation}
  \overline\Ceil(x) = \min\setcond{n\in\integers}{n>x}=\begin{cases}
    x+1& \text{if $x\in\integers$},\\
    \ceil{x}& \text{otherwise},
  \end{cases}
\end{equation}
and set
\begin{gather}
  k = \max\left\{0, \overline\Ceil(-\log_2 (q-p))\right\},\\
  m= \overline\Ceil (2^k p).
\end{gather}
Then $m, k\in\naturals$, and $m/2^k\in (p, q)$ is a dyadic rational.

\begin{proof}[Proof of \cref{thm:diffapprox}]
  \begin{figure}
    \begin{tikzpicture}
      \def\base{0.4}
      \draw (0, 0) rectangle (2*8*\base, 11*\base);
      \foreach \y in {1, ..., 10}
        \draw[thin, color=gray, yshift=\y*\base cm] (0, 0) -- (2*8*\base, 0);
      \draw (0, 0) node[below left] {$0$};
      \draw (0, 1*\base) node[left] {$\frac{1}{2^6}$};
      \draw (0, 2*\base) node[left=6pt] {$\frac{2}{2^6}$};
      \draw (0, 3.8*\base) node[left=5pt] {$\vdots$};
      \draw (0, 11*\base) node[left] {$\frac{11}{2^6}$};

      \draw (1*8*\base, 0) node[below] {$\frac{1}{2^3}$} -- ++(0, 11*\base);
      \draw (2*8*\base, 0) node[below] {$\frac{2}{2^3}$};

      \fill[lightgray!50] (0, -0.65) rectangle ++(21.6*\base, -0.6);
      \fill[lightgray!10] (0, -1.25) rectangle ++(21.6*\base, -0.6);
      \fill[lightgray!50] (0, -1.85) rectangle ++(21.6*\base, -0.6);
      \foreach \x in {1, 3}
        \draw[xshift=\x*4*\base cm, thin] (0, 11*\base) -- (0, -0.6) node[below] {$\frac{\x}{2^4}$};
      \foreach \x in {1, 3, 5, 7}
        \draw[xshift=\x*2*\base cm, thin] (0, 11*\base) -- (0, -1.2) node[below] {$\frac{\x}{2^5}$};
      \foreach \x in {1, 3, 5}
        \draw[xshift=\x*1*\base cm, thin] (0, 11*\base) -- (0, -1.8) node[below] {$\frac{\x}{2^6}$};

      \draw (2*8*\base+0.5, -0.95) node[right] {$n=1$};
      \draw (2*8*\base+0.5, -1.55) node[right] {$n=2$};
      \draw (2*8*\base+0.5, -2.15) node[right] {$n=l=3$};

      \draw[thick, line cap=round] (0, 0) -- (3*2*\base, 3*2*\base) -- (2*8*\base, 11*\base);
    \end{tikzpicture}
    \caption{Illustration of how to cut the sides of a dyadic rectangle such that all sides are divided into dyadic partitions with equally many subintervals. In this example, $m_1/2^{k_1}=11/2^6$ and $m_2/2^{k_2}=2/2^3$. Since $11>2$, we divide the left side of the rectangle into $11$ intervals, each of length $1/2^6$. The bottom side is first divided into $2$ intervals, each of length $1/2^3$. Then we successively cut all its intervals in half, repeatedly going from left to right, until the bottom side is also divided into $11$ intervals. The thick line shows the graph of the piecewise linear function arising from these partitions.}
    \label{fig:partitions}
  \end{figure}

  Let two distinct points $p=(p_1, p_2)$ and $q=(q_1, q_2)$ in $\reals^2$ be given, with $p_1<q_1$ and $p_2<q_2$ and such that all coordinates $p_i$, $q_i$ are dyadic rational numbers. Then $r=q-p$ also has dyadic rational coordinates $r_1$ and $r_2$ which can be written as
  \begin{equation}
    r_1=\frac{m_1}{2^{k_1}}, \quad r_2=\frac{m_2}{2^{k_2}}
  \end{equation}
  with $m_1, m_2, k_1, k_2\in\naturals$ and $m_1, m_2>0$. We proceed as illustrated in \cref{fig:partitions}. Let $(a, b)=(1, 2)$ if $m_1\le m_2$ and $(a, b)=(2, 1)$ if $m_1>m_2$, so that $m_b=\max\{m_1, m_2\}$ and $m_a=\min\{m_1, m_2\}$. Set $d=m_b-m_a$. For the moment, assume $d>0$. Consider the sequence $(c_n)$ defined by $c_0=0$ and
  \begin{equation}
    c_n=m_a\sum_{i=0}^{n-1} 2^i=m_a(2^n-1)
  \end{equation}
  for $n\ge 1$. Let $l\ge 1$ be the smallest integer with $c_l\ge d$. Define a sequence of dyadic numbers $\xi_1, \dotsc, \xi_d$ by setting
  \begin{equation}
    \xi_{i+c_n}=\frac{2i-1}{2^{k_a+n}}
  \end{equation}
  for all $i, n$ with either $1\le i\le 2^nm_a$ and $0\le n\le l-1$, or $1\le i\le d-c_{l-1}$ when $n=l$. Set
  \begin{equation}
    X=\Setcond{\frac{m}{2^{k_a}}}{0\le m\le m_a} \cup \{\xi_1, \dotsc, \xi_d\}
  \end{equation}
  for $d>0$ and
  \begin{equation}
    X=\Setcond{\frac{m}{2^{k_a}}}{0\le m\le m_a}
  \end{equation}
  for $d=0$. We arrange the $m_a+d=m_b$ elements of $X$ in increasing order and denote them $x^a_1\le\dotsb\le x^a_{m_b}$. They are the breakpoints of a standard dyadic partition of $[0, m_a/2^{k_a}]$ into $m_b$ intervals. Furthermore, set $x^b_m=m/2^{k_b}$ for $0\le m\le m_b$. The points
  \begin{equation}
    p_1 + x^1_1,\ p_1+ x^1_2,\ \dotsc,\ p_1 + x^1_n
  \end{equation}
  and
  \begin{equation}
    p_2 + x^2_1,\ p_2+ x^2_2,\ \dotsc,\ p_2 + x^2_n
  \end{equation}
  form standard dyadic partitions dividing the intervals $[p_1, p_2]$ and $[q_1, q_2]$, respectively, into equally many subintervals. To these partitions corresponds a piecewise linear function. By construction, it is bijective, has breakpoints only at dyadic rationals, and only slopes wich are powers of $2$. We call such a function a dyadic interpolation from the point $(p_1, q_1)$ to the point $(p_2, q_2)$.

  Now let $f\in\diff^1_+([0, 1])$ and $\epsilon>0$ be given, and assume $\epsilon<1$ without loss of generality. Set $S=\max_{x\in [0, 1]} f'(x)$ and note that $S\ge 1$. Let $\Delta=\lceil-\log_2\frac{\epsilon}{3S}\rceil\in\naturals$ and $n=2^\Delta$, and note that $\Delta\ge 1$. Set
  \begin{equation}
    \xi_i=i/n, \quad i=0, \dotsc, n.
  \end{equation}
  (This implies that $\xi_0=f(\xi_0)=0$ and $\xi_n=f(\xi_n)=1$.)
  Moreover, set
  \begin{equation}
    \delta=\min\{\epsilon/2, (f(\xi_n)-f(\xi_{n-1})/2)\}
  \end{equation}
  and note that the interval
  \begin{equation}
    I_i=(\max\{f(\xi_{i-1})+\delta, f(\xi_i)\},f(\xi_i)+\delta)
  \end{equation}
  is non-empty and a subset of $(0, 1)$ for $i=1, \dotsc, n-1$. We pick a dyadic rational $\eta_i\in I_i$ for each $i=1, \dotsc, n-1$. Let $\eta_0=0$ and $\eta_n=1$, and define the function $g\colon [0, 1]\to [0, 1]$ by setting
  \begin{equation}\label{eq:gdefinition}
    g(x)=\gamma_i(x)
  \end{equation}
  for $x\in [\xi_i, \xi_{i+1}]$ and $i=0, \dotsc, n-1$, where $\gamma_i$ is a dyadic interpolation from the point $(\xi_i, \eta_i)$ to the point $(\xi_{i+1}, \eta_{i+1})$. From the definitions of $\gamma$, $\{\xi_i\}$ and $\{\eta_i\}$ it is clear that $g\in F$. Furthermore, for all $i=0, \dotsc, n-1$ and $x\in [\xi_i, \xi_{i+1}]$, consider the sequence of statements
  \begin{align}
    \abs{g(x)-f(x)}&\le g(\xi_{i+1})-f(\xi_i) \label{eq:epsilon1}\\
    &< f(\xi_{i+1})-f(\xi_i)+\epsilon/2 \label{eq:epsilon2}\\
    &= \frac{f(\xi_{i+1})-f(\xi_i)}{\xi_{i+1}-\xi_i} (\xi_{i+1}-\xi_i)+\epsilon/2 \label{eq:epsilon3}\\
    &< S2^{-\Delta}+\epsilon/2 \label{eq:epsilon4}\\
    &< \epsilon/3+\epsilon/2 < \epsilon \label{eq:epsilon5}.
  \end{align}
  \cref{eq:epsilon1} holds since $f$ and $g$ are strictly increasing and $g(\xi_i)>f(\xi_i)$. For $\cref{eq:epsilon2}$, recall that $g(\xi_{i+1})=\eta_{i+1}<f(\xi_{i+1})+\delta$. \cref{eq:epsilon3,eq:epsilon4,eq:epsilon5} are obvious. We have thus found $g\in F$ with $\max_{x\in [0, 1]}\abs{f(x)-g(x)}<\epsilon$.

  If instead $f\in\diff^1_{+}(\sircle^1)$, $f$ corresponds to a function $\tilde f\colon\reals\to\reals$ with $\im(\tilde f)=[u, u+1]$ for some $u\in\reals$ and such that $\tilde f\colon[0, 1]\to [u, u+1]$ is a diffeomorphism (as explained above). Define $S$, $\Delta$, $n$, $\xi_i$ and $I_i$ as above, but with $\delta=\min\{\epsilon/2, (\tilde f(\xi_1)-\tilde f(\xi_0))/2\}$. Choose $\eta_i\in I_i$ for $i=1, \dotsc, n-1$ as before. Let $\eta_0$ be a dyadic rational in the interval $(\tilde f(\xi_0)+\delta, \tilde f(\xi_1))$ and set $\eta_n=\eta_0+1$. This ensures that
  \begin{equation}
    \max\{\tilde f(\xi_{n-1})+\delta, \tilde f(\xi_n) \}<\eta_n.
  \end{equation}
  Now we can define a function $\tilde g\colon [0, 1]\to\reals$ as in \cref{eq:gdefinition}. It follows that \cref{eq:epsilon1,eq:epsilon2,eq:epsilon3,eq:epsilon4,eq:epsilon5} hold, and that $g\in T$ upon taking the quotient $\sircle^1=\reals/\integers$.
\end{proof}

The next logical question is whether there is an approximation for the first derivatives of diffeomorphisms. While generally elements of both $F$ and $T$ are not everywhere differentiable, we can define a function
\begin{equation}
  d(f, g)=\sup_{x\in \sircle^1\setminus B_g} \abs{f'(x)-g'(x)}
\end{equation}
that measures the distance between the first derivatives of $f\in\diff^1_+(\sircle^1)$ and $g\in T$ wherever $g'$ is defined. Here $B_g$ denotes the set of breakpoints of $g$. (The definition of $d$ for $\diff^1_+([0, 1])$ and $F$ is analogous.) We can therefore rephrase the question: Given a diffeomorphism $f$ and $\epsilon>0$, is there a function $g$ from the appropriate Thompson group such that $d(f, g)<\epsilon$? The answer is that such an approximation is not possible since the set of all integer powers of $2$ is very sparse in $(0, 1)$. This fact is made precise in the following proposition, which is similar to \cite[Théorème~III.2.3]{GhysSergiescu}.

\begin{prop}\label{prop:discreteness}
  For every $f\in\diff^1_+(\sircle^1)$ which is not a rotation, there exists $\mu>0$ such that $d(f, g)>\mu$ for all $g\in T$. The same holds when $\sircle^1$ is replaced by $[0, 1]$ and $T$ is replaced by $F$.
\end{prop}

Here the rotations in $\diff^1_+(\sircle^1)$ are all elements $f$ with $f'(x)=1$ for all $x\in \sircle^1$, which includes the identity. In $\diff^1_+([0, 1])$, the identity is the only rotation.

Note that the proof is also valid in the more general case when $\diff^1_+([0, 1])$ and $\diff^1_+(\sircle^1)$ are replaced by the sets of all differentiable bijections of $[0, 1]$ or $\sircle^1$, respectively, whose inverses are also differentiable.

\begin{proof}[Proof of \cref{prop:discreteness}]
  Let $g\in T$ and $f\in\diff^1_+(\sircle^1)$. We will identify $f$ and $g$ with functions on the interval $[0, 1]$ as before. Let $x_0\in [0, 1]\setminus B_{g}$. The two powers of $2$ closest to $f'(x_0)$ are given by
  \begin{equation}
    2^{\floor{\log_2 f'(x_0)}} \le f'(x_0) \le 2^{\ceil{\log_2 f'(x_0)}}.
  \end{equation}
  If $f'(x_0)$ is not a power of $2$, the inequalities are strict and therefore
  \begin{equation}
    d(f, g) \ge \min\Bigl\{ \bigl\lvert f'(x_0)-2^{\floor{\log_2 f'(x_0)}}\bigr\rvert, \bigl\lvert f'(x_0)-2^{\ceil{\log_2 f'(x_0)}}\bigr\rvert \Bigr\}>0.
  \end{equation}
  The case that $f'(x_0)$ is not a power of $2$ for some $x_0\in [0, 1]\setminus B_g$ occurs for all differentiable $f\in\diff^1_+(\sircle^1)$ except for rotations. For if $f$ is not a rotation, there exists $x_1\in [0, 1]$ with $f'(x_1)=c\neq 1$. By the mean value theorem, there also exists $x_2\in [0, 1]$ with $f'(x_2)=1$. Without loss of generality, assume $c<1$ and $x_1<x_2$. Then by Darboux's theorem, $[c, 1]\subset f'([x_1, x_2])$. Since $B_g$ is finite, $[c, 1]\setminus f'(B_g)\subset\im(f')$ surely contains points which are not powers of $2$.

  It is clear that $\diff^1_+([0, 1])$ and $F$ are a special case of this argument, which concludes the proof.
\end{proof}

\part{Applications}\label{part:applications}


\chapter{Dynamics for Holographic Codes}
\label{cha:dynamics}

In this chapter, we describe how to introduce dynamics for the holographic codes and states introduced by \citeauthor{PastawskiYoshidaHarlowPreskill2015} \cite{PastawskiYoshidaHarlowPreskill2015}. This task requires the definition of a kinematical Hilbert space which can be constructed as a colimit, or direct limit, of finite-dimensional Hilbert spaces. This Hilbert space is known as the semicontinuous limit. The dynamics is then introduced by building a unitary representation of Thompson's group $T$ using the machinery from \cref{cha:localization} and specifically \cref{sec:fractrep}. The bulk Hilbert space is realized as a particular subspace of the semicontinuous limit Hilbert space spanned by a class of distinguished states which can be assigned a discrete bulk geometry. The bulk-boundary correspondence in this toy model is given by the fact that the analogue of the group of (large) bulk diffeomorphisms, the Ptolemy group, is isomorphic to $T$.
This chapter is based on the paper \cite{OsborneStiegemann2017}, but we focus here more on the mathematical description.

\section{Dyadic Tessellations of the Poincaré Disk}

In this section we describe how the Poincaré disk $\disk$ can be understood as an equal-time slice of $\mathrm{AdS}_3$. We describe the most important properties of the geometry of $\disk$, and introduce several important groups that act on $\disk$ and its boundary.

\subsection{$(2+1)$-Dimensional Anti de Sitter Space}

For convenience, we describe $\mathrm{AdS}_3$ in coordinates, where we closely follow \cite{Holst1996}. We take it as the $3$-dimensional surface
\begin{equation}
  X^2+Y^2-Z^2-T^2=-1
\end{equation}
embedded in $4$-dimensional flat space with metric
\begin{equation}
  \dif s^2=\dif X^2+\dif Y^2-\dif Z^2-\dif T^2.
\end{equation}
Choosing the parametrization
\begin{equation}
  \left\{\begin{aligned}
    X&=\sinh r\cos\varphi\\
    Y&=\sinh r\sin\varphi\\
    Z&=\cosh r\cos t\\
    T&=\cosh r\sin t
  \end{aligned}\right.
\end{equation}
in terms of coordinates $r$, $\varphi$, and $t$ gives the metric
\begin{equation}
  \dif s^2=-\cosh^2r\dif t^2+\dif r^2+\sinh^2 r\dif\varphi^2.
\end{equation}
$(r, \varphi)$ can be viewed as ordinary polar coordinates, while $t$ is the time coordinate. Note that $t$ is $2\pi$-periodic, and that the spacetime contains closed timelike curves.

It is useful to further transform the radial coordinate according to
\begin{equation}
  \rho(r)=\frac{\cosh r-1}{\sinh r}.
\end{equation}
In this case the metric becomes
\begin{equation}
  \dif s^2=-\left(\frac{1+\rho^2}{1-\rho^2}\right)^{\!\!2}\dif t^2 + \frac{4}{(1-\rho^2)^2}(\dif\rho^2+\rho^2\dif\varphi^2).
\end{equation}
In the coordinates $(\rho, \varphi, t)$ the $\mathrm{AdS}_3$ space is confined to a cylinder and time goes lengthwise. The boundary at $\rho\to 1$ (that is, $r\to\infty$) is timelike and called conformal infinity. A cross-section, or $t=\mathrm{constant}$ surface, of this cylinder is a Poincaré disk with polar coordinates $(\rho, \varphi)$. Spacelike geodesics in $\mathrm{AdS}_3$ are arcs of circles that meet the boundary circle $\rho=1$ at a right angle, and timelike geodesics inside $\mathrm{AdS}_3$ spiral around the time axis $\rho=0$ and never reach infinity.

\subsection{The Poincaré Disk Model}
\label{sub:poincare-disk-model}

The Poincaré half-plane model of hyperbolic space is defined on the upper half-plane
\begin{equation}
  \halfplane=\{z\in\complexes : \Re z>0\}.
\end{equation}
The metric tensor is given by
\begin{equation}
  \dif s^2=\frac{\dif z\dif\bar z}{y^2},
\end{equation}
where we write $z=x+\imag y$. The symmetry group of $\halfplane$ is given by the Möbius transformations
\begin{equation}
  f(z)=\frac{az+b}{cz+d}
\end{equation}
with \emph{real} coefficients $a, b, c, d\in\reals$ that satisfy $ad-bc=1$. The Möbius group is isomorphic to $\mathit{PSL}_2(\reals)$; here $\mathit{SL}_2(\reals)$ consists of all real $2\times 2$ matrices $A$ with $\det(A)=1$, and $\mathit{PSL}_2(\reals)/\{\pm I\}$ is the projective version.

The Poincaré disk model is defined on the open unit disk
\begin{equation}
  \disk=\{z\in\complexes : \abs{z}<1\}.
\end{equation}
In this case, the metric tensor is given by
\begin{equation}
  \dif s^2=\frac{4\dif z\dif\bar z}{(1-\abs{z}^2)^2}.
\end{equation}
Geodesics in $\disk$ are diameters and (parts of) circle segments that meet the boundary at right angles. Examples are shown in \cref{fig:poincare-geodesics}.

\begin{figure}
  \includegraphics{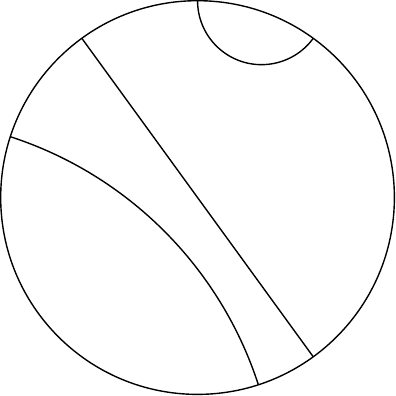}
  \caption{Geodesics in the Poincaré disk.}
  \label{fig:poincare-geodesics}
\end{figure}

There is a special fractional linear transformation, called the Cayley transform, which is given by the function
\begin{equation}
  \complexes\to\complexes,\quad z\mapsto\frac{z-\imag}{z+\imag}\,.
\end{equation}
It is a conformal transformation from $\halfplane$ to $\disk$. It also maps $\reals\cup\{\infty\}$ to $\sircle^1$ in $\complexes$, where $\sircle^1=\partial\disk$ is the boundary circle of $\disk$ at infinity. By conjugation with the Cayley transform, we can regard every element of $\mathit{PSL}_2(\reals)$ as a map $\disk\to\disk$.

An important subgroup is given by $\mathit{PSL}_2(\integers)$, the modular group. It has the presentation
\begin{equation}
  \langle a, b \mid a^2=b^3=1 \rangle.
\end{equation}
The two generators can be written as matrices
\begin{equation}
  a=\begin{pmatrix}0&-1\\1&0\end{pmatrix}, \quad b=\begin{pmatrix}0&-1\\1&1\end{pmatrix}.
\end{equation}
As functions on $\complexes$, they take the form
\begin{equation}
  a(z)=-\frac{1}{z}, \quad b(z)=-\frac{1}{z+1}.
\end{equation}
$\mathit{PSL}_2(\integers)$ can be regarded as a subgroup of Thompson's group $T$ (that is, it is isomorphic to a subgroup of $T$). The isomorphism
\begin{equation}
  g\mapsto {?}\circ g\circ {?}^{-1}
\end{equation}
is given by conjugation with the function ${?}\colon\reals P^1\to [0, 1]$ called Minkowski's question mark function. We will identify $\mathit{PSL}_2(\integers)$ with its image in $T$ under this isomorphism. One can then show that $a=CA$ and $b=C$. Hence, $\mathit{PSL}_2(\integers)$ is the subgroup of $T$ generated by $A$ and $C$.

It has been shown by \citeauthor{Imbert1997} that $T$ as a whole is isomorphic to a group known as $\mathit{PPSL}_2(\integers)$, the group of \emph{piecewise} $\mathit{PSL}_2(\integers)$ homeomorphisms of the circle \cite{Imbert1997}. In order to explain this correspondence, we need to introduce dyadic tessellations of the Poincaré disk.

\subsection{Tessellations}

In the following description, we largely follow \cite{Penner1997}.
An \defemph{ideal triangle} in $\disk$ is a triangle bounded by geodesics whose vertices all lie on the circle $\sircle^1$ at infinity.
A \defemph{tessellation} $\tau=\{\gamma_i\}$ is a countable family of pairwise disjoint geodesics in $\disk$ such that each connected component of $\disk\setminus\tau$ is an ideal triangle and the decomposition of $\disk$ into triangles is locally finite.\footnote{Here locally finite means that every point in $\disk$ has a neighbourhood that intersects only finitely many triangles.} We refer to a choice of edge $\gamma^*$ together with a specification of orientation on it as a `distinguished oriented edge', or a `doe' for short. Furthermore, we let $\tau^0\subset\sircle^1$ denote the set of vertices which are endpoints of geodesics in $\tau$.

\begin{figure}
  \includegraphics[width=0.5\textwidth]{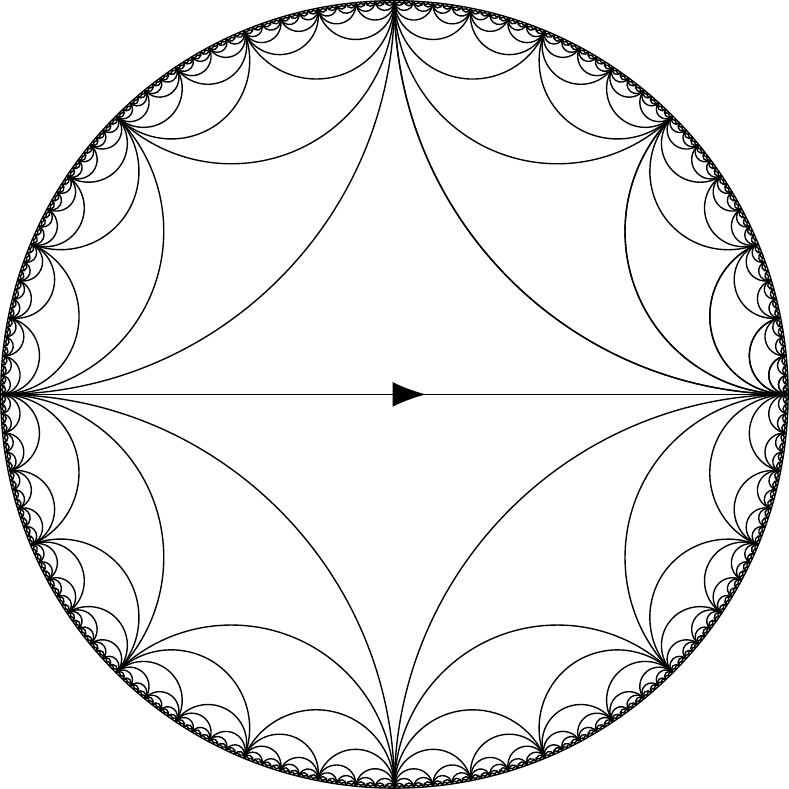}
  \caption{The standard dyadic tessellation with doe.}
\end{figure}
Instead of labelling the points of $\sircle^1$ by their preimages under the Cayley transform, we can consider $\sircle^1$ as the unit interval $[0, 1]$ in counter-clockwise order and with $0$ and $1$ identified, such that $1/0=\infty$ becomes $0\sim 1$. We are particularly interested in \emph{dyadic} tessellations, which are defined by having as elements in $\tau^0$ only dyadic rational points on the circle. The most important one is the \defemph{standard dyadic tessellation} $\tau_d$,\footnote{This is actually confusing terminology, but we follow the literature. Note that this is just the dyadic tessellation which we designate to be the \emph{standard} one. The term `standard dyadic' is not meant here.} which we define inductively according to the following rules:
\begin{genenum}
  \item Begin with the non-tessellated disk and with the trivial partition $P_0=\{[0, 1]\}$, which is thought to divide the circle at $0\sim 1$.
  \item Given a partition $P_k$ of the circle, divide all intervals in $P_k$ in two equal halves; these become precisely the elements of $P_{k+1}$. For each $I\in P_k$, connect the point at which $I$ is being divided to the endpoints of $I$ by a geodesic in $\disk$.
\end{genenum}
We define the doe $\gamma^*$ to be the geodesic from $1/2$ to $0$ (in this direction).
\begin{figure}
  \includegraphics[width=0.9\textwidth]{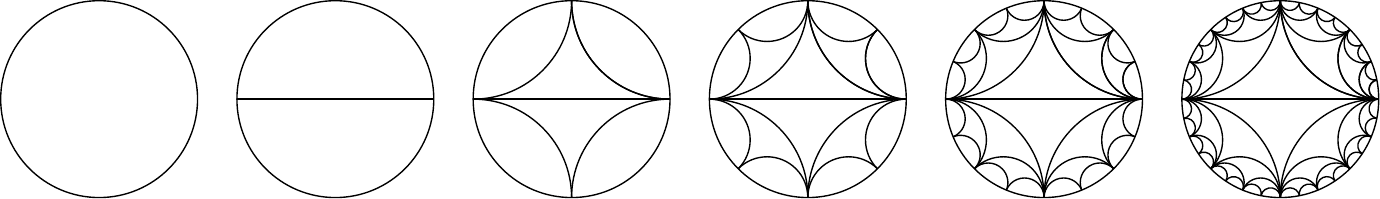}
  \caption{Step-by-step construction of the standard dyadic tessellation, beginning with the undivided disk.}
\end{figure}

\begin{rem}
  The standard dyadic tessellation (without doe) is invariant under $\mathit{PSL}_2(\integers)\subset T$, that is, the action of $A$ and $C$, whereas $B$ changes the tessellation. \cref{fig:poincare-abc} shows how $A$, $B$, and $C$ act on the standard dyadic tessellation.
\end{rem}
\begin{figure}
  \includegraphics[width=0.9\textwidth]{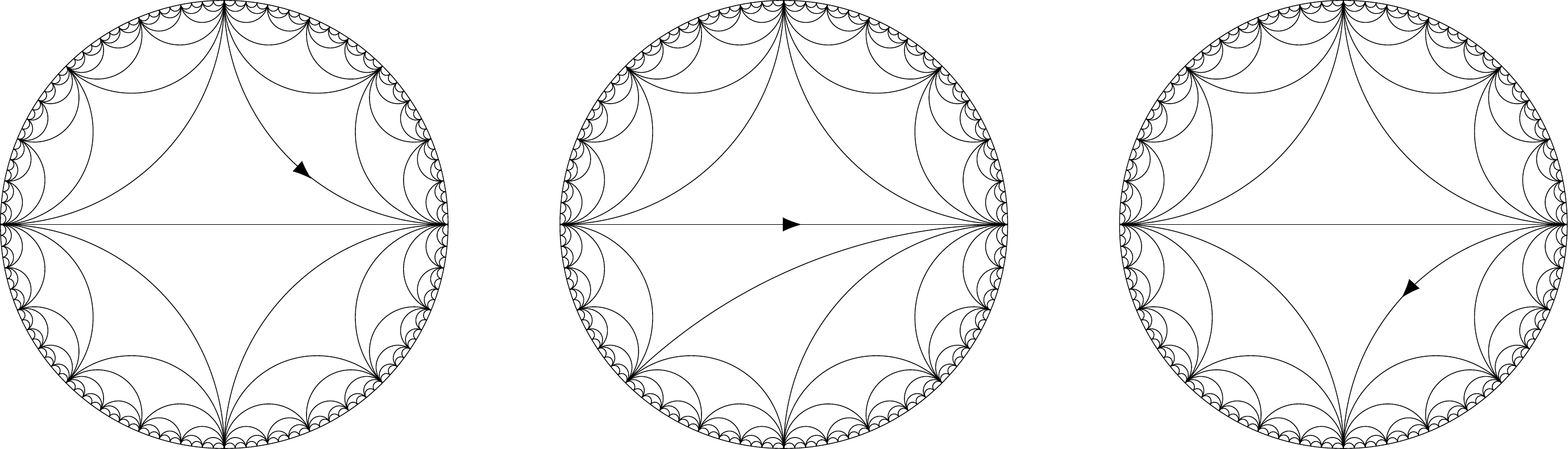}
  \caption{The images of $\tau_d$ under $A$, $B$, and $C$.}
  \label{fig:poincare-abc}
\end{figure}

Given any tessellation $\tau$ with doe, it is possible to define an order-preserving map $f_\tau\colon\tau_d^0\to\tau^0$ which extends uniquely to a homeomorphism $\sircle^1\to\sircle^1$, called the characteristic mapping of $\tau$ and also denoted $f_\tau$.\footnote{We use here a version of Penner's characteristic mapping which factors through the ${?}$ function. This allows us to limit ourselves to dyadic tessellations.}

\subsection{The Ptolemy Group}

Suppose that $\tau$ is a dyadic tessellation with doe $\gamma^*$, and suppose that $\gamma\in\tau$ is an edge of $\tau$. $\gamma$ separates two triangles in $\disk\setminus\tau$ which together comprise an ideal quadrilateral of which $\gamma$ is  diagonal; let $\gamma'$ denote the other diagonal of this quadrilateral. We may alter $\tau$ to produce a new dyadic tessellation
\begin{equation}
  \tau_\gamma=(\tau\cup\{\gamma'\})\setminus\{\gamma\}
\end{equation}
as shown in \cref{subfig:elementary-move}. If $\gamma$ is not the doe, then $\tau_\gamma$ inherits the doe from $\tau$. In this case, the elementary move has order $2$. If, however, $\gamma=\gamma^*$, then we designate $\tau'$ to be the new doe. Its orientation is obtained by rotating the orientation of $\gamma^*$ in counter-clockwise direction, see \cref{subfig:elementary-move-doe}. This kind of elementary move has order $4$.
\begin{figure}
  \begin{subfigure}{0.9\linewidth}\centering
    \includegraphics{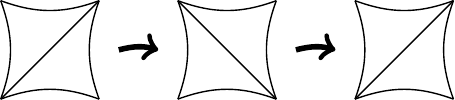}
    \caption{}
    \label{subfig:elementary-move}
  \end{subfigure}
  \par\bigskip
  \par\medskip
  \begin{subfigure}{0.9\linewidth}\centering
    \includegraphics{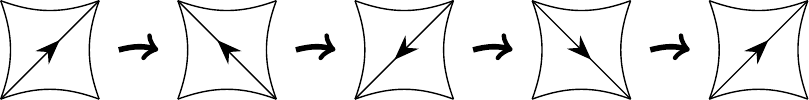}
    \caption{}
    \label{subfig:elementary-move-doe}
  \end{subfigure}
  \caption{Elementary moves.}
\end{figure}

Note that if $\tau$ is any dyadic tessellation, the characteristic mapping establishes a bijection between the edges $\tau_d$ of the standard dyadic tessellation and the edges of $\tau$. Therefore, to specify an elementary move along an edge of $\tau$, we might as well specify the corresponding edge of $\tau_d$. Thus we define the monoid $M$ of elementary moves to be the free monoid over the set $\tau_d$; if $\gamma\in M$, we let $\gamma\cdot\tau$ denote the tessellation with doe obtained from $\tau$ by applying the elementary move along the edge of $\tau$ corresponding to $\gamma\in\tau_d$ via the characteristic mapping. We obtain an action on the set of all tessellations by setting
\begin{equation}
  (\gamma_2\gamma_1)\cdot\tau=\gamma_2\cdot(\gamma_1\cdot\tau)
\end{equation}
for any tessellation $\tau$.

Observe that since $M$ is freely generated, there are many elements in $M$ which actually do not change $\tau_d$. To account for this, let $K$ be the submonoid of $M$ consisting of all elements that act identically on $\tau_d$. The \defemph{universal Ptolemy group} is then defined to be the quotient $G=M/K$. (See \cite{Penner1997} for a proof that $G$ is indeed a group.)

By utilizing the characteristic mapping, we can represent elements of $G$ by homeomorphisms of $\sircle^1$. Specifically, the map
\begin{equation}
  G\to\operatorname{Homeo}_+(\sircle^1), \quad g\mapsto f_{g\cdot\tau_d}
\end{equation}
is a faithful representation of $G$ on the group of orientation-preserving homeomorphisms. It has been shown by \citeauthor{Imbert1997} that this map actually establishes an isomorphism between $G$ and $\mathit{PPSL}_2(\integers)$. We summarize these remarkable findings in the following theorem.

\begin{thm}
  The universal Ptolemy group $G$, Thompson's group $T$ and $\mathit{PPSL}_2(\integers)$ are isomorphic.
\end{thm}


\subsection{Cutoffs}

\begin{figure}
  \includegraphics[width=0.5\textwidth]{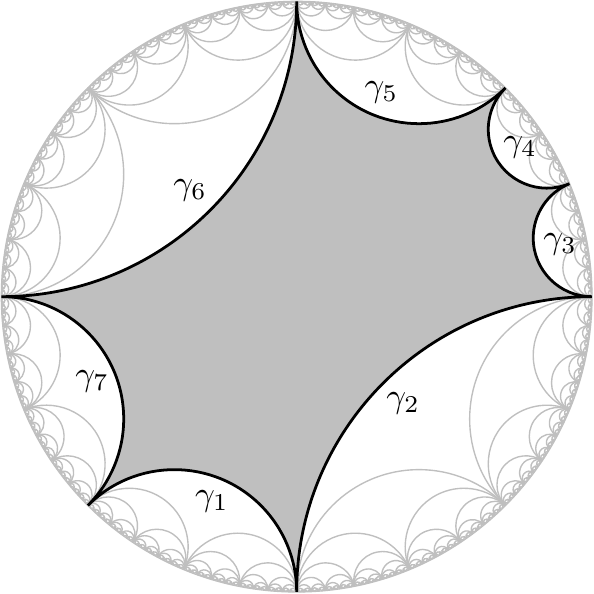}
  \caption{A cutoff with $7$ geodesics. We do not show any geodesics of the underlying tessellation inside the cutoff because they are not specific features of the cutoff.}
\end{figure}

Let $\tau$ be a tessellation\footnote{A doe is not needed since we implicitly use $0\sim 1$ on the boundary circle as a reference point.} and let $n\ge 3$. An $n$-tuple $\Gamma=(\gamma_1, \dotsc, \gamma_n)$ of geodescics in $\tau$ is called a \defemph{cutoff} if the geodesics form a counter-clockwise oriented closed cycle. This means that there are a total of $n$ distinct endpoints of geodesics in $\tau^0$; if $(\gamma_i, \gamma_{i+1})$ is a pair of two subsequent geodesics in $\Gamma$ (including the case $(\gamma_{n}, \gamma_1)$), then $\gamma_i$ and $\gamma_{i+1}$ share exactly one common endpoint $p_2$ in $\tau^0$, and the two remaining endpoints, $p_1$ of $\gamma_i$ and $p_3$ of $\gamma_{i+1}$, are such that $p_1$, $p_2$ and $p_3$, in this order, appear on $\sircle^1$ in counter-clockwise order. We will refer to both the tuple of geodesics as well as the area they enclose as the cutoff. Note that a cutoff can also be specified by naming the $n$ vertices on the boundary circle.

We denote by $\mathit{Co}$ the set of all cutoffs of all admissible tessellations, and by $\mathit{Co}_d$ the subset of all cutoffs of the standard dyadic tessellation. On either set, a directed partial order $\le$ is given by inclusion of cutoff areas.
\begin{figure}
  \includegraphics[width=0.5\textwidth]{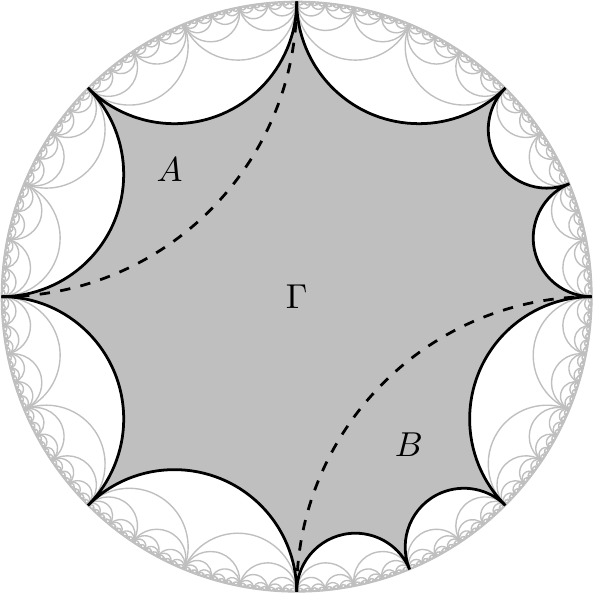}
  \caption{Inclusion of cutoffs. We have $\Gamma\le\Gamma'$, where $\Gamma'$ is the cutoff containing both $\Gamma$ and the regions $A$ and $B$.}
\end{figure}
Cutoffs should be thought of as UV cutoffs, and different cutoffs correspond to different levels of fine-graining.

\section{Planar Perfect Tensors}

In \cite{PastawskiYoshidaHarlowPreskill2015}, a perfect tensor is defined as a tensor $T_{j_1j_2\dots j_n}$ such that for any bipartition of its indices into two sets $\{ j_1, j_2, \dotsc \}=A\sqcup A^c$, where $\abs{A}\le\abs{A^c}$ without loss of generality, $T$ is proportional to an isometry from the Hilbert space associated with $A$ to the Hilbert space associated with $A^c$. This notion of perfect tensor can be generalized to monoidal categories in different ways, depending on the context in which the tensors are used. In a symmetric monoidal category with duals, one can arbitrarily divide the legs of a morphism into input and ouput legs by composing with the evaluation and coevaluation maps because any crossing of legs is trivially braided. For example, the legs of the tensor
\begin{equation}
  \begin{tikzpicture}[scale=0.8, rounded corners=2pt]
    \def\len{0.309}
    \draw[fill=vlightgray] (0, 0) rectangle (1.618, 1) node[pos=0.5] {$T$};
    \draw (0.309, 1) -- ++(0, \len);
    \draw (0.809, 1) -- ++(0, \len);
    \draw (1.309, 1) -- ++(0, \len);
  \end{tikzpicture}
\end{equation}
can easily be rearranged as
\begin{equation}
  \begin{tikzpicture}[scale=0.8, baseline=6pt, rounded corners=2pt]
    \def\len{0.309}
    \draw[fill=vlightgray] (0, 0) rectangle (1.618, 1) node[pos=0.5] {$T$};
    \draw (0.309, 1) -- ++(0, 0.5*\len) -- ++(1, \len) -- ++(0, 0.5*\len);
    \draw (0.809, 1) -- ++(0, 0.5*\len) -- ++(-0.5, \len) -- ++(0, 0.5*\len);
    \draw (1.309, 1) -- ++(0, 0.5*\len) -- ++(-0.5, \len) -- ++(0, 0.5*\len);
  \end{tikzpicture}\,.
\end{equation}
For this reason no ambiguities arise when we interpret the above definition of perfect tensor here. This means that the definition of perfect tensor can be transferred to the case of symmetric monoidal categories without change.
In general, however, monoidal categories may only have proper braidings, if at all. Therefore we use a different notion of perfection, called planar perfect, in which leg crossings are not mentioned at all.

Let $\ccal$ be a $\field$-linear semisimple spherical dagger category with a single generating symmetrically self-dual object $Y$. Denote by $\widehat\coev$ and $\widehat\ev$ the modified coevaluation and evaluation for $Y$.

Let $m, k\in\naturals$. If $m, k\ge 1$, we define a linear map
\begin{gather}
  {\rot}\colon \ccal(Y^{\otimes m}, Y^{\otimes k}) \to \ccal(Y^{\otimes m}, Y^{\otimes k}),\\
  T \mapsto (\widehat\ev\otimes\id_Y^{\otimes k})\circ (\id_Y\otimes T\otimes\id_Y)\circ (\id_Y^{\otimes m}\otimes\widehat\coev).
\end{gather}
(Here, we take $\id_Y^0=\id_I$.) For $m\ge 0$ and $k\ge 1$, we set
\begin{gather}
  {\down}\colon \ccal(Y^{\otimes m}, Y^{\otimes k}) \to \ccal(Y^{\otimes m+1}, Y^{\otimes (k-1)}),\\
  T \mapsto (\widehat\ev\otimes\id_Y^{\otimes(k-1)})\circ (\id_Y\otimes T).
\end{gather}
We write $\rot^n(T)$ for $n$-fold application $\rot\circ\dotsm\circ\rot(T)$, and similarly for ${\down}$.
$\rot^0$ and $\down^0$ do not change their arguments. A morphism $T\colon A\to B$ is called an \defemph{isometry} if $T^\dagger\circ T=\id_A$. The following definition is adapted from \cite[Def.~2]{PastawskiYoshidaHarlowPreskill2015} and \cite{BergerOsborne2018}.

\begin{defn}
  A morphism $T\colon I\to Y^{\otimes k}$ is called \defemph{planar perfect} if
  \begin{equation}
    \rot^j\down^i(T)
  \end{equation}
  is proportional to an isometry for all $i, j\in\naturals$ with $1\le i\le k/2$ and $0\le j\le k-1$, or with $i=0$ and $j=0$.
\end{defn}
It is easy to see that whenever $k$ is even, that is, $k=2k'$, one only needs to check the cases where $i=k'$.

\begin{defn}
  A morphism $T\colon Y^{\otimes m}\to Y^{\otimes k}$ with $m,k\in\naturals$ and $k\ge 1$ is called \defemph{rotation-invariant} if
  \begin{equation}
    T=\rot^n\down(T)
  \end{equation}
  for all $n$.
\end{defn}

Of course, application of ${\down}$ is unnecessary whenever $m\ge 1$.

\begin{exmp}\label{exmp:V-spin-system}
  The map $V\colon\complexes^3\to\complexes^3\otimes\complexes^3$ defined by
  \begin{equation}
  	\langle jk|V|l\rangle = \begin{cases}
  		0 & \text{if $j=k$, $k=l$, or $l=j$,}\\
  		\frac{1}{\sqrt{2}} & \text{otherwise,}
  	\end{cases}
  \end{equation}
  is planar perfect and rotation-invariant. This is an example where $Y$ is the Hilbert space $\complexes^3$. (We have used braket notation to denote basis vectors.)
\end{exmp}

\begin{exmp}\label{exmp:V-trivalent-vertex}
  In any trivalent category, we can take $Y=X$. Then the trivalent vertex is planar perfect and rotation-invariant in the above sense. However, this is not a very interesting example since $\dim\ccal_3=1$.
\end{exmp}

\begin{exmp}\label{exmp:trivalent}
  In any trivalent category, we can take $Y=X\otimes X$. The corresponding subcategory generated by $Y$ inherits all relevant properties from the trivalent category. Then
  \begin{equation}
  	V=\begin{tikzpicture}[scale=0.7,baseline=-13mm,yscale=-1]
      \draw[rounded corners] (2.5, 3) -- (2.5, 2) -- (1.875, 1.375);
      \draw (2.25, 1) -- (2.5, 0.75);
      \draw[rounded corners] (2.5, 0.75) -- (1.25, 2) -- (0, 0.75);
      \draw[draw=white,double=black,double distance=0.4pt,line width=3pt] (2, 3) -- (2, 2) -- (0.75, 0.75);
  		\draw (2, 2) -- (2.2, 1.8);
      \draw[draw=white,double=black,double distance=0.4pt,line width=3pt] (2.1, 1.9) -- (3.25, 0.75);
    \end{tikzpicture}
  \end{equation}
  is planar perfect and rotation-invariant.
\end{exmp}

\cref{exmp:V-spin-system,exmp:trivalent} will be our guiding examples in the rest of this work. More examples of planar perfect tensors can be found in \cite{HarrisMcMahonBrennenStace2018}, where they are called ``block perfect''. The notion of a perfect tensor is older, and there are currently more examples, see e.g.\ \cite{PastawskiYoshidaHarlowPreskill2015,EnriquezWintrowiczZyczkowski2016,RaissiGogolinRieraAcin18}.

\section[Unitary Representations]{Unitary Representations of Thompson's Groups\\on the Semicontinuous Limit}

In this section, we explain our main example for the construction introduced in \cref{sec:fractrep}. We will first introduce it without any reference to tessellations of the Poincaré disk or cutoffs and then later explain the connection.

Throughout this section, we fix a finite-dimensional Hilbert space $\shfrak=\complexes^d$ and a planar perfect and rotation-invariant tensor $V\colon\shfrak\to\shfrak\otimes\shfrak$.

\subsection{Cutoffs and Holographic States}

We begin by associating with every cutoff $\Gamma\in\mathit{Co}_d$ a Hilbert space
\begin{equation}
  \hfrak_\Gamma=\shfrak^{\otimes\abs{\Gamma}}.
\end{equation}
Recall that $\abs{\Gamma}$ denotes the number of edges of $\Gamma$. Every vector $\psi\in\hfrak_\Gamma$ can be interpreted as a tensor with $\abs{\Gamma}$ uncontracted legs.
\begin{figure}
  \begin{subfigure}[b]{0.4\textwidth}\centering
    \includegraphics[width=\textwidth]{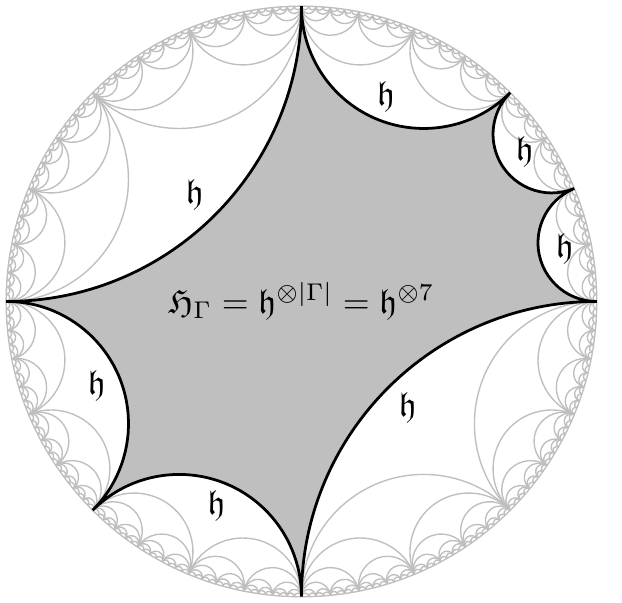}
    \caption{}
    \label{subfig:cutoff-hilbert-space}
  \end{subfigure}
  \hspace*{0.05\textwidth}
  \begin{subfigure}[b]{0.4\textwidth}\centering
    \includegraphics[width=\textwidth]{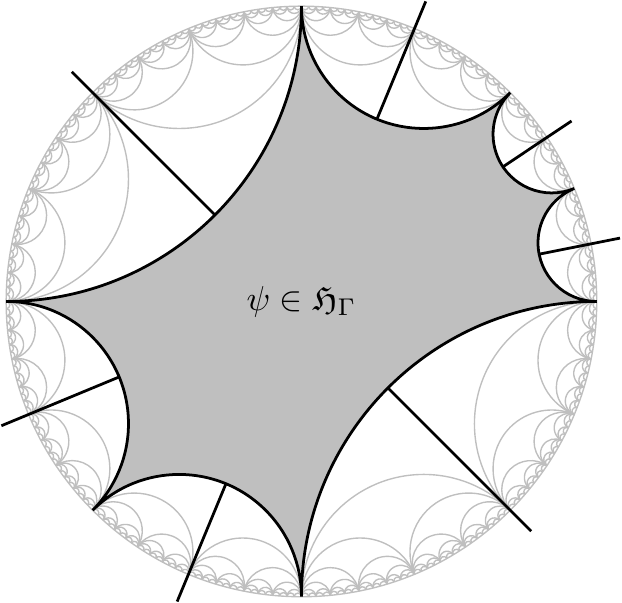}
    \caption{}
    \label{subfig:cutoff-state}
  \end{subfigure}
  \caption{The Hilbert space $\hfrak_\Gamma$ and a state vector associated with a cutoff $\Gamma$. In both pictures, the cutoff is depicted as the gray area. \labelcref{sub@subfig:cutoff-hilbert-space} The Hilbert space $\hfrak_\Gamma$ is $\mathfrak{h}$ raised to the tensor power of $\abs{\Gamma}$, where $\abs{\Gamma}$ is the number of geodesics that bound the cutoff region. In this example, we have $\abs{\Gamma}=7$. \labelcref{sub@subfig:cutoff-state}~A state vector $\psi\in\hfrak_\Gamma$ may be regarded as a tensor network with $\abs{\Gamma}$ open legs.}
  \label{fig:cutoff-hilbert-state}
\end{figure}

These Hilbert spaces contain special states called \defemph{holographic states}. To this end, assume that $\Gamma$ has an underlying tessellation $\tau$. We place one copy of $V$ in every ideal triangle that is contained within the area of $\Gamma$. In every triangle, we associate each one of its edges with one open leg of $V$. At every edge of $\tau$ inside the area of $\Gamma$, two triangles meet, and every such edge is therefore associated with two legs, one coming from each copy of $V$ in the two triangles. We contract these legs, that is, sum over the respective indices. The resulting tree tensor network is shown in \cref{fig:holographic-state}.
\begin{figure}
  \includegraphics[width=0.5\textwidth]{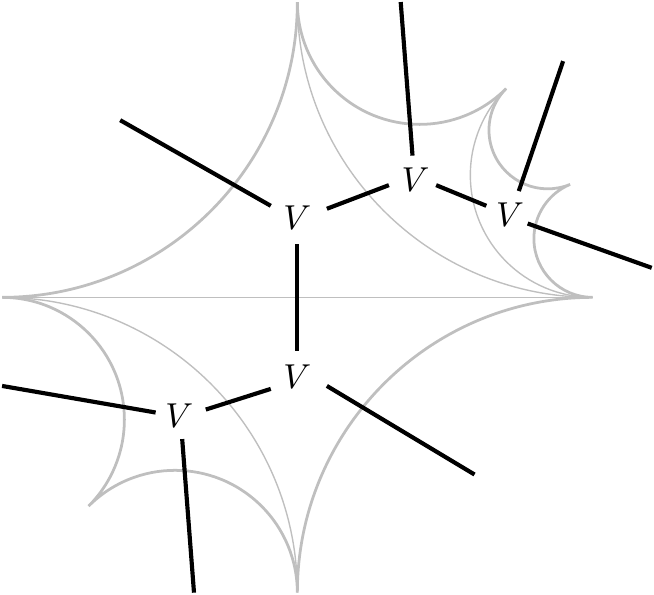}
  \caption{A holographic state inside a cutoff region.}
  \label{fig:holographic-state}
\end{figure}

\subsection{The Semicontinuous Limit}

Define a functor $\Phi\colon \annfor\to\hsans$ by setting
\begin{equation}
  \Phi\bigl(\,\tikz[scale=0.2,baseline=1mm]{
  \draw (0, 0) -- (0, 1);
  \draw (0, 1) -- (-1, 2);
  \draw (0, 1) -- (1, 2);
  }\,\bigr)=V
\end{equation}
and extending in the obvious way. Here we write $\hsans$ for the category of Hilbert spaces with bounded operators as morphisms.

There is a special subset of cutoffs in $\mathit{Co}_d$, the cutoffs whose vertices on the boundary form a standard dyadic partition of $\sircle^1$.
The set of such partitions is denoted $D(\sircle^1)$.
As mentioned before, we have a directed partial order on $D(\sircle^1)$, so we can define another functor $\Psi\colon D(\sircle^1)\to\hsans$ given by
\begin{equation}
  \Psi(f\to g)=\Phi(p)\quad \text{whenever $p\circ f=g$.}
\end{equation}
Of course, $D(\sircle^1)$ is here identified with the subset of $\chom(\annfor)$ consisting of morphisms with domain $1$.

The procedure from \cref{sec:fractrep} gives us a direct limit
\begin{equation}
  \hfrak = \varinjlim_{f} \hfrak_f
\end{equation}
for $\Psi$. Here $f$ denotes a standard dyadic partition, which is the same as the corresponding cutoff, or a binary tree. The direct limit Hilbert space and its tensor powers are effectively the images of the Kan extension $\pi$ evaluated at different objects of $\annfor$.

\begin{rem}
  Note that the direct limit Hilbert space can be constructed directly. We set
  \begin{equation}
    \hat\hfrak=\bigsqcup_f\hfrak_f.
  \end{equation}
  We can identify each $\psi\in\hfrak_f$ with all its images under $\Psi^f_g$ for some $g\ge f$. (\cref{fig:state-and-v} shows two equivalent states.) This gives an equivalence relation $\sim$, and since the $\Psi^f_g$ are isometries and hence preserve the inner product, $\hat\hfrak/{\sim}$ is naturally an inner-product space: If $\psi_f\in\hfrak_f$ and $\phi_g\in\hfrak_g$ are two vectors which are understood to be embedded in $\hat\hfrak$, their inner product is
  \begin{equation}
    \langle \Psi^f_h\psi_f, \Psi^g_h\phi_g \rangle,
  \end{equation}
  where $h\ge f,g$ exists because $\binfor$ and $\annfor$ satisfy \cref{it:lf3} from \cref{defn:calculus-fractions}. Completing $\hat\hfrak$ it gives the direct limit $\hfrak$.
\end{rem}
\begin{figure}
  \includegraphics[width=0.5\textwidth]{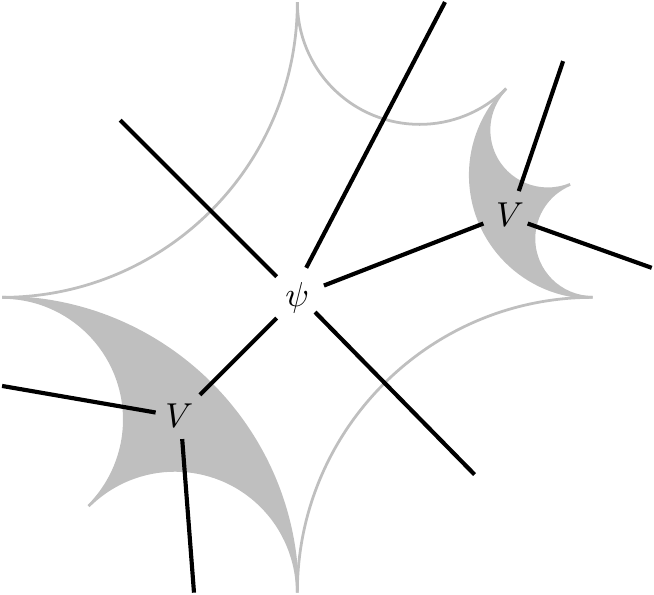}
  \caption{In the semicontinuous limit, a state $\psi$ localized in a cutoff region (white) is equivalent to a state in a larger cutoff region. The other state is formed by placing a $V$ in every triangle that must be added (gray).}
  \label{fig:state-and-v}
\end{figure}

\begin{rem}
  Note that although a general cutoff $\Gamma\in\mathit{Co}_d$ does not necessarily correspond to a standard dyadic partition, there always exists a cutoff $\Gamma'\ge\Gamma$ whose vertices form a partition in $D(\sircle^1)$. Since everything is embedded in $\tau_d$, there is a unique way of adding ideal triangles to $\Gamma$ in order to obtain $\Gamma'$, which corresponds to an isometric embedding $\hfrak_\Gamma\hookrightarrow\hfrak_{\Gamma'}$. Therefore, $\varinjlim_{\mathit{Co}_d} \hfrak_\Gamma=\varinjlim_{D(\sircle^1)} \hfrak_P$.
\end{rem}

We come now to the most important part: the action of $\pi$ on morphisms. Since $V$ is an isometry, $\pi$ maps all elements of Thompson's groupoid $\tcal$ to unitary operators. Let $f/g$ be a fraction representing an element of Thompson's group $T$, and let $\psi$ be a vector in $\hfrak_h$, which we denote as a pair $(h, \psi)$. Then
\begin{equation}\label{eq:t-action}
  \pi(f/g) (h, \psi) = (p\circ f, \Phi(q)\psi),
\end{equation}
where $p, q$ have been chosen such that $p\circ g=q\circ h$. We can rewrite \cref{eq:t-action} in an informal way as
\begin{equation}
  \frac{f}{g}\cdot\frac{h}{\psi} = \frac{pf}{\cancel{pg}}\cdot\frac{\cancel{qh}}{\Phi(q)\psi} = \frac{pf}{\Phi(q)\psi}.
\end{equation}
This illustrates the fact that the action is similar to the composition of two fractions.

Of course, the representation can also be composed with linear maps between Hilbert spaces. This automatically gives us dynamics for holographic \emph{codes}.

There is a special state $\Omega\in\hfrak$ that is invariant under the modular group $\mathit{PSL}_2(\integers)$. It is the equivalence class of a single ideal triangle containing a $V$ in $\tau_d$. The state in \cref{fig:holographic-state} is in the same equivalence class. The invariance follows from the fact that $V$ is perfect and rotation-invariant. For this reason, we will call $\Omega$ the \defemph{vacuum state}.

\subsection{Discontinuity}

It has been shown by \citeauthor{Jones18} \cite{Jones18} and by \citeauthor{KlieschKoenig2018} \cite{KlieschKoenig2018} that the above representation of $T$ is highly discontinuous. More precisely, for all but a zero-measure set of isometries $V$, there exists a sequence $(f_n)$ in $T$ such that $\lim_k\norm{f_k-\id}_\infty=0$ but
\begin{equation}
  \lim_k \langle\phi, \pi(f_k)\psi\rangle\neq\langle\phi, \psi\rangle
\end{equation}
for some $\phi, \psi\in\hfrak$; see \cite{KlieschKoenig2018} for a proof.

There are two ways to deal with it. First, we can look for those few isometries than \emph{do} admit a continuous representation of $T$. While \cite{KlieschKoenig2018} only contains a necessary condition that such isometries need to fulfill, one can use variants of this condition to try to single out those isometries which allow continuous representations.

The other viewpoint is that discontinuity is just a property that our toy models possess, precisely because they are \emph{toy} models. After all, our models are highly discrete, and the semicontinuous limit is not a full continuum limit.

\section{The Bulk-Boundary Correspondence}

In this section we build a Hilbert space $\hfrak_{\mathrm{bulk}}\subset\hfrak$ that we will argue is the space of bulk states on which an analogue of the group of large diffeomorphisms acts---this will turn out to be the Ptolemy group $G$. To this end, we take again the state $\Omega$ consisting only of contractions of $V$. We build up $\hfrak_{\mathrm{bulk}}$ from $\Omega$ by adding all states that arise from $\Omega$ by the action of $T$, so that
\begin{equation}
  \hfrak_{\mathrm{bulk}}=\overline{\mathrm{span}\raisebox{2.5mm}{}}\{\pi(f)\Omega : f\in T\}.
\end{equation}
Each state of the form $\pi(f)\Omega$ represents a different bulk geometry.
By construction, we have a unitary representation of $T$ on $\hfrak_{\mathrm{bulk}}$.
Because the Ptolemy group $G$ is isomorphic to $T$ we can directly take it to act on $\hfrak_{\mathrm{bulk}}$; we now have the strongest possible manifestation of the bulk-boundary correspondence. The `discrete conformal group' $T$ of the boundary is precisely the group of `discrete diffeomorphisms' of the bulk. In analogy to the AdS-CFT correspondence, we call this the Ptolemy-Thompson correspondence.


\addtocontents{toc}{\protect\newpage}
\chapter{Quantum Fields for Thompson's Groups}
\label{cha:quantum-fields}

We now explore further the aspects that should turn our model into a proper toy model of a conformal field theory. We propose a definition for quantum field operators on the semicontinuous limit Hilbert space, define the corresponding correlation functions, and extract information resembling conformal data. Very recent related work in this direction may be found in \cite{BrothierStottmeister2019,Osborne2019}. At the end of the chapter, we collect some observations on trees that are used in the other sections.
This chapter is based on \cite{OsborneStiegemann2019}.

\section[Correlation Functions]{Correlation Functions for Tree Tensor Networks\\ in Quantum Spin Systems}\label{sec:corrfunstrees}

In this section we discuss how to compute correlation functions for tree states. To make things easier, we present this for the \emph{symmetric} monoidal category of Hilbert spaces.

Throughout, we fix a $d$-dimensional Hilbert space $\shfrak=\complexes^d$.
Our focus in this section is on the case where
\begin{equation}
	\hfrak_N = \shfrak^{\otimes N}= \complexes^d\otimes\dotsm\otimes\complexes^d.
\end{equation}
To define a tree tensor network in such a setting, we need an isometry $V\colon\complexes^d\to\complexes^d\otimes\complexes^d$, which can be represented graphically as a trivalent vertex
\begin{equation}
    V = \tikz[baseline=-13pt, scale=0.75,yscale=-1]{
      \draw (0.0, 0.0) -- (0.5, 0.5) -- (1.0, 0.0);
      \draw (0.5, 0.5) -- (0.5, 1.0);
    }\,.
\end{equation}
It considerably simplifies our discussion to assume that
\begin{equation}
	V\mathrm{swap} = V,
\end{equation}
where
\begin{equation}
	\operatorname{swap}(\phi\otimes\psi) = \psi\otimes\phi, \quad \phi,\psi\in\complexes^d.
\end{equation}
Graphically, this can be expressed as
\begin{equation}
	\begin{tikzpicture}[scale=0.75,baseline=20pt]
		\draw (0, 0) -- (0, 0.5);
		\draw[rounded corners] (0, 0.5) -- (-0.5, 1) -- (0.5, 2);
		\draw[rounded corners] (0, 0.5) -- (0.5, 1) -- (-0.5, 2);
	\end{tikzpicture} =
	\tikz[baseline=-13pt, scale=0.75,yscale=-1]{
		\draw (0.0, 0.0) -- (0.5, 0.5) -- (1.0, 0.0);
		\draw (0.5, 0.5) -- (0.5, 1.0);
	}\,,
\end{equation}
where the cross stands for the symmetric braiding.

Using $V$, we set up the completely positive linear map
\begin{equation}
	E\colon M_d(\complexes)\to M_d(\complexes),\quad A\mapsto V^\dagger(A\otimes\idmtx)V.
\end{equation}
Therefore,
\begin{equation}
	E(A)=
	\begin{tikzpicture}[baseline=28pt,scale=0.75]
		\draw (0, 0) -- (0, 0.5);
		\draw[rounded corners] (0, 0.5) -- (-0.5, 1) -- (-0.5, 2) -- (0, 2.5);
		\draw[rounded corners] (0, 0.5) -- (0.5, 1) -- (0.5, 2) -- (0, 2.5);
		\draw (0, 2.5) -- (0, 3);
		\filldraw[fill=white] (-0.5, 1.5) circle [radius=0.3] node {$A$};
	\end{tikzpicture}
\end{equation}
in our graphical language.
We are going to further simplify things and assume that $E$ is diagonalizable. Then we can obtain right eigenvalues $\lambda_\alpha$ and eigenvectors $\mu^\alpha\in M_d(\complexes)$, so that
\begin{equation}
	E(\mu^\alpha)=
	\begin{tikzpicture}[baseline=30pt,scale=0.75]
		\draw (0, 0) -- (0, 0.5);
		\draw[rounded corners] (0, 0.5) -- (-0.5, 1) -- (-0.5, 2) -- (0, 2.5);
		\draw[rounded corners] (0, 0.5) -- (0.5, 1) -- (0.5, 2) -- (0, 2.5);
		\draw (0, 2.5) -- (0, 3);
		\filldraw[fill=white] (-0.5, 1.5) circle [radius=0.4] node {$\mu^\alpha$};
	\end{tikzpicture}
	=\lambda_\alpha
	\tikz[baseline=30pt, scale=0.75]{
		\draw (0.0, 0.0) -- (0.0, 1.1);
		\draw (0.0, 1.5) circle [radius=0.4] node {$\mu^\alpha$};
		\draw (0.0, 1.9) -- (0.0, 3.0);
	} = \lambda_\alpha \mu^\alpha, \quad \alpha = 0,1,\dotsc, d^2-1.
\end{equation}
Without loss of generality we assume that $\lambda_0 = 1$ and $\mu^0 = I$. It is worth emphasizing that the left eigenvectors $\nu^\alpha$, $\alpha = 0, 1, \dotsc, d^2-1$, furnish us with a way to expand an operator $x \in M_d(\complexes)$ with respect to $\mu^\alpha$, so that
\begin{equation}
	A = \sum_{\alpha=0}^{d^2-1} \langle\nu^\alpha, A\rangle_2 \, \mu^\alpha,
\end{equation}
where $\langle\blank, \blank\rangle_2$ is the Hilbert-Schmidt inner product $\langle a, b\rangle_2 = \frac1d\tr(a^\dag b)$.  (See the Appendix of \cite{OsborneStiegemann2019} for more on this.) The operators $\mu^\alpha$ are called the \defemph{ascending operators}.

Now that we have the eigenvalues and eigenvectors of $E$, we introduce the \defemph{fusion map}
\begin{equation}
	F\colon M_d(\complexes)\times M_d(\complexes) \to M_d(\complexes), \quad (A, B) \mapsto V^\dag(A\otimes B)V,
\end{equation}
which is diagrammatically given as
\begin{equation}
	F(A, B) =\begin{tikzpicture}[baseline=28pt,scale=0.75]
		\draw (0, 0) -- (0, 0.5);
		\draw[rounded corners] (0, 0.5) -- (-0.5, 1) -- (-0.5, 2) -- (0, 2.5);
		\draw[rounded corners] (0, 0.5) -- (0.5, 1) -- (0.5, 2) -- (0, 2.5);
		\draw (0, 2.5) -- (0, 3);
		\filldraw[fill=white] (-0.5, 1.5) circle [radius=0.3] node {$A$};
		\filldraw[fill=white] (0.5, 1.5) circle [radius=0.3] node {$B$};
	\end{tikzpicture}\, .
\end{equation}
The \defemph{fusion coefficients} for this map are given by
\begin{equation}
	{f^{\alpha\beta}}_\gamma = \frac{1}{d}\tr\left((\nu^\gamma)^\dag F(\mu^\alpha, \mu^\beta)\right)
\end{equation}
so that
\begin{equation}
	F(\mu^\alpha, \mu^\beta) = \sum_{\gamma} {f^{\alpha\beta}}_\gamma \mu^\gamma.
\end{equation}
The fusion coefficients may be regarded as the structure constants for an, in general, non-associative and non-commutative algebra built on the indices $\alpha$, and
\begin{equation}
	\alpha \star \beta = \sum_{\gamma} {f^{\alpha\beta}}_\gamma \gamma.
\end{equation}

Let $\Omega_N$, $N=2^m$, be the tree state defined by the following diagram:
\begin{equation}
	\begin{tikzpicture}[scale=0.2, yscale=-1]\small
		\draw (15, 15) -- (15, 16);
    \draw (0.0, 0.0) -- (1.0, 1.0) -- (2.0, 0.0);
    \draw (4.0, 0.0) -- (5.0, 1.0) -- (6.0, 0.0);
    \draw (1.0, 1.0) -- (3.0, 3.0) -- (5.0, 1.0);
    \draw (8.0, 0.0) -- (9.0, 1.0) -- (10.0, 0.0);
    \draw (12.0, 0.0) -- (13.0, 1.0) -- (14.0, 0.0);
    \draw (9.0, 1.0) -- (11.0, 3.0) -- (13.0, 1.0);
    \draw (3.0, 3.0) -- (7.0, 7.0) -- (11.0, 3.0);
    \draw (16.0, 0.0) -- (17.0, 1.0) -- (18.0, 0.0);
    \draw (20.0, 0.0) -- (21.0, 1.0) -- (22.0, 0.0);
    \draw (17.0, 1.0) -- (19.0, 3.0) -- (21.0, 1.0);
    \draw (24.0, 0.0) -- (25.0, 1.0) -- (26.0, 0.0);
    \draw (28.0, 0.0) -- (29.0, 1.0) -- (30.0, 0.0);
    \draw (25.0, 1.0) -- (27.0, 3.0) -- (29.0, 1.0);
    \draw (19.0, 3.0) -- (23.0, 7.0) -- (27.0, 3.0);
    \draw (7.0, 7.0) -- (15.0, 15.0) -- (23.0, 7.0);

    \draw[decoration={brace, mirror, amplitude=10}, decorate] (-1.5, -0.5) -- node[left=12pt] {$m$ levels} (-1.5, 16.5);
    \draw[decoration={brace, amplitude=10}, decorate] (-0.5, -1) -- node[above=12pt] {$N=2^m$ leaves} (30.5, -1);
  \end{tikzpicture}
\end{equation}
Here $V$ is the trivalent vertex. (Actually, $\Omega_N$ is not a state \emph{vector}, but $A\mapsto\tr(\Omega_N A)$ is a state.) This state corresponds to the state $\Omega$ of \cref{cha:dynamics}. We first show how to compute one-point correlation functions for the ascending operators, that is, the expectation values
\begin{equation}
	\langle \mu_j^{\alpha}\rangle_N = \tr\left(\Omega_N (\idmtx_0\otimes \cdots \otimes \idmtx_{j-1}\otimes \mu_j^\alpha \otimes \idmtx_{j+1}\otimes \cdots \otimes \idmtx_{N-1})\right).
\end{equation}
Here, $A_j$ indicates that $A$ is the $j$-th tensor factor.
The calculation is expedited upon noting that
\begin{equation}\label{eq:expval}
	\langle \mu_j^{\alpha}\rangle_N = \frac1d (\lambda_{\alpha})^{m-1} \tr(\mu^\alpha) = (\lambda_{\alpha})^{m-1} \langle \idmtx,\mu^\alpha\rangle_2.
\end{equation}
The following figure shows an example:
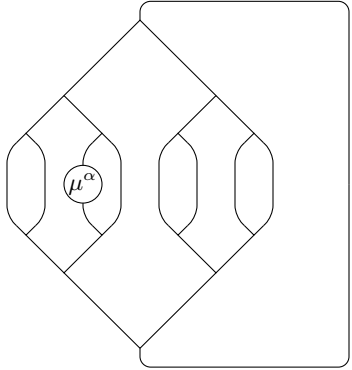
\begin{figure}[h]
	\begin{tikzpicture}[scale=0.25]\footnotesize
		\draw[rounded corners] (7, 7) -- (7, 8) -- (18, 8) -- (18, -1.7);
		\draw (0.5, 0.5) -- (1.0, 1.0) -- (1.5, 0.5);
		\draw (4.5, 0.5) -- (5.0, 1.0) -- (5.5, 0.5);
		\draw (1.0, 1.0) -- (3.0, 3.0) -- (5.0, 1.0);
		\draw (8.5, 0.5) -- (9.0, 1.0) -- (9.5, 0.5);
		\draw (12.5, 0.5) -- (13.0, 1.0) -- (13.5, 0.5);
		\draw (9.0, 1.0) -- (11.0, 3.0) -- (13.0, 1.0);
		\draw (3.0, 3.0) -- (7.0, 7.0) -- (11.0, 3.0);

		\foreach \i in {0,2,4,6}{
			\draw[rounded corners] (2*\i+0.5, 0.5) -- (2*\i, 0) -- ++(0, -3.4) -- ++(0.5, -0.5);
		}

		\foreach \i in {1,3,5,7}{
			\draw[rounded corners] (2*\i-0.5, 0.5) -- (2*\i, 0) -- ++(0, -3.4) -- ++(-0.5, -0.5);
		}

		\filldraw[fill=white] (4, -1.7) circle[radius=1] node {$\mu^\alpha$};

		\begin{scope}[yscale=-1,yshift=3.4cm]
			\draw[rounded corners] (7, 7) -- (7, 8) -- (18, 8) -- (18, -1.7);
			\draw (0.5, 0.5) -- (1.0, 1.0) -- (1.5, 0.5);
			\draw (4.5, 0.5) -- (5.0, 1.0) -- (5.5, 0.5);
			\draw (1.0, 1.0) -- (3.0, 3.0) -- (5.0, 1.0);
			\draw (8.5, 0.5) -- (9.0, 1.0) -- (9.5, 0.5);
			\draw (12.5, 0.5) -- (13.0, 1.0) -- (13.5, 0.5);
			\draw (9.0, 1.0) -- (11.0, 3.0) -- (13.0, 1.0);
			\draw (3.0, 3.0) -- (7.0, 7.0) -- (11.0, 3.0);
		\end{scope}
	\end{tikzpicture}
	\caption{$\langle \mu_j^{\alpha}\rangle_N$ for $m=3$, $N=2^3=8$, and $j=2$.}
\end{figure}

Building on this we next focus on the goal of computing the two-point correlation function
\begin{equation}
	\begin{split}
		\langle \mu^\alpha_j \mu^\beta_k\rangle_N &= \tr\bigl(\Omega_N (\idmtx_0\otimes \dotsm \otimes \idmtx_{j-1}\otimes \mu_j^\alpha \otimes \idmtx_{j+1}\otimes \dotsm\\
		&\mspace{180mu}\dotsm \otimes \idmtx_{k-1}\otimes \mu^{\beta}_k \otimes \idmtx_{k+1}\otimes \dotsm \otimes \idmtx_{N-1})\bigr),
	\end{split}
\end{equation}
for $0\le j<k < 2^m-1$. To this end we label the leaves of the regular binary tree $t_m$ having $2^m$ leaves with binary expansions as per \cref{sub:trees}. In this way we associate to the $j$-th vertex, $j\in \{0,1,\ldots, 2^m-1\}$, the number $x= j/2^m\in [0,1)$. We write
\begin{equation}
	C^{\alpha\beta}_m(x,y) = \langle \mu^\alpha_j \mu^\beta_k\rangle_N,
\end{equation}
where $x = j/2^m$ and $y = k/2^m$.

The key to computing two-point correlation functions is to note that we can relate $C^{\alpha\beta}(x,y)$ to $C^{\alpha\beta}(x^{(1)},y^{(1)})$, where we have employed the notation $x^{(1)}$ in \cref{sub:trees} for the number whose binary expansion has one fewer digit than that for $x$. From \cref{sub:trees} we also take the tree metric $d_T$.

\begin{lem}
	Let $x, y\in [0,1)$ be two points with $m$-digit binary expansions
	\begin{equation}
		x = 0.x_{-1}\cdots x_{-m}, \qquad y = 0.y_{-1}\cdots y_{-m}.
	\end{equation}
	 Then, writing $d_T(x,y) = k+1$, we have that
	\begin{equation}
		C^{\alpha\beta}_m(x,y) = (\lambda_{\alpha}\lambda_{\beta})^k C^{\alpha\beta}_{m-k}(x^{(k)},y^{(k)}).
	\end{equation}
\end{lem}
\begin{proof}
	The details are omitted because the basic argument is a relatively straightforward induction. Since $\mu^\alpha_j$ and $\mu^\beta_k$ are on different carets of the tree we can independently apply $E$ to these operators to relate the expectation values at different levels:
	\begin{equation}
		C^{\alpha\beta}_m(x,y) = \langle \mu^\alpha_j \mu^\beta_k\rangle_N = \langle E(\mu^\alpha_j) E(\mu^\beta_k)\rangle_{N/2} = (\lambda_{\alpha}\lambda_{\beta})C^{\alpha\beta}_{m-1}(x^{(1)},y^{(1)}).
	\end{equation}
	The two operators only meet at a caret at level $m-k$, so we can repeat the above process $k$ times.
\end{proof}

Now we study what happens when $d_T(x,y) = 1$. In this case we can relate the two-point correlation function to an expectation value.
\begin{lem}
	Suppose that we have two leaves labelled $x$ and $y$, and that $d_T(x,y) = 1$. Then
	\begin{equation}
		C^{\alpha\beta}_m(x,y) = \tr(\Omega_{N/2}F(\mu^\alpha, \mu^\beta)).
	\end{equation}
	Exploiting the formula \cref{eq:expval} for the expectation value we hence obtain
	\begin{equation}
		C^{\alpha\beta}_m(x,y)  = \frac1d  \sum_{\gamma = 0}^{d^2-1} {f^{\alpha\beta}}_\gamma (\lambda_{\gamma})^{m-2} \tr(\mu^\gamma).
	\end{equation}
\end{lem}

Putting these lemmas all together, along with \cref{lem:treemetric}, we obtain the following formula for the two-point correlation function
\begin{prop}
	Let $x, y\in [0,1)$ be two points with $m$-digit binary expansions $x = 0.x_{-1}\cdots x_{-m}$ and $y = 0.y_{-1}\cdots y_{-m}$. Then
	\begin{equation}
		C^{\alpha\beta}_m(x,y) = \frac1d (\lambda_\alpha\lambda_\beta)^{m+\lfloor \log_2(y\ominus x) \rfloor} \sum_{\gamma = 0}^{d^2-1} {f^{\alpha\beta}}_\gamma (\lambda_{\gamma})^{-\lfloor \log_2(y\ominus x) \rfloor-2} \tr(\mu^\gamma).
	\end{equation}
\end{prop}
A particularly important special case is the end-to-end correlation function between the leaf labelled $0.00\cdots0$ and the leaf labelled $0.11\cdots 1$. This may be computed in the following way.
\begin{cor}
	Let $x, y\in [0,1)$ be two points with $m$-digit binary expansions
	\begin{equation}
		x = 0.0_{-1} \cdots 0_{-m}, \qquad y = 0.1_{-1}\cdots 1_{-m}.
	\end{equation}
	Then
	\begin{equation}
		C^{\alpha\beta}_m(x,y) = \frac1d (\lambda_\alpha\lambda_\beta)^{m-1} \sum_{\gamma = 0}^{d^2-1} {f^{\alpha\beta}}_\gamma (\lambda_{\gamma})^{-1} \tr(\mu^\gamma).
	\end{equation}
\end{cor}

\section[Primary Fields]{Primary Fields for Thompson's Groups\\in the Semicontinuous Limit}\label{sec:primaryfields}
In this section we build observables on the semicontinuous limit Hilbert space $\hfrak$ intended to represent smeared equal-time field operators. We generalize everything to the case of non-regular partitions of $\sircle^1$ as it costs no extra effort to do so (and is arguably more elegant). The discussions here are framed in terms of quantum spin systems, however, they generalize in a straightforward way to anyonic systems via trivalent categories.

Recall that $D([0, 1])$ denotes the directed set of standard dyadic partitions of $[0,1]$.
\begin{defn}
	Let $V\colon\complexes^d\to \complexes^d\otimes \complexes^d$, $E(A) = V^\dag (A\otimes\idmtx) V$, $E(\mu^\alpha)= \lambda_\alpha\mu^\alpha$, and let $P\in D([0, 1])$ be a standard dyadic partition. Let $f\in L^2([0,1],M_d(\complexes))$. Assume that $\lambda_\alpha\neq0$ for all $\alpha$ and define the following operator in $\mathcal{B}(\hfrak_{P})$, where $\hfrak_P = \shfrak^{\otimes\abs{P}}$,
	\begin{equation}\label{eq:smearedfieldoperator}
		\phi_{{P}}(f) = \sum_{\substack{\alpha = 0 \\ \lambda_\alpha \neq 0}}^{d^2-1}\sum_{I\in {P}} \bar{f}_\alpha(I) (\lambda_\alpha)^{\log_2(\abs{I})}\mu_I^\alpha,
	\end{equation}
	where
	\begin{equation}
		\bar{f}_\alpha(I) = \frac{1}{d}\int_{I} \tr\bigl((\nu^\alpha)^\dag f(x)\bigr) \dif x,
	\end{equation}
	$\abs{I}=\sup(I)-\inf(I)$ is the length of an interval $I$, and $\lambda_\alpha$ and $\nu^\alpha$ are the corresponding ascending weights and dual operators for $E$.
\end{defn}

\begin{rem}
	The discretized field operator $\phi_{{P}}(f)$ is meant to represent a continuum field operator first smeared out by $f$ and then \emph{discretized}, \emph{averaged}, or \emph{coarse grained}, over the intervals making up the partition $P$. Intuitively, as the partition $P$ is taken finer and finer we should recover a dyadic version of the standard smeared field operator $\phi(f)$ of quantum field theory.
\end{rem}

We want to use the discretized field operator $\phi_P(f)$ to build $n$-point correlation functions. In the language of physics, this is achieved by replacing the smearing function $f$ with a delta function $f(x)=\delta(x-z)A$, where $A\in M_d(\complexes)$. Mathematically this can be achieved by using operator-valued distributions. However, in order to avoid a long digression on distributions, we make the substitution $f(x)=\delta(x-z)M$ in \cref{eq:smearedfieldoperator} and note that
\begin{equation}
	\overline{f}_\alpha(I) = \tr\left((\nu^\alpha)^\dag M\right)\int \delta(x-z)\chi_I(x)\, dx = \tr\left((\nu^\alpha)^\dag M\right)\chi_I(z),
\end{equation}
where $\chi_I$ denotes the indicator function of $I$.
From this, we can distil the following \emph{ad hoc} definition.

\begin{defn}
	The \defemph{discretized field operator} of type $\alpha$ at $z\in \sircle^1$ with respect to the partition $P$ is defined to be
	\begin{equation}\label{eq:discretefieldoperator}
		\phi_P^\alpha(z) = \sum_{I\in {P}} \chi_I(z) (\lambda_\alpha)^{\log_2\abs{I}}\mu_I^\alpha.
	\end{equation}
\end{defn}

Using products of discretized field operators at various positions should allow us to define $n$-point functions. The general idea is that, for a given state $\psi\in\hfrak$, the $n$-point function is
\begin{equation}
	C_\psi^{\alpha_1\alpha_2\dotsb\alpha_n}(x_1, x_2, \dotsc, x_n) = \langle \psi, \phi^{\alpha_1}(x_1)\phi^{\alpha_2}(x_2)\dotsm \phi^{\alpha_n}(x_n)\psi\rangle,
\end{equation}
where the $\phi^\alpha(x)$ are putative quantum field operators that do not depend on a partition $P$. We therefore need to perform some kind of limit to remove the partition $P$ in the definition \cref{eq:discretefieldoperator} of the discretized field operator. In order to perform this limit, we introduce some auxiliary definitions.

\begin{defn}
	Let $\mathbf{x} = (x_1, x_2, \ldots, x_n)$ be an ordered tuple of numbers lying in $[0,1)$, that is, $0\le x_1 <x_2<\cdots < x_n < 1$. We say that $P\in D([0, 1])$ is a \defemph{supporting partition} for $\mathbf{x}$ if in every interval $I\in P$ there is \emph{at most} one of the elements of the tuple $\mathbf{x}$ in $I$. We say that $P$ is a \defemph{minimal supporting partition} for $\mathbf{x}$ if there is no coarser supporting partition.
\end{defn}

\begin{exmp}
	Let $\mathbf{x} = (\tfrac{1}{7}, \tfrac{2}{3}, \tfrac{5}{6})$. Then the partition
	\begin{equation}
		\{ [0,\tfrac14), [\tfrac14,\tfrac12), [\tfrac12,\tfrac34), [\tfrac34,1]\}
	\end{equation}
	is supporting for $\mathbf{x}$:
	\begin{center}
		\begin{tikzpicture}[scale=5]
			\draw (0, 0) -- (1,0);
			\draw (0, -0.025) -- (0, 0.025);
			\draw (0.25, -0.025) -- (0.25, 0.025);
			\draw (0.5, -0.025) -- (0.5, 0.025);
			\draw (0.75, -0.025) -- (0.75, 0.025);
			\draw (1, -0.025) -- (1, 0.025);
			\draw[fill=black] (1/7, 0) circle (0.01);
			\draw[fill=black] (2/3, 0) circle (0.01);
			\draw[fill=black] (5/6, 0) circle (0.01);
		\end{tikzpicture}
	\end{center}
	The partition $\{ [0,\tfrac12), [\tfrac12,\tfrac34), [\tfrac34,1]\}$ is a minimal supporting partition:
	\begin{center}
		\begin{tikzpicture}[scale=5]
			\draw (0, 0) -- (1,0);
			\draw (0, -0.025) -- (0, 0.025);
			\draw (0.5, -0.025) -- (0.5, 0.025);
			\draw (0.75, -0.025) -- (0.75, 0.025);
			\draw (1, -0.025) -- (1, 0.025);
			\draw[fill=black] (1/7, 0) circle (0.01);
			\draw[fill=black] (2/3, 0) circle (0.01);
			\draw[fill=black] (5/6, 0) circle (0.01);
		\end{tikzpicture}
	\end{center}
\end{exmp}

One can prove the following lemmas.

\begin{lem}
	Let $\mathbf{x}$ be an ordered tuple in $[0,1)$ and let $P$ support $\mathbf{x}$. Suppose that $Q$ refines $P$, that is, $P\le Q$. Then $Q$ supports $\mathbf{x}$.
\end{lem}

\begin{lem}
	Let $\mathbf{x}$ be an ordered tuple lying in $[0,1)$. Then the minimal supporting partition for $\mathbf{x}$ is unique.
\end{lem}
\begin{proof}
	Recall that $\mathcal{T}$ is the infinite binary tree of standard dyadic intervals whose nodes are labelled by $[\tfrac{a}{2^m}, \tfrac{a+1}{2^m})$, $a\in\{0,1,\ldots, 2^m-1\}$. Denote by $V$ the vertices of the tree (which are in bijection with the standard dyadic intervals) and by $I_v$ the standard dyadic interval associated with a vertex $v$ of the tree $\mathcal{T}$.

	Define the function $n:V\rightarrow \mathbb{Z}^+$ via
	\begin{equation}
		n(v) = \abs{\{x_j\in\mathbf{x}: x_j\in I_v\}}.
	\end{equation}
	The function $n$ has the property that
	\begin{equation}\label{eq:nleafadds}
		n(v) = n(\text{left leaf of $v$}) + n(\text{right leaf of $v$}).
	\end{equation}
	Find the subtree $T_P = (V_P, E_P)$ of $\mathcal{T}$ defined by the property that $n(v)>1$ for all $v\in V_P$. This induces a connected subtree by virtue of \cref{eq:nleafadds}. Deleting $T_P$ (and all its associated edges in $E_P$) from $\mathcal{T}$ gives $m$ disconnected infinite binary trees whose root nodes induce a minimal supporting partition. Any minimal supporting partition would have to exclude $T_P$, and hence $P$ thus constructed is unique.
\end{proof}

The utility of minimal supporting partitions for tuples $\mathbf{x}$ is that they directly allow us to reduce the computation of an $n$-point correlation function in the limit of fine partitions to a \emph{finite} computation. To see this we specialize henceforth to the $n$-point functions of the vacuum vector $\Omega\in\hfrak$.

Consider an $n$-tuple $\mathbf{x}$ in $[0,1)$ and let $P$ be its minimal supporting partition. Define for any tuple $\boldsymbol{\alpha} = (\alpha_1, \alpha_2, \ldots, \alpha_n)$ and any $Q\in D([0, 1])$ refining $P$, that is, $P\le Q$, the operator
\begin{equation}
	M^{\boldsymbol{\alpha}}_Q(\mathbf{x}) = \prod_{j=1}^n \phi_Q^{\alpha_j}(x_j).
\end{equation}
Since $Q$ is a supporting partition (it refines the minimal supporting partition) each of the factors in the product commutes with the others. Therefore the expression is well-defined.

\begin{lem}
	Suppose $Q\in\dcal$ is a partition refining the minimal supporting partition $P$ of a tuple $\mathbf{x}$, that is, $P\le Q$. Then
	\begin{equation}
		\langle \Omega_P, M^{\boldsymbol{\alpha}}_P(\mathbf{x})\Omega_P\rangle = \langle \Omega_Q, M^{\boldsymbol{\alpha}}_Q(\mathbf{x}) \Omega_Q\rangle,
	\end{equation}
	where $[\Omega_P] = [\Omega_Q]$.
\end{lem}

\begin{proof}
	The first observation we make is that $\Psi^P_Q$ acts in a simple way on ascending operators localized on intervals $I\in Q$: let $f\in\abs{P}\to\abs{Q}$ be the planar forest connecting the objects $\abs{P}$ and $\abs{Q}$ corresponding to the isometry $\Psi^P_Q$ and note
	\begin{equation}
		(\Psi^P_Q)^\dag(\mu^\alpha_I)\Psi^P_Q = (\lambda_\alpha)^{d_f(I,J)-1}\mu^\alpha_J,
	\end{equation}
	where $d_f(I,J)$ is the number of edges in the planar forest connecting the leaf node associated with $I$ to the node associated with its corresponding root $J$. If $P$ and $Q$ are supporting partitions for $\mathbf{x}$ then the intervals in $P$ (respectively, $Q$) containing the elements $x_j$ of the tuple $\mathbf{x}$ belong to disconnected components of the planar forest $f$.
	Denote these intervals by $I_j$ (that is, one for each $x_j$) and their corresponding roots by $J_j$.	The ascending operators $\mu_I^\alpha$ and $\mu_J^\alpha$ before and after the action of $\Psi^P_Q$ all commute and we have that
	\begin{equation}
		(\Psi^P_Q)^\dag\biggl(\prod_j\mu^\alpha_{I_j}\biggr)\Psi^P_Q = \prod_j(\lambda_\alpha)^{d_f(I_j,J_j)-1}\mu^\alpha_{J_j}.
	\end{equation}
	Noting that $(\lambda_\alpha)^{d_f(I,J)-1} = (\lambda_\alpha)^{\log_2(|J|)-\log_2(|I|)}$ and taking expectation values gives us the result.
\end{proof}

We have assembled enough information to prove the following theorem.

\begin{thm}\label{thm:npt}
	Let $\mathbf{x}$ be an ordered $n$-tuple in $[0,1)$ and $\boldsymbol{\alpha}$ be an $n$-tuple in $\{0,1,\ldots,d^2-1\}^{\times n}$. Then the limit
	\begin{equation}
		C^{\alpha_1\alpha_2\cdots\alpha_n}(x_1, x_2, \ldots, x_n) = \lim_{P\le Q} \langle \Omega_Q, M^{\boldsymbol{\alpha}}_Q(\mathbf{x})\Omega_Q\rangle
	\end{equation}
	exists and is equal to
	\begin{equation}
		\langle \Omega_P, M^{\boldsymbol{\alpha}}_P(\mathbf{x}) \Omega_P\rangle,
	\end{equation}
	where $P$ is the minimal supporting partition of $\mathbf{x}$.
\end{thm}
\begin{proof}
	Write $e_Q = \langle \Omega_Q, M^{\boldsymbol{\alpha}}_Q(\mathbf{x}) \Omega_Q\rangle$. We need to argue that the net $(e_Q)$ is eventually in any neighbourhood around $\langle \Omega_P,M^{\boldsymbol{\alpha}}_P(\mathbf{x})\Omega_P\rangle$. But this is immediate since there always exists $R\in D([0, 1])$ such that $P\le R$ and $Q\le R$: for any partition $S$ refining $R$ we have that $e_S = e_P$, that is, $e_Q$ is eventually equal to $e_P$.
\end{proof}

The previous theorem tells us that an arbitrary $n$-point function makes sense and, further, is computable in terms of operators on a finite-dimensional Hilbert space.

\begin{cor}
	Let $f$ be an element of Thompson's group $F$ (respectively, $T$), and let $|f\rangle = U(f)\Omega \in \mathcal{H}$ be the vector in the unitary representation afforded by $V$ resulting from applying $f$. Suppose $\mathbf{x}$ is an ordered $n$-tuple in $[0,1)$ and let $\boldsymbol{\alpha}$ be an $n$-tuple in $\{0,1,\ldots,d^2-1\}^{\times n}$. Then the limit
	\begin{equation}
		C^{\alpha_1\alpha_2\cdots\alpha_n}_{|f\rangle}(x_1, x_2, \ldots, x_n) = \lim_{R\le Q} \langle f_Q, M^{\boldsymbol{\alpha}}_Q(\mathbf{x}) f_Q\rangle
	\end{equation}
	exists and is equal to
	\begin{equation}
		\langle \Omega_{P'}, U(f)^\dag {T^{f(P')}_R}^\dag M^{\boldsymbol{\alpha}}_R(\mathbf{x})T^{f(P')}_R U(f) \Omega_{P'}\rangle,
	\end{equation}
	where $P$ is the minimal supporting partition of $\mathbf{x}$, $P'\ge P$ is good for $f$, and $R$ refines both $P$ and $f(P')$. (Here we use the notation $|f_Q\rangle$ for a representation of $|f\rangle$ on partition $Q$.)
\end{cor}

\section{Short-Distance Behaviour of the $n$-Point Functions}\label{sec:shortdistance}
Many of the properties of $C^{\alpha_1\alpha_2\cdots\alpha_n}(x_1, x_2, \ldots, x_n)$ are immediate consequences of the formula in \cref{thm:npt}. The first important result concerns the short-distance behaviour of the two-point function $C^{\alpha\beta}(x,y)$. To understand this we focus first on the case where $x$ and $y$ are dyadic, in which case we can express them in binary as
\begin{equation}
	\begin{split}
		x &= 0.x_{-1}x_{-2}\cdots x_{-m}, \quad\text{and} \\
		y &= 0.y_{-1}y_{-2}\cdots y_{-n},
	\end{split}
\end{equation}
where $x_j\in\{0,1\}$ and $y_j \in \{0,1\}$. (Such expansions are assumed to have an infinite sequence of trailing zeroes.) Without loss of generality we assume that $n>m$. The minimal supporting partition for the pair $(x,y)$ is easy to derive: first express $x = \bar{x} + x'$ and $y = \bar{x} + y'$, where $\bar{x}$ contains the first $l$ digits of the binary expansions of $x$ and $y$ which are in common. Now recursively subdivide the interval $[0,1]$ according to the following recipe: set $I\leftarrow [0,1]$ and $j\leftarrow -1$ and repeat steps (1) and (2) while $j \ge -l-1$:
\begin{enumerate}
	\item subdivide $I$ into $I=I_0\cup I_1$ and set $I\leftarrow I_{\overline{x}_{j}}$, where $\overline{x}_{j}$ is the $j$-th digit of $\overline{x}$,
	\item set $j\leftarrow j-1$.
\end{enumerate}
The subdivisions carried out via this procedure induce a standard dyadic partition $P$ which is minimal for the pair $(x,y)$. Note that at the final iteration $x$ and $y$ are located in neighbouring intervals of length $2^{-l-1}$. Indeed, the two intervals separating $x$ and $y$ are none other than $I = [\bar{x},\bar{x}+\frac{1}{2^{l+1}})$ and $I'=[\bar{x}+\frac{1}{2^{l+1}}, \bar{x}+\frac{1}{2^{l}})$.
Now that we have the minimal separating partition, we can immediately apply \cref{thm:npt} to deduce the two-point function
\begin{equation}
	C^{\alpha\beta}(x,y) = \langle \Omega_P, (\lambda_\alpha^{-l-1}\mu^\alpha_I)(\lambda_\beta^{-l-1}\mu^\beta_{I'}) \Omega_P\rangle.
\end{equation}
By making use of the structure constants for $\star$ we can explicitly evaluate this expression. This is summarised in the following lemma.

\begin{lem}
	Let $0\le x<y <1$ be two dyadic fractions and let $\alpha,\beta \in \{0,1,\ldots, d^2-1\}$. Write $x = \bar{x} + x'$ and $y = \bar{x} + y'$, where $\bar{x}$ contains the first $l$ digits of the binary expansions of $x$ and $y$ which are in common. Then
	\begin{equation}
		C^{\alpha\beta}(x,y) = \sum_{\gamma=0}^{d^2-1}\lambda_\gamma^{-1} D(x,y)^{\log_2(\lambda_\alpha)+\log_2(\lambda_\beta)-\log_2(\lambda_\gamma)}{f^{\alpha\beta}}_{\gamma}\langle\Omega_{[0,1]}, \mu^\gamma \Omega_{[0,1]}\rangle,
	\end{equation}
	where $D(x,y) = 2^{-l-1}$ is the \emph{coarse-graining distance} between $x$ and $y$.
\end{lem}

\begin{proof}
	Start with the expression
	\begin{equation}
		C^{\alpha\beta}(x,y) = \langle \Omega_P,(\lambda_\alpha^{-l-1}\mu^\alpha_I)(\lambda_\beta^{-l-1}\mu^\beta_{I'}) \Omega_P\rangle.
	\end{equation}
	Since the intervals $I = [\overline{x},\overline{x}+\frac{1}{2^{l+1}})$ and $I'=[\overline{x}+\frac{1}{2^{l+1}}, \overline{x}+\frac{1}{2^{l}})$ are neighbours we can exploit the $\star$ operation to evaluate this expression on the coarse-grained partition $P'$ where the neighbouring intervals $I$ and $I'$ are joined to the interval $I_{\overline{x}}$ of length $l$:
	\begin{equation}
		C^{\alpha\beta}(x,y) = \sum_{\gamma=0}^{d^2-1}(\lambda_\alpha\lambda_\beta)^{-l-1} {f^{\alpha\beta}}_{\gamma}\langle\Omega_{P'}, \mu^\gamma_{I_{\overline{x}}} \Omega_{P'}\rangle.
	\end{equation}
	This expression is easy to simplify via the action of the CP map $E$:
	\begin{equation}
		C^{\alpha\beta}(x,y) = \sum_{\gamma=0}^{d^2-1}{f^{\alpha\beta}}_{\gamma}(\lambda_\alpha\lambda_\beta)^{-l-1}\lambda_\gamma^l \langle\Omega_{[0,1]}, \mu^\gamma \Omega_{[0,1]}\rangle.
	\end{equation}
	Now write $l+1=-\log_2(D(x,y))$: we finally obtain
	\begin{equation}
		\begin{split}
			&C^{\alpha\beta}(x,y)\\ &\quad= \sum_{\gamma=0}^{d^2-1}\lambda_\gamma^{-1} D(x,y)^{\log_2(\lambda_\alpha)+\log_2(\lambda_\beta)-\log_2(\lambda_\gamma)}{f^{\alpha\beta}}_{\gamma}\langle\Omega_{[0,1]}, \mu^\gamma \Omega_{[0,1]}\rangle.\qedhere
		\end{split}
	\end{equation}
\end{proof}

\begin{rem}
	When $x$ and $y$ are a \emph{standard dyadic pair}, that is, $x = \frac{a}{2^l}$ and $y = \frac{a+1}{2^l}$, with $l\in\naturals$ and $a \in \{0,1,\ldots, 2^l-1\}$, then $D(x,y) = \abs{x-y}$, so that we can rewrite
	\begin{equation}\label{eq:2ptdyadicpair}
		\begin{split}
			C^{\alpha\beta}(x,y) &= \sum_{\gamma=0}^{d^2-1}\lambda_\gamma^{-1} \abs{x-y}^{\log_2(\lambda_\alpha)+\log_2(\lambda_\beta)-\log_2(\lambda_\gamma)}\\
			&\qquad\qquad\cdot {f^{\alpha\beta}}_{\gamma}\langle\Omega_{[0,1]}, \mu^\gamma \Omega_{[0,1]}\rangle.
		\end{split}
	\end{equation}
\end{rem}

In the context of conformal field theory an expression such as \cref{eq:2ptdyadicpair} for standard dyadic pairs is especially suggestive. We therefore make the following prototype definition for the analogue of the scaling dimension.

\begin{defn}
	The \defemph{scaling dimension} $h_\alpha$ for the field $\phi^\alpha(x)$ is
	\begin{equation}
		h_\alpha = -\Re\log_2(\lambda_\alpha).
	\end{equation}
\end{defn}

Contrary to the situation in conformal field theory, there is no reason to expect that, in general,
\begin{equation}\label{eq:2pointCFTcorr}
	C^{\alpha\beta}(x,y) \sim C^{\alpha\beta} D(x,y)^{-2h}
\end{equation}
only when $h_\alpha = h = h_\beta$. We hence promote \cref{eq:2pointCFTcorr} to a \emph{necessary condition} for the existence of a physical continuum limit.

Note also that, in general, the semicontinuous correlation functions for a tree state may be discontinuous and asymmetric.

\section{Fusion Rules and the Operator Product Expansion}\label{sec:ope}

So far we have studied the two-point correlation function. Now we look at the three-point function in an attempt to obtain an analogue of the operator product expansion.

In general, a three-point function is of the form
\begin{equation}
	C^{\alpha\beta\gamma}(x,y,z) = \langle\Omega, \widehat{\phi}^{\alpha}(x)\widehat{\phi}^{\beta}(y)\widehat{\phi}^{\gamma}(z) \Omega\rangle
\end{equation}
for $x,y,z\in[0,1)$ with $x<y<z$.

We can compute this correlation function by first finding the minimal supporting partition $P$ for $(x,y,z)$ and setting
\begin{equation}
	C^{\alpha\beta\gamma}(x,y,z) = \langle\Omega_P, (\lambda_{\alpha}^{-l-1}\mu^{\alpha}_I)(\lambda_{\beta}^{-m-1}\mu^{\beta}_\mathrm{J})(\lambda_{\gamma}^{-n-1}\mu^{\gamma}_\mathrm{K}) \Omega_P\rangle,
\end{equation}
where $I$, $J$, and $K$ are the intervals containing $x$, $y$, and $z$, respectively.

To calculate this expression note that we can exploit the formulas we already have for the two-point function. The important observation here is that when $d_T(x,y) < d_T(y,z)$ we can first fuse operators $\mu^{\alpha}$ and $\mu^{\beta}$ resulting in some linear combination of the $\mu^{\gamma'}$ and then we fuse these with $\mu^\gamma$. Correspondingly, if $d_T(x,y) > d_T(y,z)$ we first fuse the last two, then fuse on the first operator. Thus, the three-point function is completely determined by knowledge of the fusion coefficients ${f^{\alpha\beta}}_\gamma$. The observation is also particularly reminiscent of the operator product expansion, or OPE. We exemplify this by promoting it to a prototype definition.

\begin{defn}
	Given formal \defemph{primary fields} $\phi^\alpha(x)$, $\alpha = 0, 1, \ldots, d^2-1$, the \emph{formal} short-distance expansion
	\begin{equation}
		  \widehat{\phi}^\alpha(x)\widehat{\phi}^\beta(y) \sim \sum_{\gamma=0}^{d^2-1} {f^{\alpha\beta}}_{\gamma}D(x,y)^{h_\gamma-h_\alpha-h_\beta} \widehat{\phi}^\gamma(y)
	\end{equation}
	is called the \defemph{operator product expansion}.
\end{defn}

Here the $\sim$ means that the expression only makes sense in a correlation function, and that oscillatory behaviour is neglected, that is, we only study the divergence up to an overall absolute value sign.

A crucial role is played by the structure of the dimensions $h_\alpha$ as they control, via the quantity $h_\gamma-h_\alpha-h_\beta$, the divergence of the $n$-point correlation functions as $x\rightarrow y$.

The fusion coefficients ${f^{\alpha\beta}}_{\gamma}$ determine the structure of the three-point function. In particular, whether ${f^{\alpha\beta}}_{\gamma}=0$ or not determines whether a given correlation function is nontrivial or not. To this end we introduce the three-index tensor
\begin{equation}
	{N^{\alpha\beta}}_\gamma = \begin{cases}
		1 \quad \text{if ${f^{\alpha\beta}}_{\gamma}\not=0$,}\\ 0\quad \text{otherwise.}
	\end{cases}
\end{equation}
This tensor can be used to construct an (in general) non-associative and non-commutative algebra $\afrak$ over $\integers$. As a set we define $\afrak$ to be the lattice
\begin{equation}
	\mathfrak{A} = \mathbb{Z}^{d^2},
\end{equation}
and we choose some basis $\{\phi^\alpha\}_{\alpha\in I}$, where $I=\{0,1,\ldots,d^2-1\}$, and introduce the product operator $\star$ via
\begin{equation}
	\phi^\alpha\star \phi^\beta = \sum_{\gamma\in I} {N^{\alpha\beta}}_\gamma \phi^\gamma.
\end{equation}

Usually the algebra $\afrak$ will be neither associative nor commutative. However, in special cases, it can be the case that ${N^{\alpha\beta}}_\gamma$ ends up satisfying these additional constraints. In this case $\afrak$ becomes a so-called \emph{fusion ring}. We can obtain a representation for the fusion ring via the commuting matrices $N^\alpha$ with matrix elements
\begin{equation}
	[N^\alpha]_{\beta\gamma} = {N^{\alpha\beta}}_\gamma.
\end{equation}

\section{The Action of Thompson's Groups $F$ and $T$\\on $n$-Point Functions}\label{sec:thompsonaction}
The analogy between CFT and quantum mechanics symmetric under Thompson's groups $F$ and $T$ manifests itself most strongly when considering how $n$-point
functions transform under an element $f$ of one of Thompson's groups. Here we discuss the $n$-point correlation function with respect to the vacuum vector $\Omega$ and its transformed version $U(f)\Omega$.

\begin{thm}\label{thm:npointaction}
Let $f\in T$ be an element of Thompson's group $T$ and $U(f)$ its unitary representation. Then the action of $T$ on $\hfrak$ in terms of $n$-point functions is
	\begin{equation}\label{eq:nptaction}
		C^{\boldsymbol{\alpha}}(x_1, x_2, \ldots, x_n) = \prod_{j=1}^n \left(\frac{df}{dx}\bigg|_{x = x_j}\right)^{-h_{\alpha_j}}C^{\boldsymbol{\alpha}}_{|f\rangle}(f(x_1), f(x_2), \ldots, f(x_n)),
	\end{equation}
	where the limit in the derivative is taken above via $x\rightarrow x_j+\epsilon$.
\end{thm}

\begin{proof}
	Let $P$ be the minimal supporting partition for the tuple $(x_1,\dotsc, x_n)$. Choose a refinement $P'$ which is good for $f$ and choose $R$ which refines both $P'$ and $f(P')$. Then
	\begin{equation}
		(T^{P'}_R)^\dag M^{\boldsymbol{\alpha}}_{R}(\mathbf{x})T^{P'}_R = M^{\boldsymbol{\alpha}}_{P'}(\mathbf{x})
	\end{equation}
	and the left-hand side of \cref{eq:nptaction} is directly equal to
	\begin{equation}
		C^{\boldsymbol{\alpha}}(x_1, x_2, \ldots, x_n) = \langle\Omega_{P'}, M^{\boldsymbol{\alpha}}_{P'}(\mathbf{x})   \Omega_{P'}\rangle = \langle\Omega_{R}, M^{\boldsymbol{\alpha}}_{R}(\mathbf{x})   \Omega_{R}\rangle.
	\end{equation}
	Now we compare left and right-hand sides: the correlation function on the right-hand side of \cref{eq:nptaction} is the expectation value of
	\begin{equation}
		M^{\boldsymbol{\alpha}}_{f(P')}(f(x_1), f(x_2), \ldots, f(x_n)) = \prod_{j=1}^n (\lambda_{\alpha_j})^{\log_2(|f(I_j)|)}\mu_{f(I_j)}^{\alpha_j},
	\end{equation}
	with respect to $U(f)\Omega_{P'}$ (noting that $f(P')$ is a supporting partition for $(f(x_1), f(x_2), \ldots, f(x_n))$). Rewriting this expression as
	\begin{equation}
		\begin{split}
			&M^{\boldsymbol{\alpha}}_{f(P')}(f(x_1), f(x_2), \ldots, f(x_n))\\&\quad= \prod_{j=1}^n (\lambda_{\alpha_j})^{\log_2(|f(I_j)|)-\log_2(|I_j|)}(\lambda_{\alpha_j})^{\log_2(|I_j|)}\mu_{f(I_j)}^{\alpha_j}
		\end{split}
	\end{equation}
	and taking the expectation value with respect to $U(f)\Omega_{P'}$ gives us
	\begin{equation}
		\prod_{j=1}^n (\lambda_{\alpha_j})^{\log_2(|f(I_j)|)-\log_2(|I_j|)} C^{\boldsymbol{\alpha}}_{|f\rangle}(f(x_1), f(x_2), \ldots, f(x_n)).
	\end{equation}
	Now the we can calculate the length of the interval $f(I_j)$ as
	\begin{equation}
		|f(I_j)| = \left(\frac{df}{dx} \bigg|_{x = x_j}\right) |I_j|.
	\end{equation}
	(Here the derivative is defined with a limit from the right so as to avoid singularities when $x_j$ is at a breakpoint.)
	Taking logs and exchanging exponents using the identity $a^{\log(b)} = b^{\log(a)}$ gives us the result.
\end{proof}

By substituting $x_j = f^{-1}(z_j)$ we can rewrite this result in a somewhat more useful form,
\begin{equation}
	\begin{split}
		&C^{\boldsymbol{\alpha}}_{|f\rangle}(z_1, z_2, \ldots, z_n)\\
		&\quad= \prod_{j=1}^n \left(\frac{df}{dz}\bigg|_{z = f^{-1}(z_j)}\right)^{h_{\alpha_j}} C^{\boldsymbol{\alpha}}(f^{-1}(z_1), f^{-1}(z_2), \ldots, f^{-1}(z_n)).
	\end{split}
\end{equation}

\begin{rem}
	This corollary tells us that knowledge of the vacuum $n$-point functions $\langle \Omega, {\phi}^{\alpha_1}(z_1){\phi}^{\alpha_2}(z_2)\cdots {\phi}^{\alpha_n}(z_n) \Omega\rangle$ alone is enough to calculate the $n$-point functions with respect to any transformed state $|f\rangle = U(f)|\Omega\rangle$:
	\begin{multline}
		\langle f|{\phi}^{\alpha_1}(z_1){\phi}^{\alpha_2}(z_2)\cdots {\phi}^{\alpha_n}(z_n)|f\rangle = \\ \prod_{j=1}^n \left(\frac{df}{dz}\bigg|_{z = f^{-1}(z_j)}\right)^{h_{\alpha_j}}\langle \Omega, {\phi}^{\alpha_1}(f^{-1}(z_1)){\phi}^{\alpha_2}(f^{-1}(z_2))\cdots {\phi}^{\alpha_n}(f^{-1}(z_n)) \Omega\rangle.
	\end{multline}
\end{rem}

In the case where our unitary representation is determined by a planar perfect tangle we deduce that the correlation function is $\mathit{PSL}(2,\mathbb{Z})$ invariant because $|f\rangle = \Omega$ for $f\in\mathit{PSL}(2,\mathbb{Z})$.

\section{Application: Spin System}\label{sec:example1}

Here we apply the formalism of the previous sections to a simple example quantum spin system comprised of a lattice of \emph{qutrits}, that is,
\begin{equation}
	\mathfrak{H}_N = \bigotimes_{j=0}^{2^m-1} \complexes^3,
\end{equation}
where, as usual, we set $N=2^m$. As usual, we choose the perfect tensor $V\colon\complexes^3\to\complexes^3\otimes\complexes^3$ from \cref{exmp:V-spin-system} given by
\begin{equation}
	\langle jk|V|l\rangle = \begin{cases}
		0 & \text{if $j=k$, $k=l$, or $l=j$,}\\
		\frac{1}{\sqrt{2}} & \text{otherwise.}
	\end{cases}
\end{equation}
The ascending operator $E$ constructed from this perfect tensor has the three eigenvalues
\begin{equation}
	\lambda_1=1, \qquad \lambda_\alpha=-\frac{1}{2}, \qquad \lambda_\beta=\frac{1}{2}.
\end{equation}
$\lambda_1=1$ has the (right) eigenvector $\mu^1=I$; $\lambda_\alpha=-\frac{1}{2}$ has eigenvectors
\begin{gather}
	\mu^{\delta^1}=\begin{pmatrix}
		-1&0&0\\
		0&0&0\\
		0&0&1
	\end{pmatrix}, \qquad
	\mu^{\alpha^1}=\begin{pmatrix}
		0&0&0\\
		0&0&-1\\
		0&1&0
	\end{pmatrix}, \qquad
	\mu^{\alpha^2}=\begin{pmatrix}
		0&0&-1\\
		0&0&0\\
		1&0&0
	\end{pmatrix},\\
	\mu^{\delta_2}=\begin{pmatrix}
		-1&0&0\\
		0&1&0\\
		0&0&0
	\end{pmatrix}, \qquad
	\mu^{\alpha^3}=\begin{pmatrix}
		0&-1&0\\
		1&0&0\\
		0&0&0
	\end{pmatrix};
\end{gather}
$\lambda_{\beta}=\frac{1}{2}$ has eigenvectors
\begin{equation}
	\mu^{\beta^1}=\begin{pmatrix}
		0&0&0\\
		0&0&1\\
		0&1&0
	\end{pmatrix}, \qquad
	\mu^{\beta^2}=\begin{pmatrix}
		0&0&1\\
		0&0&0\\
		1&0&0
	\end{pmatrix}, \qquad
	\mu^{\beta^3}=\begin{pmatrix}
		0&1&0\\
		1&0&0\\
		0&0&0
	\end{pmatrix}.
\end{equation}
They result in the fusion rules
\begin{center}
	\vspace*{0.3\baselineskip}
	\begin{tabular}{|Sc|Sc|Sc|Sc|Sc|Sc|Sc|Sc|Sc|Sc|} \hline
		\rowcolor{gray}${\times}$&$1$&$\delta^1$&$\delta^2$&$\beta^1$&$\beta^2$&$\beta^3$&$\alpha^1$&$\alpha^2$&$\alpha^3$\\ \hline
	 	\cellcolor{gray}$1$&$1$&$\delta^1$&$\delta^2$&$\beta^1$&$\beta^2$&$\beta^3$&$\alpha^1$&$\alpha^2$&$\alpha^3$\\ \hline
		\cellcolor{gray}$\delta^1$&$\delta^1$&$\Sigma$&$\Sigma$&$\beta^1$&$0$&$\beta^3$&$\alpha^1$&$0$&$\alpha^3$\\ \hline
		\cellcolor{gray}$\delta^2$&$\delta^2$&$\Sigma$&$\Sigma$&$\beta^1$&$\beta^2$&$0$&$\alpha^1$&$\alpha^2$&$0$\\ \hline
		\cellcolor{gray}$\beta^1$&$\beta^1$&$\beta^1$&$\beta^1$&$\Sigma$&$\beta^3$&$\beta^2$&$0$&$\alpha^3$&$\alpha^2$\\ \hline
		\cellcolor{gray}$\beta^2$&$\beta^2$&$0$&$\beta^2$&$\beta^3$&$\Sigma$&$\beta^1$&$\alpha^3$&$0$&$\alpha^1$\\ \hline
		\cellcolor{gray}$\beta^3$&$\beta^3$&$\beta^3$&$0$&$\beta^2$&$\beta^1$&$\Sigma$&$\alpha^2$&$\alpha^1$&$0$\\ \hline
		\cellcolor{gray}$\alpha^1$&$\alpha^1$&$\alpha^1$&$\alpha^1$&$0$&$\alpha^3$&$\alpha^2$&$\Sigma$&$\beta^3$&$\beta^2$\\ \hline
		\cellcolor{gray}$\alpha^2$&$\alpha^2$&$0$&$\alpha^2$&$\alpha^3$&$0$&$\alpha^1$&$\beta^3$&$\Sigma$&$\beta^1$\\ \hline
		\cellcolor{gray}$\alpha^3$&$\alpha^3$&$\alpha^3$&$0$&$\alpha^2$&$\alpha^1$&$0$&$\beta^2$&$\beta^1$&$\Sigma$\\ \hline
	\end{tabular}
	\vspace*{0.3\baselineskip}
\end{center}
with $\Sigma=1+\delta^1+\delta^2$.

From the eigenvalues we get $h_1=0$ and $h_\alpha=h_\beta=1$. For the OPE, we give the two examples
\begin{align}
	\widehat{\phi}^{\delta^1}(x)\widehat{\phi}^{\delta^2}(y) &\sim -\frac{1}{6} D(x, y)^{-2} \widehat{\phi}^1(y) -\frac{1}{3} D(x, y)^{-1} \widehat{\phi}^{\delta^1}(y) - \frac{1}{3} D(x, y)^{-1}\widehat{\phi}^{\delta^2}(y),\\
	\widehat{\phi}^{\beta^2}(x)\widehat{\phi}^{\alpha^3}(y) &\sim \frac{1}{3} D(x, y)^{-1} \widehat{\phi}^{\alpha^1}(y).
\end{align}

\section{Application: the Fibonacci Lattice}\label{sec:example2}

We now illustrate the formalism developed in the previous sections in terms of a tree state defined for the Fibonacci category $\mathcal{F}$. The computations in this section were performed using the \textit{TriCats} package \cite{Stiegemann2019,Stiegemann2018}.

The fusion ring of $\mathcal{F}$ is generated by the two elements $1$ and $\tau$ and fusion rules
\begin{align}
	1 \times 1&= 1\\
	1 \times \tau &= \tau\\
	\tau\times\tau&=1+\tau.
\end{align}
$\mathcal{F}$ is a trivalent category with $\dim\mathcal{C}_4=2$ and $d=\frac{1}{2}(1\pm\sqrt{5})$, and it is a special case of an $\mathit{SO}(3)_q$ category with $q=4$. $\mathcal{C}_4$ is spanned by the two vectors
\begin{equation}\label{eq:fibbasis}
	\begin{tikzpicture}[scale=0.5]
    \draw (0, 0) .. controls (1, 1) .. (0, 2);
    \draw (2, 0) .. controls (1, 1) .. (2, 2);
  \end{tikzpicture}\,,\quad%
	\begin{tikzpicture}[scale=0.5]
    \draw (0, 0) .. controls (1, 1) .. (2, 0);
    \draw (0, 2) .. controls (1, 1) .. (2, 2);
  \end{tikzpicture}\,.
\end{equation}

We use a modification of the trivalent vertex, effectively doubling lines and replacing the trivalent vertex with
\begin{equation}
	V=\begin{tikzpicture}[scale=0.6,baseline=10mm]
    \draw[rounded corners] (2.5, 3.5) -- (2.5, 2) -- (1.875, 1.375);
    \draw (2.25, 1) -- (2.75, 0.5);
    \draw[rounded corners] (2.5, 0.75) -- (1.25, 2) -- (-0.25, 0.5);
    \draw[draw=white,double=black,double distance=0.4pt,line width=3pt] (2, 3.5) -- (2, 2) -- (0.5, 0.5);
		\draw (2, 2) -- (2.2, 1.8);
    \draw[draw=white,double=black,double distance=0.4pt,line width=3pt] (2.1, 1.9) -- (3.5, 0.5);
  \end{tikzpicture}.
\end{equation}
The braiding appearing in $V$ is given by
\begin{equation}
	\begin{tikzpicture}[scale=0.5,baseline=3.5mm]
    \draw (0, 2) -- (2, 0);
    \draw[draw=white,double=black,double distance=0.4pt,line width=3pt] (0, 0) -- (2, 2);
  \end{tikzpicture}
	=
	\begin{tikzpicture}[scale=0.5,baseline=3.5mm]
    \draw (0, 0) .. controls (1, 1) .. (0, 2);
    \draw (2, 0) .. controls (1, 1) .. (2, 2);
  \end{tikzpicture}
	+e^{4i\pi/5}
  \begin{tikzpicture}[scale=0.5,baseline=3.5mm]
    \draw (0, 0) .. controls (1, 1) .. (2, 0);
    \draw (0, 2) .. controls (1, 1) .. (2, 2);
  \end{tikzpicture}.
\end{equation}
In the basis (\ref{eq:fibbasis}), the ascending operator has matrix elements
\begin{equation}
	\begin{pmatrix}
		1 & \frac{1}{2}(3-\sqrt{5})\\
		0 & \frac{1}{2}(3-\sqrt{5})
	\end{pmatrix}
\end{equation}
and eigenvalues
\begin{equation}
	\lambda_1=1, \qquad \lambda_\tau=\frac{1}{2}\bigl(3-\sqrt{5}\bigr).
\end{equation}
The fusion coefficients are given by
\begin{equation}
	f^1=\begin{pmatrix}
		1 & 0\\
		0 & \frac{1}{2}(3-\sqrt{5})
	\end{pmatrix},\quad
	f^\tau=\begin{pmatrix}
		0 & \frac{1}{2}(3-\sqrt{5})\\
		\sqrt{5}-2 & 5-2\sqrt{5}
	\end{pmatrix}.
\end{equation}
The OPE then gives us the short-distance behaviour
\begin{align}
	\widehat{\phi}^1(x)\widehat{\phi}^1(y) &\sim \widehat{\phi}^1(y), \\
	\widehat{\phi}^1(x)\widehat{\phi}^\tau(y) &\sim \frac{1}{2}(3-\sqrt{5}) \widehat{\phi}^\tau(y), \\
	\widehat{\phi}^\tau(x)\widehat{\phi}^\tau(y) &\sim (\sqrt{5}-2)D(x, y)^{-2h_\tau}	\widehat{\phi}^1(y) + (5-2\sqrt{5})D(x, y)^{-h_\tau}\widehat{\phi}^\tau(y)
\end{align}
with $h_\tau=-\log_2\bigl( \frac{1}{2}(3-\sqrt{5}) \bigr)\approx 1.388$.
We obtain a representation of the fusion ring via the matrices
\begin{equation}
	N^1=\begin{pmatrix}
		1 & 0\\
		0 & 1
	\end{pmatrix},\quad
	N^\tau=\begin{pmatrix}
		0 & 1\\
		1 & 1
	\end{pmatrix}.
\end{equation}
It is very interesting to note that these matrices again describe the fusion rules of the Fibonacci category!

\section{A Few Technical Observations Concerning Trees}
\label{sub:trees}

Here we collect together some basic observations concerning trees and the circle. Our systems are thought of as living on the circle $\sircle^1$ which is taken to be the interval $[0,1]$ with $0$ and $1$ identified. It is rather convenient to express points $x\in \sircle^1$ in terms of their binary expansions, that is, we write
\begin{equation}
	x = 0.x_{-1}x_{-2}\cdots x_{-l}, \quad x_{-j} \in \{0,1\}, \quad j = 1, 2, \ldots, l,
\end{equation}
to stand for the representation
\begin{equation}
	x = \sum_{j=1}^l \frac{x_{-j}}{2^{j}},
\end{equation}
for some $l\in\integers$.

We introduce the somewhat baffling operation $\ominus$ on $x$ and $y$ in $\sircle^1$ according to
\begin{equation}
	y\ominus x = \sum_{j=1}^l \frac{(y_{-j}-x_{-j})\, \text{mod $2$}}{2^{j}},
\end{equation}
where the arithmetic in the term $y_{-j}-x_{-j}$ is carried out in the finite field $\mathbb{F}_2$ and then embedded back in $\mathbb{R}$ in the natural way.
We pad out the expansions of $x$ or $y$ with zeros as necessary. $x\ominus y$ corresponds to bitwise XOR on the binary digits of $x$ and $y$.

We identify partitions of $\sircle^1$ with trees in the standard way:
\begin{align}
	\{[0,1)\} &\leftrightarrow \mathcal{T}_0 \\
	\{[0,\tfrac12), [\tfrac12,1)\} &\leftrightarrow \mathcal{T}_1 \\
	\{[0,\tfrac14), [\tfrac14,\tfrac12), [\tfrac12,\tfrac34), [\tfrac34,1)\} &\leftrightarrow \mathcal{T}_2 \\
	&\ \vdots
\end{align}
Here $\mathcal{T}_l$ is the regular binary tree with $2^l$ leaves. Each interval in the partition is identified with a leaf of $\mathcal{T}_l$. The nonnegative integer $l$ is called the \emph{level}.

We can alternatively specify a standard dyadic interval $[x,y)= [\tfrac{j}{2^l}, \tfrac{j+1}{2^l})$ by simply writing out the left end point in binary to $l$ significant digits:
\begin{equation}
	[\tfrac{j}{2^l}, \tfrac{j+1}{2^l}) \leftrightarrow 0.x_{-1}x_{-2}\cdots x_{-l}.
\end{equation}
Here the number $l$ of significant digits, the \emph{level}, tells us what kind of standard dyadic interval it is: once you know $x$ you can get $y$ by adding $1/2^l$. Here is a simple example:
\begin{equation}
	[\tfrac{13}{32}, \tfrac{14}{32}) \leftrightarrow 0.01101.
\end{equation}
In this way we can label the leaves of $\mathcal{T}_l$ with binary expansions with exactly $l$ significant digits.

We introduce the following \emph{tree metric} on the leaves of the regular binary tree $\mathcal{T}_l$ as follows. Let $x$ and $y$ be the binary labels corresponding to two leaves of $\mathcal{T}_l$ and recursively define
\begin{equation}
	d_T(x,y) = 1+ d_T(x^{(1)},y^{(1)})
\end{equation}
and
\begin{equation}
	d_T(x,x) = 0, \quad \forall x,
\end{equation}
where
\begin{equation}
	x^{(j)} =  0.x_{-1}x_{-2}\cdots x_{-l+j},
\end{equation}
that is, by dropping the last $j$ digits of the binary expansion for $x$.
For example, if $x = 13/32$ and $y = 15/32$ we have
\begin{equation}
	d_T(0.01101, 0.01111) = 1+ d_T(0.0110, 0.0111) = 2+d_T(0.011, 0.011) = 2.
\end{equation}

\begin{lem}\label{lem:treemetric}
The tree metric between $x$ and $y$ in $\sircle^1$ labelling the leaves of $\mathcal{T}_l$ may be computed according to
\begin{equation}
	d_T(x,y) = l+1+\lfloor \log_2(y\ominus x) \rfloor.
\end{equation}
\end{lem}

As can be seen from the previous example and made rigorous in the lemma, $d_T$ counts, from the right of the binary expansions of $x$ and $y$, the leftmost position at which the digits of $x$ and $y$ are different.

\begin{proof}
	Suppose that $d_T(x,y) = j$. Then we know that $x$ and $y$ share the same first $l-j$ digits, that is,
	\begin{equation}
		x = 0.x_{-1}x_{-2}\cdots x_{-l}, \quad \text{and}\quad y=0.x_{-1}x_{-2}\cdots x_{-l+j} y_{-l+j-1}\cdots y_{-l}.
	\end{equation}
	Now notice that
	\begin{equation}
		y\ominus x = 0.00\cdots 0 (y_{-l+j-1}\oplus x_{-l+j-1})\cdots (y_{-l}\oplus x_{-l}).
	\end{equation}
	In particular, note that the digit in the $(l-j+1)$ term is $1$. Hence
	\begin{equation}
		y\ominus x = 0.00\cdots 0 1 \cdots (y_{-l}\oplus x_{-l}) = \frac{1}{2^{l-j+1}}(1+\delta),
	\end{equation}
	where $\delta \in [0,\tfrac{1}{2})$. Take logs of both sides to find
	\begin{equation}
		\log_2(y\ominus x) = -(l-j+1) + \log_2(1+\delta).
	\end{equation}
	Adding $l$ to both sides and taking the floor gives the answer.
\end{proof}

For the special case where $x=0$ and $y = x$ we have the formula
\begin{equation}
	d_T(0,x) = l+1+\lfloor \log_2(x) \rfloor.
\end{equation}

We note the following
\begin{lem}
	Let $x$ and $y$ be two $l$-digit binary numbers in $[0, 1)$ with $y\ge x$. Then
	\begin{equation}
		y\ominus x \ge y-x
	\end{equation}
	and hence
	\begin{equation}
		d_T(x,y) \ge l+1+\lfloor \log_2(\abs{y-x}) \rfloor.
	\end{equation}
\end{lem}
\begin{proof}
	First note that for $a,b\in \{0,1\}$,
	\begin{equation}
		a-b = \bigl(a-b \Mod 2\bigr) - 2\delta_{a,1}\delta_{b,0},
	\end{equation}
	so that
	\begin{equation}
		a-b \le a-b \Mod 2.
	\end{equation}
	Now
	\begin{equation}
		y\ominus x = \sum_{j=1}^l \frac{(y_{-j}-x_{-j})\Mod 2}{2^{j}} = y-x + \delta,
	\end{equation}
	where
	\begin{equation}
		\delta = 2\sum_{j=1}^l \frac{\delta_{x_{-j},1}\delta_{y_{-j},0}}{2^{j}}.
	\end{equation}
	Since $\delta$ is nonnegative we have that
	\begin{equation}
		y\ominus x \ge y-x.
	\end{equation}
	This concludes the proof.
\end{proof}


\chapter{Particles, Black Holes, and Discrete Cobordisms}
\label{cha:particles}

In this chapter we present a few new ideas built upon the concepts developed in the previous chapters.

\section{Particle Creation and Annihilation}
\label{sec:particles}

In the following we briefly describe how to include a matter source in the form of a pointlike particle in the continuous spacetime and in our discrete models.

Following Matschull \cite{Matschull1999}, we take a quick look at how point particles are described in an equal-time slice of $\mathrm{AdS}_3$. (We do not yet have a full understanding of the \emph{time}-evolution of particles in the discrete model, so we omit the same also in the case of $\mathrm{AdS}_3$.) We begin with the Poincaré disk and two geodesics intersecting at a point inside the disk, as in \cref{fig:particle}.
\begin{figure}[h]
  \includegraphics[width=0.3\textwidth]{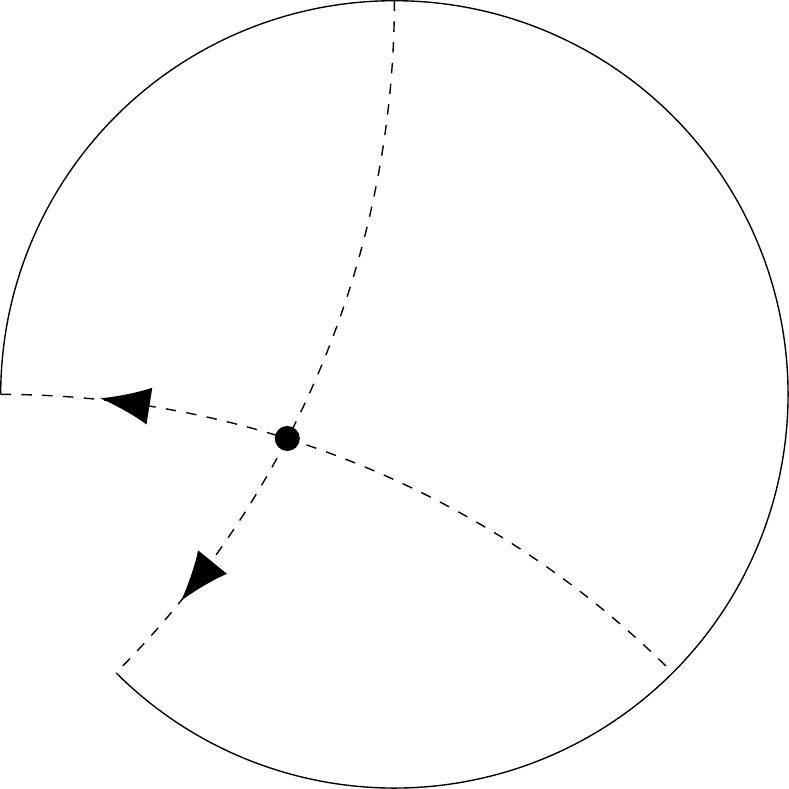}
  \caption{A Poincaré disk with one particle. The arrows indicate the orientation of the two line segments that are glued together.}
  \label{fig:particle}
\end{figure}

We cut a wedge out of the disk, bounded by the two geodesics and located between the boundary circle and the intersection. Then we identify the boundaries as indicated by the small arrows. The topology of the resulting spacetime is the same as before, except around the intersection point. We can interpret this defect of the topology as a point particle.

To transfer this to the discrete model, note that we can simply cut out wedges bounded by geodesics of the standard dyadic tessellation, but when a tensor network is laid upon the tessellation---for example, inside a cutoff region---we need to specify what happens to uncontracted tensor legs. In other words, we need to suppress anything from being attached to his leg, and it should become a dead end. This corresponds to contracting it with a 1-leg tensor, that is, a vector $\psi\in\shfrak$.

We take another step back and ask how disabling a leg is implemented on the level of trees and forests. We propose the following solution. Take the category $\annfor$ and add a single morphism $b\colon 1\to 0$. We denote this new category by $\annfor^\bullet$. We choose the pictorial representation
\begin{equation}
  b=\tikz[scale=0.2,baseline=-3mm,yscale=-1]{\draw (0, 2) -- (0, 0);
  \draw[fill] (0, 0) circle [radius=0.5];}
\end{equation}
and therefore call it the \defemph{blob}. When it is added to the leaf of a forest, it effectively disables that leaf. We further impose the relation
\begin{equation}\label{eq:blob-relation}
  (b\otimes b)\circ t=b,
\end{equation}
where $t\colon 1\to 2$ is the trivalent vertex. In pictures, this is
\begin{equation}\label{eq:blob-relation-picture}
  \begin{tikzpicture}[scale=0.3,baseline=-4mm,yscale=-1]
    \draw (0, 0) -- (1, 1) -- (2, 0);
    \draw (1, 2) -- (1, 1);

    \draw[fill] (0, 0) circle [radius=0.25];
    \draw[fill] (2, 0) circle [radius=0.25];
  \end{tikzpicture}
  \;=\;\;
  \begin{tikzpicture}[scale=0.3,baseline=-4mm,yscale=-1]
    \draw (0, 2) -- (0, 0);
    \draw[fill] (0, 0) circle [radius=0.25];
  \end{tikzpicture}\,.
\end{equation}
The relation \cref{eq:blob-relation} encodes the requirement that it should not matter whether we simply delete one part of the disk ($\,\begin{tikzpicture}[scale=0.15,baseline=-2.2mm,yscale=-1]
  \draw (0, 0) -- (1, 1) -- (2, 0);
  \draw (1, 2) -- (1, 1);

  \draw[fill] (0, 0) circle [radius=0.25];
  \draw[fill] (2, 0) circle [radius=0.25];
\end{tikzpicture}\,$), or first separate it in two halves ($\,\begin{tikzpicture}[scale=0.15,baseline=-2.4mm,yscale=-1]
  \draw (0, 0) -- (1, 1) -- (2, 0);
  \draw (1, 2) -- (1, 1);
\end{tikzpicture}\,$) and then delete both ($\,\tikz[scale=0.15,baseline=-2.3mm,yscale=-1]{\draw (0, 2) -- (0, 0);
\draw[fill] (0, 0) circle [radius=0.25];}\;\tikz[scale=0.15,baseline=-2.3mm,yscale=-1]{\draw (0, 2) -- (0, 0);
\draw[fill] (0, 0) circle [radius=0.25];}\,$).

We can now form a category $\annfor^\bullet_1$ which contains all morphisms of $\annfor^\bullet$ except the empty forest and the blob. The localization $\tcal^\bullet=\annfor^\bullet[(\annfor^\bullet_1)^{-1}]$ contains a group $T^\bullet=\tcal^\bullet(1, 1)$ which is like $T$ but with blobs. The reason that we localize with respect to $\annfor^\bullet_1$, and not all morphisms in $\annfor^\bullet$, is that once the empty diagram and the blob are included, localization makes the group $\tcal^\bullet(1, 1)$ trivial.

Under a linear functor, the condition \cref{eq:blob-relation} translates to
\begin{equation}\label{eq:blobcondition}
  \begin{tikzpicture}[scale=0.45,baseline=-7mm,yscale=-1]\footnotesize
    \def\a{1.5}

    \draw[rounded corners] (0, 1.5) -- (-\a, 0.5) -- (-\a, -0.5);
    \draw[rounded corners] (0, 1.5) -- (\a, 0.5) -- (\a, -0.5);
    \draw (0, 2) -- (0, 3);

    \draw[fill=white] (0, 1.5) circle[radius=0.55] node {$V$};
    \draw[fill=white] (-\a, -0.5) circle[radius=0.45] node {$b$};
    \draw[fill=white] (\a, -0.5) circle[radius=0.45] node {$b$};
  \end{tikzpicture}
  \;=\;\begin{tikzpicture}[scale=0.45,baseline=-7mm,yscale=-1]\footnotesize
    \draw (0, -0.5) -- (0, 3);
    \draw[fill=white] (0, -0.5) circle[radius=0.45] node {$b$};
  \end{tikzpicture}
\end{equation}
for a 3-leg tensor $V$. We close the description with our two guiding examples.

\begin{exmp}
  For $V\colon\mathbb{C}^3\to\mathbb{C}^3\otimes\mathbb{C}^3$ from \cref{exmp:V-spin-system}, the only possible values for (the image under the functor of) $b$ are
  \begin{align}
    |b_1\rangle&= \tfrac{1}{\sqrt{2}}\bigl( |0\rangle + |1\rangle + |2\rangle \bigr),\\
    |b_2\rangle&= \tfrac{1}{\sqrt{2}}\bigl( |0\rangle - |1\rangle - |2\rangle \bigr),\\
    |b_3\rangle&= \tfrac{1}{\sqrt{2}}\bigl( -|0\rangle + |1\rangle - |2\rangle \bigr),\\
    |b_4\rangle&= \tfrac{1}{\sqrt{2}}\bigl( -|0\rangle - |1\rangle + |2\rangle \bigr).
  \end{align}
  Note that these vectors are not normalized.
\end{exmp}

\begin{exmp}
  In a braided trivalent category, one possible choice of $V$ is the perfect tensor from \cref{exmp:trivalent}.
  Then \cref{eq:blobcondition} is only satisfied for $\beta=\tikz[scale=0.25,baseline=0mm]{\draw (0, 0) arc[start angle=0, end angle=180, radius=1]}\,$.
\end{exmp}

There are immediately many open questions. On the mathematical side, it would be interesting to investigate the group $T^\bullet$: Is it finitely presented? What are its generators and relations? Is it isomorphic to a group acting on the circle or another manifold? Is it related to any group acting on the bulk, that is, the Poincaré disk? On the physical side, we can ask: Can $T^\bullet$ sensibly model the time evolution of a point particle? What happens when multiple particles are created?

\section{Black Holes}
\label{sec:black-holes}

In $2+1$ dimensions there is a well-known family of black-hole solutions of Einstein's equations in the presence of a negative cosmological constant due to \citeauthor{BanadosTeitelboimZanelli1992}, and they are correspondingly called BTZ black holes \cite{BanadosTeitelboimZanelli1992,Matschull1999}. These geometries are multiply connected and may be produced from the Poincaré disk $\disk$ by dividing out a discrete subgroup of the group of isometries of $\disk$. We illustrate the simplest example here.

A tessellation $\tau_{\mathrm{BTZ}}$ for the BTZ spacetime can be built from the standard dyadic tessellation by choosing two opposite geodesics and identifying them, see \cref{fig:black-hole}.
\begin{figure}[h]
  \includegraphics[width=0.4\textwidth]{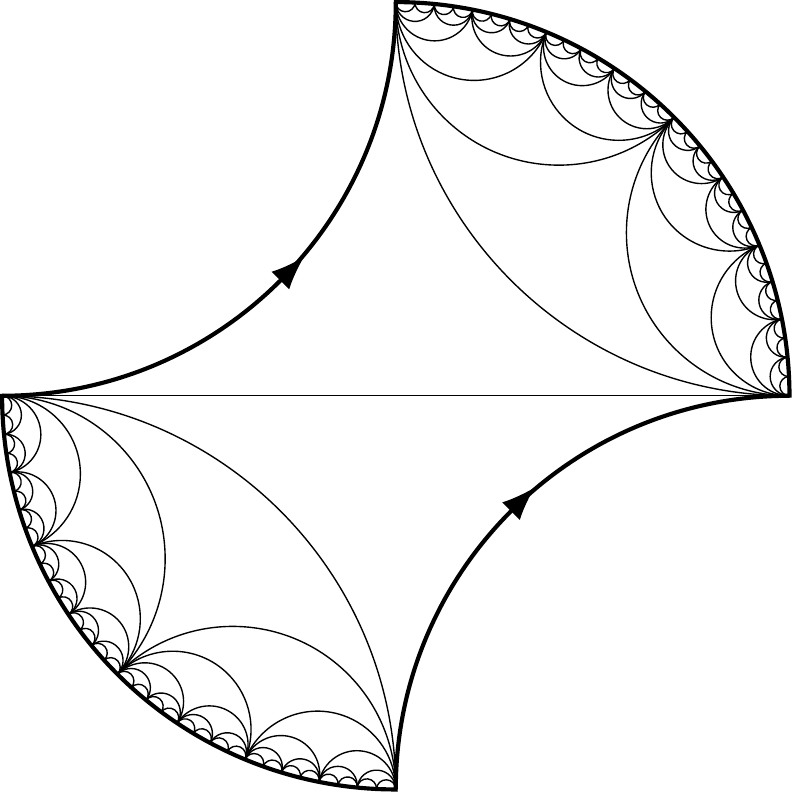}
  \caption{}
  \label{fig:black-hole}
\end{figure}

The two identified geodesics are indicated with arrows. The result of this procedure is a tessellation of the cylinder with two boundaries $A$ and $B$ at spatial infinity. The two boundaries may each be identified with $\sircle^1$.

By associating a perfect tensor $V$ with each triangle in the BTZ tessellation we obtain a corresponding tensor network.
This network can be thought of in two ways. Firstly, it may be understood as a state $\psi_{AB}$ with `geometry' given by $\tau_{\mathrm{BTZ}}$: here we are thinking of the Hilbert space of the entire system given by $\hfrak_{AB} = \hfrak_A\otimes \hfrak_B$, the semicontinuous limit built on the two boundaries $A$ and $B$ at spatial infinity. Note that $\psi_{AB}$ is \emph{not} a product state with respect to the tensor product over $A$ and $B$, it is an \emph{entangled state}. This gives rise to the second interpretation of $\psi_{AB}$, namely, as an entangled state of the two distinct subsystems $A$ and $B$. This equivalence between entanglement and geometry is the manifestation of the $\text{ER}=\text{EPR}$ proposal for our discrete version of the bulk-boundary correspondence \cite{MaldacenaSusskind2013}.

\section{Discrete Cobordisms}
\label{sec:cobordisms}

We close with a few words on cobordisms. Roughly, a cobordism $f\colon X\to Y$ is an $n$-dimensional manifold whose boundary consists of two disjoint $(n-1)$-dimensional manifolds $X$ and $Y$ \cite{BaezStay2010,Segal1988}. (The exact definition is longer.) Two typical examples for $n=2$ are the following:
\vspace*{0.2\baselineskip}
\begin{equation}
  \begin{tikzpicture}[every tqft/.style={draw,boundary lower style={draw}}]
    \node[tqft/cylinder, at={(0,0)}] {};
    \node[tqft/pair of pants, at={(4,0)}] {};
  \end{tikzpicture}
\end{equation}

\noindent In one dimension less, there are similar pictures:
\vspace*{0.2\baselineskip}
\begin{equation}
  \begin{tikzpicture}
    \draw (0, 0) -- (0, 2);
    \filldraw (0, 0) circle[radius=0.07];
    \filldraw (0, 2) circle[radius=0.07];

    \begin{scope}[xshift=4cm]
      \draw (-1, 0) to[bend left] (0, 1);
      \draw (1, 0) to[bend right] (0, 1);
      \draw (0, 1) -- (0, 2);
      \filldraw (-1, 0) circle[radius=0.07];
      \filldraw (1, 0) circle[radius=0.07];
      \filldraw (0, 2) circle[radius=0.07];
    \end{scope}
  \end{tikzpicture}
\end{equation}

\noindent The unitary representations of $\fcal$ and $\tcal$ are such that depending on the number of trees in the numerator and the denominator of a fraction, we can get several copies of the semicontinuous limit Hilbert space. For instance, consider the two fractions:
\vspace*{0.2\baselineskip}
\begin{equation}
  \begin{tikzpicture}[scale=0.25,baseline=5mm]
    \begin{scope}[yscale=-1, yshift=-4cm]
      \draw (2.0, 0.0) -- (3.0, 1.0) -- (4.0, 0.0);
      \draw (0.0, 0.0) -- (2.0, 2.0) -- (3.0, 1.0);
      \draw (2, 2) -- (2, 3);
    \end{scope}

    \draw (-0.5, 0) -- (4.5, 0);

    \begin{scope}[yscale=-1, yshift=1cm]
      \draw (0.0, 0.0) -- (1.0, 1.0) -- (2.0, 0.0);
      \draw (1.0, 1.0) -- (2.0, 2.0) -- (4.0, 0.0);
      \draw (2, 2) -- (2, 3);
    \end{scope}
  \end{tikzpicture}
  \mspace{100mu}
  \begin{tikzpicture}[scale=0.25,baseline=5mm]
    \begin{scope}[yscale=-1, yshift=-7cm]
      \draw (2.0, 0.0) -- (3.0, 1.0) -- (4.0, 0.0);
      \draw (0.0, 0.0) -- (2.0, 2.0) -- (3.0, 1.0);
      \draw (6.0, 0.0) -- (7.0, 1.0) -- (8.0, 0.0);
      \draw (7.0, 1.0) -- (8.0, 2.0) -- (10.0, 0.0);
      \draw (2.0, 2.0) -- (5.0, 5.0) -- (8.0, 2.0);
      \draw (5, 5) -- (5, 6);
    \end{scope}

    \draw (-0.5, 0) -- (11, 0);

    \begin{scope}[yscale=-1, yshift=1cm]
      \draw (0.0, 0.0) -- (1.0, 1.0) -- (2.0, 0.0);
      \draw (1.0, 1.0) -- (2.0, 2.0) -- (4.0, 0.0);
      \draw (2, 2) -- (2, 3);
      \begin{scope}[xshift=6cm]
        \draw (2.0, 0.0) -- (3.0, 1.0) -- (4.0, 0.0);
        \draw (0.0, 0.0) -- (2.0, 2.0) -- (3.0, 1.0);
        \draw (2, 2) -- (2, 3);
      \end{scope}
    \end{scope}
  \end{tikzpicture}
\end{equation}
Under the representation, the left-hand fraction becomes a unitary $\hfrak\to\hfrak$, while the one on the right becomes a unitary $\hfrak\to\hfrak\otimes\hfrak$. If, for the moment, we read the above pictures of cobordisms from top to bottom, then the resemblance is striking: The right-hand element of Thompson's groupoid is a map from one spacetime to two copies of that spacetime, and the fraction seems to have an `underlying' cobordism. It remains to be seen if we can find more similarities between these mathematical structures that are physically plausible.

\defbibnote{myprenote}{\small The bibliography entries are sorted alphabetically by the surname of the first author. References to arXiv versions or other preprints are always included if available, even when there exists a published version elsewhere. For each entry, referencing pages are given in the form $\hookrightarrow$ \textelp{}.\newline}

\begingroup
\setlength{\emergencystretch}{1em}
\printbibliography[prenote=myprenote]
\endgroup

\end{document}